\documentclass[10pt]{article}
\usepackage[english]{babel}
\usepackage[latin1]{inputenc}
\usepackage{graphicx}
\usepackage{float}
\usepackage{graphics}
\usepackage{subfig}
\usepackage{array}
\usepackage{hyperref}
\usepackage{enumitem}
\usepackage{amsmath}
\usepackage{amssymb}
\usepackage{amsthm}
\usepackage{latexsym}
\usepackage{appendix}
\usepackage{color} 
\usepackage{enumerate}
\usepackage[T1]{fontenc}
\usepackage{todonotes}
\usepackage{geometry}

\definecolor{red}{rgb}{0.7,0.15,0.15}
\definecolor{green}{rgb}{0,0.5,0}
\definecolor{blue}{rgb}{0,0,0.7}
\hypersetup{colorlinks, linkcolor={red},citecolor={green}, urlcolor={blue}}
\newtheorem{Theorem}{Theorem}[part]
\newtheorem{Definition}{Definition}[part]
\newtheorem{Proposition}{Proposition}[part]

\newtheorem{Assumption}{Assumption}[part]
\newtheorem{Lemma}{Lemma}[part]
\newtheorem{Corollary}{Corollary}[part]
\newtheorem{Remark}{Remark}[part]

\makeatletter \@addtoreset{equation}{section}

\@addtoreset{Definition}{section}

\@addtoreset{Theorem}{section}

\@addtoreset{Proposition}{section}

\@addtoreset{Property}{section}

\@addtoreset{Assumption}{section}

\@addtoreset{Corollary}{section}

\@addtoreset{Lemma}{section}

\@addtoreset{Remark}{section}

\@addtoreset{Example}{section}

\@addtoreset{Condition}{section}

\def \E{\mathbb{E}}
\def \F{\mathbb{F}}

\def \P{\mathbb{P}}

\def \R{\mathbb{R}}

\def \G{\mathbb{G}}

\def\Ac{{\cal A}}

\def\Dc{{\cal D}}

\def\Fc{{\cal F}}
\def\Gc{{\cal G}}

\def\Vc{{\cal V}}

\newcommand{\ds}{\displaystyle}
\def\reff#1{{\rm(\ref{#1})}}

 \title{Bank monitoring incentives under moral hazard and adverse selection 
\footnote{The authors gratefully acknowledges the support of the French Ministry of Foreign Affairs and the Merlion programme.}}

 \author{Nicol\'as {\sc Hern\'andez Santib\'a\~nez} \footnote{University of Michigan, nihernan@umich.edu.} \and Dylan {\sc Possama\"{i}} \footnote{Columbia University, IEOR, dp2917@columbia.edu. This author gratefully acknowledges the support of the ANR project Pacman, ANR-16-CE05-0027.}
 \and Chao {\sc Zhou} \footnote{Department of Mathematics, National University of Singapore, Singapore, matzc@nus.edu.sg. Research supported by NUS Grant R-146-000-179-133, Singapore MOE AcRF Grants R-146-000-219-112 and R-146-000-255-114.}}
 
             \date{\today}

\begin{document}

 \maketitle

\begin{abstract}In this paper, we extend the optimal securitisation model of Pag\`es \cite{pages2012bank} and Possama\"i and Pag\`es \cite{pages2014mathematical} between an investor and a bank to a setting allowing both moral hazard and adverse selection. Following the recent approach to these problems of Cvitani\'c, Wan and Yang \cite{cvitanic2013dynamics}, we characterise explicitly and rigorously the so-called credible set of the continuation and temptation values of the bank, and obtain the value function of the investor as well as the optimal contracts through a recursive system of first-order variational inequalities with gradient constraints. We provide a detailed discussion of the properties of the optimal menu of contracts.

\vspace{3mm}

\vspace{3mm}

\vspace{3mm}


\end{abstract}

\section{Introduction}

Principal--Agent problems with moral hazard have an extremely rich history, dating back to the early static models of the 70s, see among many others Zeckhauser \cite{zeckhauser1970medical}, Spence and Zeckhauser \cite{spence1971insurance}, or Mirrlees \cite{mirrlees1972population,mirrlees1974notes,mirrlees1976optimal,mirrlees1999theory}, as well as the seminal papers by Grossman and Hart \cite{grossman1983analysis}, Jewitt, \cite{jewitt1988justifying}, Holmstr\"om \cite{holmstrom1979moral} or Rogerson \cite{rogerson1985first}. If moral hazard results from the inability of the Principal to monitor, or to contract upon, the actions of the Agent, there is a second fundamental feature of the Principal-Agent relationship which has been very frequently studied in the literature, namely that of adverse selection, corresponding to the inability to observe private information of the Agent, which is often referred to as his type. In this case, the Principal offers to the Agent a menu of contracts, each having been designed for a specific type. The so-called {\it revelation principle}, states then that it is always optimal for the Principal to propose menus for which it is optimal for the Agent to truthfully reveal his type. Pioneering research in the latter direction were due to Mirrlees \cite{mirrlees1971exploration}, Mussa and Rosen \cite{mussa1978monopoly}, Roberts \cite{roberts1979welfare}, Spence \cite{spence1980multi}, Baron and Myerson \cite{baron1982regulating}, Maskin and Riley \cite{maskin1984monopoly}, Guesnerie and Laffont \cite{guesnerie1984complete}, and later by Salani\'e \cite{salanie1990adverse}, Wilson \cite{wilson1993non}, or Rochet and Chon\'e \cite{rochet1998ironing}. However, despite the early realisation of the importance of considering models involving both these features at the same time, the literature on Principal-Agent problems involving both moral hazard and adverse selection has remained, in comparison, rather scarce. As far as we know, they were considered for the first time by Antle \cite{antle1980moral}, in the context of auditor contracts, and then, under the name of generalised Principal-Agent problems, by Myerson \cite{myerson1982optimal}\footnote{There were earlier attempts in this direction, but providing a less systematic treatment of the problem; see the income tax model of Mirrlees \cite{mirrlees1971exploration}, the Soviet incentive scheme study of Weitzman \cite{weitzman1976new}, or the papers by Baron and Holmstr\"om \cite{baron1980investment} and Baron \cite{baron1982model}.}. These generalised agency problems were then studied in a wide variety of economic settings, notably by Dionne and Lasserre \cite{dionne1985dealing}, Laffont and Tirole \cite{laffont1986using}, McAfee and McMillan \cite{mcafee1986bidding}, Picard \cite{picard1987design}, Baron and Besanko \cite{baron1987monitoring,baron1988monitoring}, Melumad and Reichelstein \cite{melumad1987centralization,melumad1989value}, Guesnerie, Picard and Rey \cite{guesnerie1989adverse}, Page \cite{page1991optimal}, Zou \cite{zou1992threat}, Caillaud, Guesnerie and Rey \cite{caillaud1992noisy}, Lewis and Sappington \cite{lewis1995optimal}, or Bhattacharyya \cite{bhhha1997div}\footnote{We refer the interested reader to the more recent works of Faynzilberg and Kumar \cite{faynzilberg2000generalized}, Theilen \cite{theilen2003simultaneous}, Jullien, Salani\'e and Salani\'e \cite{jullien2007screening}, Gottlieb and Moreira \cite{gottlieb2014simultaneous}.}.

\vspace{0.5em}
All the previously mentioned models are either in static or discrete--time settings. The first study of the continuous time problem with moral hazard and adverse selection was made by Sung \cite{sung2005optimal}, in which the author extends the seminal finite horizon and continuous--time model of Holmstr\"om and Milgrom \cite{holmstrom1987aggregation}. A more recent work, to which our paper is mostly related has been treated by Cvitani\'c, Wan and Yang \cite{cvitanic2013dynamics}, where the authors extend the famous infinite horizon model of Sannikov \cite{sannikov2008continuous} to the adverse selection setting. If one of the main contributions of Sannikov \cite{sannikov2008continuous} was to have identified that the {\it continuation value} of the Agent was a fundamental state variable for the problem of the Principal, \cite{cvitanic2013dynamics} shows that in a context with both moral hazard and adverse selection, the Principal has also to keep track of the so-called {\it temptation value}, that is to say the continuation utility of the Agent who would not reveal his true type. Although close to the latter paper, our work is foremost an extension of the bank incentives model of Pag\`es and Possama\"i \cite{pages2014mathematical}, which studies the contracting problem between competitive investors and an impatient bank who monitors a pool of long term loans subject to Markovian contagion (we also refer the reader to the companion paper by Pag\`es \cite{pages2012bank} for the economic intuitions and interpretations of the model). 

\vspace{0.5em}
The home loan crash of 2008 has strongly highlighted the inherent weaknesses of the securitisation agreements created during the 2000s, and was at the heart of the decision from the US government to impose tight deadlines for the adoption of new and tighter regulations for credit risk retention. Among these, one of particular interest for us is the Dodd--Frank Act, which prescribes that sponsors retain at least five percent of the credit risk in most securitisation transactions. The purpose of \cite{pages2012bank,pages2014mathematical} was to study optimal securitisation when the sponsor remains involved with its retail originations, and can engage in unobservable actions that result in private benefits at the expense of performance. The assumption that the bank itself can have impact on the default rate of the pool over time is a metaphor for the distinction between its exogenous base quality and the endogenous default probability obtained after monitoring. Moral hazard then emerges because the bank has more "skin in the game" than the investors, and has the opportunity, {\it ex ante} and {\it ex post}, to exercise a (costly) monitoring of the non--defaulted loans. This is a stylised way to sum up all the actions that the bank can enter into to ensure itself of the solvability of the borrowers. There is much that the bank can do to improve performance over the life of a transaction. First, a strong quality control process helps lenders exercise due diligence in evaluating borrowers' current income, and keep track of those who might be getting closer to default. This is a surveillance action which has to be undertaken continuously, and not only prior to the inception of the contract. Second, the bank can efficiently assist troubled borrowers by acting early and firmly, before mortgages become seriously delinquent. The selection of bank employees in charge of these actions is also usually assumed to potentially affect loss severity by as much as 30\%. For instance, Agarwal et al. \cite{agarwal2014inconsistent} have put into light important and systematic changes in the default rates of state--chartered banks' real estate loans, when a so--called "rotation" policy between federal and state supervisors at predetermined time periods is put into place. This clearly shows  that banks are perfectly able to enter into corrective actions in the event of delinquencies, when they have incentives to do so.

\vspace{0.5em}
The findings of \cite{pages2012bank,pages2014mathematical} were that since the investors cannot observe the monitoring effort of the bank, they proposed CDS type contracts offering remuneration to the bank, and giving it incentives through postponement of payments and threat of stochastic liquidation of the contract (similarly to the seminal paper of Biais, Mariotti, Rochet and Villeneuve \cite{biais2010large}). In the present paper, we assume furthermore that there are two types of banks, which we call good and bad, co--existing in the market, differing by their efficiency in using their remuneration (or equivalently differing by their monitoring costs). Even if the investor is supposed to know the distribution of the type of banks, that is to say the probability with which the bank he is currently discussing with is good or bad, he cannot know for sure what her type is. Again, this is a stylised way to express the fact that "skin in the game" might significantly vary from one bank to the other. The fact that we consider only two types is mainly for simplicity and tractability, and because fully multidimensional screening problems are already extremely hard to solve in static one--period models, and except for specific models (see Section \ref{sec:prelim} for details), nothing more than existence of an optimal contract can be hoped for, see for instance Carlier \cite{carlier2001general}.

\vspace{0.5em}
Mathematically speaking, we follow the general dynamic programming approach of Cvitani\'c, Possama\"i and Touzi \cite{cvitanic2015dynamic}, as well as that on adverse selection problems initiated by \cite{cvitanic2013dynamics}. Intuitively, these approaches require first, using martingale (or more precisely backward SDEs) arguments, to solve the (non-Markovian) optimal control problem faced by the two types of banks when choosing contracts. This requires obviously, using the terminology introduced above, to keep track of both the {\it continuation value} and the {\it temptation value} of the banks, when they choose the contract designed for them or not. The problem of the Principal rewrites then as two standard stochastic control problems, one in which he hires the good bank, and one in which he hires the bad one. Each of these problems uses in turn the aforementioned two state variables (and these two only, because the horizon is infinite and the Principal is risk-neutral), with truth-telling constraint, asserting that the continuation value should always be greater than the temptation value. This leads to optimal control problems with state constraints, and thus to Hamilton-Jacobi-Bellman (HJB for short) equations (or more precisely variational inequalities with gradient constraints, since our problem is actually a singular stochastic control problem) in a domain, which, following \cite{cvitanic2013dynamics}, we call the credible set. This set is defined as the set containing the pair of value functions of the good and bad bank under every admissible contract offered by the investor. The determination of this set is the first fundamental step in our approach. Following the the orignal ideas of  \cite{cvitanic2013dynamics}, we prove that the determination of the boundaries of this set can be achieved by solving two so-called double-sided moral hazard problems, in which one of the type of banks is actually hiring the other one. Fortunately for us, it turned out to be possible to obtain rigorously\footnote{Notice that in this respect the study in \cite{cvitanic2013dynamics} was more formal, and our paper provides, as far as we know, the first rigorous derivation of this credible set.} explicit expressions for these boundaries by solving the associated system of HJB equations and using verification type arguments. We also would like to emphasise that unlike in \cite{cvitanic2013dynamics}, there is certain dynamic component in our model, since we have to keep track of the number of non-defaulted loans, through a time inhomogeneous Poisson process. This leads to a dynamic credible set, as well as, in the end, to a {\it recursive system} of HJB equations characterising the value function of the Principal.

\vspace{0.5em}
After having determined the credible set itself, we pursue our study by concentrating on two specific forms of contracts: the shutdown contract in which the investor designs a contract which will be accepted only by the good bank, and the more classical screening contract, corresponding to a menu of contracts, one for each type of bank, which provides incentives to reveal her true type and choose the contract designed for her. These two contracts correspond simply to the offering, over the correct domain of expected utilities of the banks (so as to satisfy the proper truth--telling and participation constraints), of the best contracts that the investor can design independently for hiring the good and the bad bank. 

\vspace{0.5em}
Since we characterise, under classical verification type arguments, the value function of the investor through a system of HJB equations, we also have classically access to the optimal contracts through this value function and its derivatives. This allows us to provide an associated qualitative and quantitative analysis. It turns out that the optimal contracts designed for the good and the bad bank share the same attributes, and are close in spirit to the ones derived in the pure moral hazard case in \cite{pages2014mathematical}. On the boundaries of the credible set, the value function of the bad bank plays the role of a state process. The payments of the optimal contracts are postponed until the moment the state process reaches a sufficiently high level, depending on the current size of the project. Similarly, when one of the loans in  the pool defaults, the project is liquidated with a probability that decreases with the value of the state process. If the value function of the bad bank at the default time is below some critical level, the project will be liquidated for sure under the optimal contracts. On the other side, if the value function of the bad bank is high enough at the default time, the project will be maintained. In the interior of the credible set, the continuation value and the temptation value of the banks are the state processes for the optimal contracts. It is possible to identify zones of \emph{good performance} inside of the credible set, where the agents are remunerated and the project is maintained in case a default occurs. It is also possible to identify zones of \emph{bad performance}, where the agents are not paid and the project is liquidated in case of default. In the rest of the credible set the optimal contracts provide intermediary situations.  

\vspace{0.5em}
The rest of the paper is organised as follows. In Section \ref{Model sec}, we present the model, we define the set of admissible contracts and we state the investor's problem. In Section \ref{sec:pure}, we recall the results obtained in \cite{pages2014mathematical} for the case of pure moral hazard, which will be useful later on for us. In Section \ref{sec:credi}, we formally study the credible set and obtain an explicit expression for it. In Section \ref{sec:opti}, we study both the optimal shutdown and screening contract, describing their characteristics and the behaviour of the banks when they accept these contracts. The Appendix contains all the technical proofs of the paper.

\vspace{0.5em}
{\bf Notations:} Let $\mathbb N$ denote the set of non--negative integers. For any $n\in\mathbb N\backslash \{0\}$, we identify $\mathbb R^n$ with the set of $n-$dimensional column vectors. The associated inner product between two elements $(x,y)\in\mathbb R^n\times\mathbb R^n$ will be denoted by $x\cdot y$. For simplicity of notations, we will sometimes write column vectors in a row form, with the usual transposition operator $\top$, that is to say $(x_1,\dots,x_n)^\top\in\mathbb R^n$ for some $x_i\in\mathbb R$, $1\leq i\leq n$. Let $\mathbb R_+$ denote the set of non--negative real numbers, and $\mathcal B(\mathbb R_+)$ the associated Borel $\sigma-$algebra. For any fixed non--negative measure $\nu$ on $(\mathbb R_+,\mathcal B(\R_+))$, the Lebesgue--Stieljes integral of a measurable map $f:\mathbb R_+\longrightarrow \mathbb R$ will be denoted indifferently 
\[
\int_{[u,t]}f(s)\mathrm{d}\nu_s\text{ or } \int_u^tf(s)\mathrm{d}\nu_s,\ 0\leq u\leq t.
\]

\section{The model\label{Model sec}}
This section is dedicated to the description of the model we are going to study, presenting the contracts as well as the criterion of both the Principal and the Agent. As recalled in the Introduction, it is actually an adverse selection extension of the model introduced first by Pag\`es in \cite{pages2012bank} and studied in depth by Pag\`es and Possama\"i \cite{pages2014mathematical}.

\subsection{Preliminaries}\label{sec:prelim}
We consider a model in continuous time, indexed by $t\in \left[ 0,\infty \right) $. Without loss of generality and for simplicity, the risk--free interest rate is taken to be $0$.\footnote{As already pointed out in the seminal paper of Biais, Mariotti, Rochet and Villeneuve \cite{biais2010large}, see also \cite{pages2014mathematical}, the only quantity of interest here is the difference between the discounting factors of the Principal and the Agent.} Our first player will be a bank (the Agent, referred to as "she"), who has access to a pool of $I$ unit loans indexed by $j=1$, $\dots $,$~I$ which are {\it ex ante} identical. Each loan is a perpetuity yielding
cash flow $\mu $ per unit time until it defaults. Once a loan defaults, it
gives no further payments. As is commonplace in the Principal--Agent literature, especially since the paper of Sannikov \cite{sannikov2008continuous}, the infinite maturity assumption is here for simplicity and tractability, since it makes the problem stationary, in the sense that the value function of the Principal will not be time--dependent. We assume that the banks in the market are different, and that two types of banks coexist, each one being characterised by a parameter taking values in the set $\mathfrak R:=\{\rho_g,\rho_b\}$ with $\rho_g>\rho_b$. We call the bank good (respectively bad) if its type is $\rho_g$ (respectively $\rho_b$). Furthermore, it is considered to be common knowledge that the proportion of the banks of type $\rho_i$, $i\in\{g,b\}$, is $p_i\in(0,1)$. 

\vspace{0.5em} Denote by 
\begin{equation*}
N_{t}:=\sum_{j=1}^{I}{\bf 1}_{\left\{ \tau ^{j}\leq t\right\} },
\end{equation*}
the sum of individual loan default indicators, where $\tau ^{j}$ is the
default time of loan $j$. The current size of the pool is, at some time $t\geq 0$, $I-N_{t}$. Since
all loans are {\it a priori} identical, they can be reindexed in any order after
defaults. The action of the banks consists in deciding at each time $t\geq 0$
whether they monitor any of the loans which have not defaulted yet. These actions are
summarised by the functions $e_{t}^{j,i}$, where for $1\leq j\leq I-N_{t},\ i\in\{g,b\}$, $e_{t}^{j,i}=1$ if loan $j$ is
monitored at time $t$ by the bank of type $\rho_i$, and $e_{t}^{j,i}=0$ otherwise. Non-monitoring renders a private benefit $B>0$ per loan and per unit time to
the bank, regardless of its type. The opportunity cost of monitoring is thus proportional to the
number of monitored loans. Once more, more general cost structures could be considered, but this choice has been made for the sake of simplicity.

\vspace{0.5em} The rate at which loan $j$ defaults is controlled by the
hazard rate $\alpha _{t}^{j}$ specifying its instantaneous probability of
default conditional on history up to time $t$. Individual hazard rates are
assumed to depend on the monitoring choice of the bank and on the size
of the pool. In particular, this allows to incorporate a type of contagion effect in the model. Specifically, we choose to model the hazard rate of a
non--defaulted loan $j$ at time $t$, when it is monitored (or not) by a bank of type $\rho_i$ as
\begin{equation}
\alpha _{t}^{j,i}:=\alpha _{I-N_{t}}\big( 1+\big(1-e_{t}^{j,i}\big)\varepsilon \big) ,\ t\geq 0,\ j=1,\dots,I-N_t,\ i\in\{b,g\},
\label{hazard eq}
\end{equation}
where the parameters $\left\{ \alpha _{j}\right\} _{1\leq j\leq I}$ are positive constants representing individual \textquotedblleft baseline\textquotedblright\ risk under monitoring when the number of loans is $j$, and $\varepsilon >0$ is the proportional impact of shirking on default risk. We assume that the impact of shirking is independent of the type of the bank. There are two main reasons for this choice. First of all, it is well--known that as soon as the dimension of the type is greater or equal to $2$, we enter into the field of multidimensional screening, which, already for static one period models is notoriously hard to analyse, and deriving meaningful economic interpretations is most often elusive (see the seminal paper of Rochet and Chon\'e \cite{rochet1998ironing} or the more recent contribution of Figalli et al. \cite{figalli2011multidimensional} for more details). Notwithstanding this difficulty, we also found out that differentiating between the banks in this regard created degeneracy in the model. We refer the reader to Section \ref{sec:extension} in the Appendix for a more detailed explanation.

\vspace{0.5em} For $i\in\{b,g\}$, we define the shirking process $k^i$ as the number of loans that the bank of type $\rho_i$ fails to monitor at time $t\geq 0$. Then, according to \textrm{(\ref{hazard eq})}, the corresponding aggregate default
intensity is given by 
\begin{equation}
\lambda _{t}^{k^i}:=\sum_{j=1}^{I-N_t}\alpha_t^{j,i}=\alpha _{I-N_{t}}\big( I-N_{t}+\varepsilon k^i_{t}\big) .
\label{aggregate eq}
\end{equation}
The banks can fund the pool internally at a cost $r\geq 0$. They
can also raise funds from a competitive investor (the Principal, referred to as "he") who values income streams
at the prevailing risk--less interest rate of zero. We assume that both the
banks and the investor observe the history of defaults and liquidations, as well as the parameters $p_b$ and $p_g$, but the monitoring choices and the type of the bank are unobservable for the investor.

\vspace{0.5em}

\subsection{Description of the contracts}

Before going on, let us now describe the stochastic basis on
which we will be working. We will always place ourselves on a probability space $
(\Omega ,\mathcal{F}, \mathbb{P})$ on which $N$ is a point process with intensity $\lambda _{t}^{0}$, which is defined by \textrm{(\ref{aggregate eq})} when all loans are monitored at all times, that is

\[
\lambda _{t}^{0}=\alpha _{I-N_{t}}\big( I-N_{t}\big).
\]
We denote by $\F:=(\mathcal{F}_{t}^{N})_{t\geq 0}$ the $\P-$completion of the natural filtration of $N$. We call $\tau$ the liquidation time of the whole pool and let $H_{t}:={\bf 1}_{\left\{ t\geq \tau \right\} }$ be the
liquidation indicator of the pool. We denote by $\G:=(\mathcal{G}_{t})_{t\geq 0}$ the minimal filtration containing $\F$ and that makes $\tau $ a $\G-$stopping time. We note that this filtration
satisfies the usual hypotheses of completeness and right--continuity.

\vspace{0.5em}Contracts are offered by the
investor to the bank and agreed upon at time $0$. As usual in contracting theory, the bank can accept or refuse the contract, but once accepted, both the bank and the investor are fully committed to the contract. 
More precisely, the investor offers a menu of contracts $\Psi_i:=(k^i, \theta^i,D^i)$, $i\in\{g,b\}$ specifying on the one hand a desired level of monitoring $k^i$ for the bank of type $\rho_i$, which is a $\G-$predictable process such that for any $t\geq 0$, $k^i_t$ takes values in $\{0,\dots,I-N_t\}$ (this set is denoted by $\mathfrak K$), as well as a flow of payment $D^i$. These payments belong to set $\mathcal D$ of processes which are c\`adl\`ag, non--decreasing, non--negative, $\G-$predictable and such that there exists some $\beta>0$ 
\begin{equation}\label{eq:integD}
\mathbb E^{\P}\bigg[\mathrm{e}^{\beta\tau}\bigg(\int_0^\tau\mathrm{e}^{-rs}\mathrm{d}D^i_s\bigg)^2\bigg]<+\infty.
\end{equation}
We do not rule out the possibility of immediate lump--sum payments at the initialisation of the contract, and therefore the processes in $\mathcal D$ are assumed to satisfy $D_{0^-}=0.$ Hence, if $D_0\neq 0$, it means that a lump--sum payment has indeed been made. Notice also that since the intensities of $N$ and $H$ under $\P$ are bounded, we know that $\tau$ has at least some exponential moments under $\P$, meaning that any bounded payment belongs to $\Dc$.

\vspace{0.5em}
The contract also specifies when liquidation occurs. We
assume that liquidations can only take the form of the stochastic
liquidation of all loans following immediately default. Hence, the contract specifies the
probability $\theta^i _{t}$, which belongs to the set $\Theta$ of $[0,1]-$valued, $\G-$predictable processes, with which the pool is maintained given default ($
\mathrm{d}N_{t}=1$), so that at each point in time, if the bank has indeed chosen the contract $\Psi_i$ 
\begin{equation*}
\mathrm{d}H_{t}=
\begin{cases}
0 \ \text{with probability $\theta^i _{t},$} \\ 
\mathrm{d}N_{t} \ \text{with probability $1-\theta^i _{t}.$}
\end{cases}
\end{equation*}
With our notations, given a contract $\Psi_i$, the hazard rates associated with the default and
liquidation processes $N_{t}$ and $H_{t}$ are, if the bank does choose the contract $\Psi_i$, $\lambda _{t}^{k^i}$ and $\left(1-\theta^i_{t}\right) \lambda _{t}^{k^i}$, respectively. The above properties translate into 
\begin{equation*}
\mathbb{P}\big[\tau\in\left\{\tau^1,...,\tau^I\right\}\big]=1,\text{ and }
\mathbb{P}\big[\tau=\tau^j|\mathcal{F}_{\tau^j},\tau>\tau^{j-1}\big]=1-\theta^i_{
\tau^j},\ j\in\left\{1,\dots,I\right\}.
\end{equation*}

For ease of notations, a contract $\Psi:=(k,\theta,D)$ will be said to be admissible if $(k,\theta, D)\in\mathfrak K\times\Theta\times\mathcal D$.
%
%
%
%
%
As is commonplace in the Principal--Agent literature, we assume that the monitoring choices of the banks affect only the
distribution of the size of the pool. To formalise this, recall that, by
definition, any shirking process $k\in\mathfrak K$ is $\G-$predictable and
bounded. Then, by Girsanov's theorem, we can define a probability measure $\mathbb{P}^{k}$ on $(\Omega,\Fc)$, equivalent to $\mathbb{P}$, such that $
N_{t}-\int_{0}^{t}\lambda _{t}^{k}\mathrm{d}s$
is a $\mathbb{P}^{k}-$martingale. More precisely, we have on $\mathcal{G}_{t}$ 
\begin{equation*}
\frac{\mathrm{d}\mathbb{P}^{k}}{\mathrm{d}\mathbb{P}}=Z_{t}^{k},
\end{equation*}
where $Z^{k}$ is the unique solution of the following SDE 
\begin{equation*}
Z_{t}^{k}=1+\int_{0}^{t}Z_{s^{-}}^{k}\bigg( \frac{\lambda _{s}^{k}}{\lambda
_{s}^{0}}-1\bigg) \big( \mathrm{d}N_{s}-\lambda _{s}^{0}\mathrm{d}s\big) ,\text{ }0\leq
t\leq \tau,\text{ }\mathbb{P}-\mathrm{a.s.}
\end{equation*}
Then, if the bank of type $\rho_i$ chooses the contract $\Psi_i$, her utility at $t=0$, if she follows the recommendation $k^i$, is given by 
\begin{equation}
u_{0}^i(k^i,\theta^i,D^i):=\mathbb{E}^{\mathbb{P}^{k^i}}\bigg[ \int_{0}^{\tau
}e^{-rs}\big(\rho_i\mathrm{d}D^i_s+Bk^i_s\mathrm{d}s\big)\bigg] ,  \label{bank utility eq}
\end{equation}
while that of the investor is 
\begin{equation}
v_{0}\big((\Psi_i)_{i\in\{g,b\}} \big):=\sum_{i\in\{g,b\}}p_i\mathbb{E}^{\mathbb{P}^{k^i}}\bigg[ \int_{0}^{\tau
}\big( I-N_s\big) \mu \mathrm{d}s-\mathrm{d}D^i_s\bigg] .  \label{investors utility eq}
\end{equation}
The parameter $\rho_i$ actually discriminates between the two types of banks through the way they derive utility from the cash--flows delivered by the investor. Hence, for a same level of salary, the good bank will get more utility than a bad bank. Such a form of adverse selection is also considered in the paper of Cvitani\'c, Wan and Yang \cite{cvitanic2013dynamics}. Notice that de dependence of the value functions of both the bank and the investor depend on the contract through the process $D^i$, but also through the stopping time $\tau$, whose distribution depends both on $\theta^i$ and the effort choice $k^i$ of the bank.

\subsection{Formulation of the investor's problem}\label{sec-ppalproblem}
We assume that the bank of type $\rho_i$ has an outside opportunity to the contract which provides her reservation utility $R_0^i$. The investor's problem is to offer a menu of admissible contracts $(\Psi_i)_{i\in\{g,b\}}:=(k^i,\theta^i,D^i)_{i\in\{g,b\}}$ which maximises his utility \reff{investors utility eq}, subject to the three following constraints
\begin{align}
\label{const-reser}
&u_{0}^i(k^i,\theta^i,D^i)\geq R_0^i,\ i\in\{g,b\},\\[0.5em]
\label{const-ic}
 &u_{0}^i(k^i,\theta^i, D^i)=\underset{k\in\mathfrak K}{\sup}\ u_{0}^i(k,\theta^i,D^i),\ i\in\{g,b\},\\[0.5em]
 \label{const-screen}
 &u_{0}^i(k^i,\theta^i,D^i)\geq \underset{k\in\mathfrak K}{\sup}\ u_{0}^i(k,\theta^j,D^j), \ i\neq j,\ (i,j)\in\{g,b\}^2.
\end{align}
Condition \reff{const-reser} is the usual participation constraint for the banks. Condition \reff{const-ic} is the so--called incentive compatibility condition, stating that given $(\theta^i,D^i)$, the recommended effort $k^i$ is an optimal monitoring choice for the bank of type $\rho_i$ . Finally, Condition \reff{const-screen} means that if a bank adversely selects a contract, she cannot get more utility than if she had truthfully revealed her type at time $0$. Following the literature, we call such a contract a {\it screening} contract.

\vspace{0.5em}
In the sequel, we will start by deriving the optimal contract in the pure moral hazard case, then we will look into the so--called optimal shutdown contract, for which the investor deliberately excludes the bad bank, before finally investigating the optimal screening contract. We will invoke some results from \cite{pages2012bank} in this paper, for this reason we will require later the assumptions of their main result, Theorem 3.15, which are the following.

\begin{Assumption}\label{assump} ~
Let $\overline\alpha_I$ be the harmonic mean of the $(\alpha_j)_{1\leq j\leq I}$,

\vspace{0.5em}
$(i)$ $\mu \geq \overline{\alpha }_{I}.$

\vspace{0.5em}
$(ii)$ We have for all $j\leq I$, $rB(1+\varepsilon)\leq (\mu \varepsilon -B)
\varepsilon \overline{\alpha }_{j}.$

\vspace{0.5em}
$(iii)$ Individual default risk is non--decreasing with
past default, $\alpha _{j}\leq \alpha _{j-1},\ \text{for all }j\leq I.$
\end{Assumption}

\subsection{Comments on the modelling choices and assumptions}

Let us start by discussing Assumption \ref{assump}. Concerning $(i)$, under monitoring, the expected duration until the next default in a pool of $j$ loans is $1/\lambda_j$. Hence, the average revenue from the pool over that period will be given by $\mu/\lambda_j$, of which $1/I$ is ascribed to the original loan. The payoff of a loan corresponds then to summing this quantity over $j$, and the obtained result must be above the initial unit cost for the loan to be worth anything at all under monitoring. Assumption \ref{assump}$(ii)$ imposes an upper bound on the bank's discount rate, and basically states that it should not be so large that the cost of the rent extracted by a monitoring bank outweighs the pecuniary gains stemming from the use of the monitoring technology. Finally Assumption \ref{assump}$(iii)$ simply models a contagion effect, translating the fact that past defaults impact positively the likelihood of a further default to happen. 
 
 \vspace{0.5em}
A second important point in the model is the liquidation policy of the contract. Even though liquidations are inefficient in the first--best situation without moral hazard nor adverse selection (see \cite{pages2012bank}), they are necessary in the second--best in order to restore incentives to monitor when performance is poor. However, liquidation can take many forms, and for instance liquidating all loans with state--dependent probability is not necessarily better than partially liquidating the pool with fixed probability. Another option would be to downsize the pool by a potentially larger number of loans, possibly state--dependent as well, whenever a default occurs. Given that in practice liquidations are rarely decided in such a random fashion, it is of the utmost importance to verify that such liquidation policies {\color{black}cannot improve on} social welfare. In the case of pure moral hazard, \cite[Proposition 6]{pages2012bank} has shown that stochastic liquidation was optimal among all policies under an assumption which is met if changes in default intensities for the loans are gradual. Since the capital structure of subprime mortgage--backed securities is typically split up into a large number of tranches, where it is then reasonable to assume that default intensities are constant, such an assumption will be verified in the real--world applications of the model.

 \vspace{0.5em}
Another very important assumption here is the fact that we consider a "full commitment" dynamic contracting problem between the investor and the bank. In other words, both parties are fully committed to the long--term dynamic contract at the onset of the relationship. However, one of the central features of banks, namely the fragility of their capital structures, stems from the fact that there is usually a limited commitment in the relationship between clients and banks, since, as highlighted by the seminal paper of Diamond and Dybvig \cite{diamond1983bank}, bank cleints can withdraw funds from banks at any time. We are perfectly conscient of this fact and have chosen to postpone the discussion of how to integrate non--commitment in our model to Section \ref{sec:} below, since the mathematics behind are similar.

 \vspace{0.5em}
 In our model, adverse selection stems only from the monitoring costs of the two types of  banks. A possible extension of our model could rely on a further differentiation between the work of the banks, i.e. when both good bank and bad bank work, the good one would be more efficient in the sense that the associated default intensity is
strictly smaller than that of the bad bank. This would also be a possible way to model the fact that the pool of loans of each bank could have different qualities. However, as will be explained in Appendix \ref{sec:extension}, such a feature would actually make the problem degenerate, in the sense that the upper boundary of the credible set that will be defined in Section \ref{sec:credi} becomes infinite. This would be a rather undesirable feature of the model and would create unwanted discontinuities. We give potential solutions to extend the model in this direction in Appendix \ref{sec:extension}, but leave the exact study to future research.
\section{The pure moral hazard case}\label{sec:pure}
In this section, we assume that the type of the bank is publicly known and is fixed to be some $\rho_i$, $i\in\{g,b\}$, which makes the problem exactly similar to the one considered in \cite{pages2014mathematical} (up to the modification of some constants). In particular, the investor only offers one contract. 
The results we obtain here, in particular the dynamics of the continuation utilities of the banks, will be crucial to the study of the shutdown and screening contracts later on. Therefore, they will be used throughout the paper without further references.

\vspace{0.5em}
In this setting, the utility of the investor, when he offers a contract $(k^i,\theta^i,D^i)\in\mathfrak K\times\Theta\times\mathcal D$ is given by
\begin{equation}
v_{0}^{{\rm pm}}(k^i,\theta^i,D^i ):=\mathbb{E}^{\mathbb{P}^{k^i}}\bigg[ \int_{0}^{\tau
}\big( I-N_s\big) \mu \mathrm{d}s-\mathrm{d}D^i_s\bigg],  \label{investors utility eq2}
\end{equation}
for which we define the following dynamic version for any $t\geq 0$
\[
v_{t}^{{\rm pm}}(k^i,\theta^i,D^i ):=\mathbb{E}^{\mathbb{P}^{k^i}}\bigg[\int_{t\wedge\tau}^{\tau
}\big( I-N_{s}\big) \mu \mathrm{d}s-\mathrm{d}D^i_{s}\bigg|\mathcal G_t\bigg].
\]

\subsection{The bank's problem}
As usual, the so--called continuation value of the bank (that is to say her future expected payoff) when offered $(\theta^i,D^i)\in\Theta\times\mathcal D$ plays a central role in the analysis. It is defined, for any $(t,k)\in\R_+\times\mathfrak K$ by
\[
u_{t}^i(k,\theta^i,D^i):=\mathbb E^{\mathbb P^k}\bigg[\int_{t\wedge\tau}^\tau e^{-r(s-t)}\big(\rho_i\mathrm{d}D^i_s+k_sB\mathrm{d}s\big)\bigg|\mathcal G_t\bigg].
\]

We also define the value function of the bank for any $t\geq 0$
\[
U_t^i(\theta^i,D^i):=\underset{k\in\mathfrak K}{\rm ess \ sup}\ u_t^i(k,\theta^i,D^i).
\]
Departing slightly from the usual approach in the literature, initiated notably by Sannikov \cite{sannikov2008continuous, sannikov2012contracts}, we reinterpret the problem of the bank in terms of BSDEs, which, we believe, offers an alternative approach which may be easier to apprehend for the mathematical finance community. Of course, such an interpretation of optimal stochastic control problem with control on the drift is far from being original, and we refer the interested reader to the seminal papers of Hamad\`ene and Lepeltier \cite{hamadene1995backward} and El Karoui and Quenez \cite{el1995dynamic} for more information, as well as to the recent articles by Cvitani\'c, Possama\"i and Touzi \cite{cvitanic2014moral,cvitanic2015dynamic} for more references and a systematic treatment of Principal--Agent type problems with this backward SDE approach. Before stating the related result, let us denote by $(Y^i,Z^i)$ the unique (super--)solution (existence and uniqueness will be justified below) to the following BSDE
\begin{equation}\label{bsde}
Y^i_t=0-\int_t^\tau g^i(s,Y^i_s,Z^i_s)\mathrm{d}s+ \int_t^\tau Z^i_s\cdot \mathrm{d}\widetilde M^i_s+\int_t^\tau \mathrm{d}K^i_s,\ 0\leq t\leq \tau,\ \mathbb P-\mathrm{a.s.},
\end{equation}
where
\begin{align*}
& M_t:=\begin{pmatrix}N_t \\ H_t\end{pmatrix}, \ \widetilde M_t^i:=M_t-\int_0^t\lambda^0_s\begin{pmatrix}1\\ 1-\theta^i_s\end{pmatrix} \mathrm{d}s,\ K_t^i:=\rho_iD^i_t, \;  f^i(t,k,y,z):=ry - Bk+k\alpha_{I-N_t}\varepsilon z\cdot \begin{pmatrix}1\\1-\theta^i_t\end{pmatrix}, \\
& g^i(t,y,z):=\underset{k\in\{0,\dots,I-N_t\}}{\inf}\ f^i(t,k,y,z)=ry-(I-N_t)\bigg(\alpha_{I-N_t}\varepsilon z\cdot\begin{pmatrix}1\\ 1-\theta^i_t\end{pmatrix}-B\bigg)^-.
\end{align*}
Uniqueness holds among processes $Y^i$ and $Z^i$ which are respectively $\G-$progressively measurable and continuous, and $\G-$predictable, and satisfy
\begin{equation}\label{eq:integ}
\E^\P\bigg[\sup_{0\leq t\leq \tau}\mathrm{e}^{(\beta\varepsilon^2-2r)t}\big|Y_t^i\big|^2\bigg]<+\infty,\; \text{\rm and}\; \E^\P\bigg[\int_0^\tau\mathrm{e}^{(\beta\varepsilon^2-2r)s}\|Z_s^i\|^2\mathrm{d}s\bigg]<+\infty,
\end{equation}
where $\beta$ is the exponent associated to $D^i$ by \eqref{eq:integD}. We have the following proposition, which is basically a reformulation of \cite[Proposition 3.2]{pages2014mathematical}. The proof is postponed to Appendix \ref{app:1}
\begin{Proposition}\label{prop:rep}
For any $(\theta^i,D^i)\in\Theta\times\mathcal D$, the value function of the bank has the dynamics, for $t\in[0,\tau]$, $\mathbb P-\mathrm{a.s.}$
\begin{align*}\label{dyn-U}
 \mathrm{d}U_t^i(\theta^i,D^i)=& \bigg(rU_t^i(\theta^i,D^i)-Bk^{\star,i}_t+\lambda_t^{k^{\star,i}}Z^i_t\cdot\begin{pmatrix}1\\1-\theta^i_t\end{pmatrix}\bigg)\mathrm{d}t-\rho_i\mathrm{d}D^i_t-Z^i_t\cdot \mathrm{d}\widetilde{M}^i_t,
\end{align*}
where $Z^i$ is the second component of the solution to the {\rm BSDE} \reff{bsde}. In particular, the optimal monitoring choice of the bank is given by
\[
k^{\star,i}_t=(I-N_t){\bf 1}_{\{Z^i_t\cdot(1,1-\theta^i_t)^\top<b_t\}},\text{ \rm with }b_t:=\frac{B}{\alpha_{I-N_t}\varepsilon},\ t\geq 0.
\]
\end{Proposition}
Notice that the above result implies that the monitoring choices of the bank are necessarily of bang--bang type, in the sense that she either monitors all the remaining loans, or none at all, which in turn implies that the investor can never give the bank incentives to monitor only a fraction of the loans at a given time\footnote{We assume here, as is commonplace in the Principal--Agent literature, that in the case where the bank is indifferent with respect to her monitoring decision, that is when $Z^i_t\cdot(1,1-\theta^i_t)^\top=b_t$, she acts in the best interest of the investors, and thus monitors all the $I-N_t$ remaining loans.}.

\subsection{Introducing feasible sets}
Following the terminology of Cvitani\'c, Wan and Yang \cite{cvitanic2013dynamics}, let us discuss the so--called feasible set for the banks.
\begin{Definition}
We call $\mathcal V_t^i$ the feasible set for the expected payoff of banks of type $\rho_i$, starting from some time $t\geq 0$, that is to say all the possible utilities that a bank of type $\rho_i$ can get from all the admissible contracts offered by the investor from time $t$ on.
\end{Definition}

Our next result gives an explicit form of the the feasible set $\mathcal V_t^i$, which turns out to be independent of the type of the bank. The proof is relegated to Appendix \ref{app:1}, and requires the introduction of $k^{\rm SH}$, the strategy of a bank which does not monitor any loan at any time, {\it i.e.} $k_s^{\rm SH}:=I-N_s$ for every $s\geq 0$
\begin{Lemma} \label{lemma:feasible set}
For $i\in\left\{g,b\right\}$ and for any $t\geq 0$, we have that $\mathcal V_t^i = \mathcal V_t$, with
\[
\mathcal V_t := \left[\frac{B(I-N_t)}{r+\lambda^{k^{\rm SH}}_t},+\infty\right).
\]
\end{Lemma}

\section{Credible set}\label{sec:credi}

In this section we come back to the case in which there are two types of banks in the market, and study the so--called credible set, which is formed by the pairs of value functions of the banks under the admissible contracts. As in \cite{cvitanic2013dynamics}, we do not expect all the points in the feasible set to correspond to a pair of reachable values of the banks under some admissible contract. We will therefore follow the approach initiated by \cite{cvitanic2013dynamics} and we will characterise the credible set. We emphasise an important difference with \cite{cvitanic2013dynamics} though, in the sense that in our context, the credible set becomes dynamic as it depends on the current size of the pool. In this section we work with generic contracts $(\theta,D)\in\Theta\times\Dc$, not necessarily designed for a particular type of bank.

\subsection{Definition of the credible set and its boundaries}
We introduce some notations first. Let $\widehat\lambda_j^{\rm SH}$ be the default intensity under $k^{\rm SH}$ when there are $j$ loans left, that is to say 
\[\widehat\lambda_j^{\rm SH}:=\alpha_{j}j(1+\varepsilon).\]
Observe that $\widehat\lambda_j^{\rm SH}=\lambda_t^{k^{\rm SH}}=\alpha_{I-N_t}(I-N_t)(1+\varepsilon),$ for every $t\geq 0$ such that $I-N_t=j$. Define then for any integer $j$ between $1$ and $I$, the set $\widehat{\mathcal{V}}_j:=\big[Bj/\big(r+\widehat\lambda_j^{\rm SH}\big),\infty\big)$. Observe that the feasible set 
\[
{\mathcal{V}}_t=\left[\frac{B(I-N_t)}{r+\lambda_t^{k^{\rm SH}}},+\infty\right),
\]
satisfies ${\mathcal{V}}_t=\widehat{\mathcal{V}}_{I-N_t}$ for every $t\geq 0$, so the only dependence of the feasible set in time is due to the number of loans left. The rigorous definition of the credible set is the following.

\begin{Definition}
For any time $t\geq 0$, we define the credible set ${\mathcal C}_{t}$ as the set of $(u^{b},u^g)\in{\mathcal{V}}_{t}\times{\mathcal{V}}_{t}$ such that there exists some admissible contract $(\theta,D)\in \Theta \times \mathcal D$ satisfying $U_t^b(\theta,D)=u^b$, $U_t^g(\theta,D)=u^g$ and $(U_s^b(\theta,D), U_s^g(\theta,D)) \in{\mathcal{V}}_{s}\times{\mathcal{V}}_{s}$ for every $s\in [t,\tau)$, $\P-${\rm a.s.}
\end{Definition}

Given a starting time $t\geq 0$ and $u^{b}\in{\mathcal{V}}_{t}$, define the set of contracts under which the value function of the bad bank at time $t$ is equal to $u^b$

\begin{equation*}
\mathcal{A}^b(t,u^b) := \big\{ (\theta,D)\in\Theta\times \mathcal D: U_t^b(\theta,D)=u^b  \big\}.
\end{equation*}

We denote by $\mathfrak U_t(u^{b})$ the largest value $U_t^g(\theta,D)$ that the good bank can obtain from all the contracts $(\theta,D)\in\mathcal{A}^b(t,u^b)$. We also denote the lowest value by ${\mathfrak L}_t(u^{b})$. Next, define
\[
\overline{\mathcal C}_t := \big\{ (u^{b},u^g)\in{\mathcal{V}}_t\times{\mathcal{V}}_t : {\mathfrak L}_t(u^{b}) \leq u^g \leq {\mathfrak U}_t(u^{b}) \big\}.
\]
We will prove in Proposition \ref{credibleset} below that ${\mathcal C}_t=\overline{\mathcal C}_t$ for every $t\geq 0$, and that the dependence on time of the credible set, exactly as for the feasible set, only comes from the value of $I-N_t$. In particular, this allows us to call respectively the functions $ {\mathfrak L}_t$ and ${\mathfrak U}_t$ the lower and upper boundary of the credible set when there are $I-N_t$ loans left. The aim of the next sections is prove all these claims and to obtain explicit formulas for the boundaries. We start with some useful technical results concerning specific contracts for which the banks do not monitor the loans at all.

\subsection{Utility of not monitoring} 
Consider any starting time $t$ such that $I-N_t=j$ and any $\theta\in\Theta$. The continuation utility that the banks get from always shirking (without considering the payments) is
\begin{equation}\label{always shirking}
u_{t}^g\big(k^{\rm SH},\theta,0\big) = u_{t}^b\big(k^{\rm SH},\theta,0\big) = \E^{\P^{k^{\rm SH}}}\bigg[ \int_{t\wedge\tau}^\tau \mathrm{e}^{-r(s-t)}Bk_s^{\rm SH}\mathrm{d}s \bigg| \mathcal G_t  \bigg].
\end{equation}
This quantity is obviously non--decreasing in $\theta$, so that \eqref{always shirking} attains its minimum value under any contract with $\theta\equiv0$, which is equal to $
c(j,1):=B j/\big(r+\widehat\lambda_j^{\rm SH}\big).
$
The following proposition provides the value of \eqref{always shirking} when the pool is liquidated exactly after a fixed number of defaults $m$.

\begin{Proposition} \label{prop:utility-shirking}
Fix some $t\geq 0$ and let $j:=I-N_t$. For $m \in\{ 1,\dots,j\}$, let $\theta^m\in\Theta$ be such that the pool is liquidated exactly after the $m$th default occurring after time $t$, that is 
\[
\theta_s^m := \begin{cases}
1,\ t\leq s \leq \tau^{N_t+m}, \\ 0, \ s > \tau^{N_t+m}.
\end{cases}
\]
The utility that the bank of type $\rho_i$ gets from shirking is
\[
c(j,m): = \frac{B j}{r+\widehat\lambda_{j}^{\rm SH}} + \ds\sum_{i=j-m+1}^{j-1} \frac{Bi}{r +\widehat \lambda_{i}^{\rm SH}}   \prod_{\ell=i+1}^{j} \ds\frac{\widehat\lambda_{\ell}^{\rm SH}}{r+\widehat\lambda_{\ell}^{\rm SH}}.
\]
\end{Proposition}

In particular, under any contract such that $\theta\equiv1$, \eqref{always shirking} attains its maximum value, which is equal to 
\begin{equation}\label{eq:C}
C(j):= c(j,j) = \frac{B j}{r+\widehat\lambda_{j}^{\rm SH}} + \ds\sum_{i=1}^{j-1} \frac{Bi}{r +\widehat \lambda_{i}^{\rm SH}}   \prod_{\ell=i+1}^{j} \ds\frac{\widehat\lambda_{\ell}^{SH}}{r+\widehat\lambda_{\ell}^{\rm SH}}.
\end{equation}

\subsection{Lower boundary of the credible set}
The lower boundary of the credible set is the simpler of the two boundaries and it can be computed directly. We will see that it is a piecewise linear function corresponding to two lines with different slopes. All proofs for this section are collected in Appendix  \ref{sec:lower}. The next proposition states the main inequalities that determine the lower boundary.

\begin{Lemma}\label{lemma:inequalities}
For any $t\in[0,\tau]$ and any admissible contract $(\theta,D)\in\Theta\times\mathcal D$, the value functions of the good and the bad banks satisfy, $\P-${\rm a.s.} 
\begin{align}
& U_{t}^g(\theta,D)  \geq U_{t}^{b}(\theta,D), \label{ineq1} \\
& U_{t}^g(\theta,D)  \geq \frac{\rho_g}{\rho_b}U_{t}^{b}(\theta,D)-\frac{(\rho_g-\rho_b)}{\rho_b} C(I-N_{t}) \label{ineq2}, 
\end{align}
where the function $C$ is defined in \eqref{eq:C}.
\end{Lemma}

\vspace{0.5em}
Using Lemma \ref{lemma:inequalities}, we prove the following characterisation of the lower boundary of the credible set.
\begin{Proposition}\label{lowerboundary}
For any $t\geq 0$, and any $u^b\in\Vc_t$, the lower boundary of the credible set is given by
\[
{\mathfrak L}_t(u^b)= \begin{cases}
\displaystyle u^b, \ c(I-N_t,1) \leq u^b \leq C(I-N_t), \\
\displaystyle \frac{\rho_g}{\rho_b}u^b-\frac{(\rho_g-\rho_b)}{\rho_b}C(I-N_t), \ C(I-N_t)\leq u^b<+\infty. 
\end{cases}
\]
In particular, the dependence in $t$ of ${\mathfrak L}_t(u^b)$ only comes from the number of non--defaulted loans at time $t$ and we can define for any $j\in\{1,\dots,I\}$, the quantity $\widehat{\mathfrak L}_j(u^b)$ given by
\[
\widehat{\mathfrak L}_j(u^b):= \begin{cases}
\displaystyle u^b, \ c(j,1) \leq u^b \leq C(j), \\
\displaystyle \frac{\rho_g}{\rho_b}u^b-\frac{(\rho_g-\rho_b)}{\rho_b}C(j), \ C(j)\leq u^b<+\infty,
\end{cases}
\]
for which we have $\widehat{\mathfrak L}_{I-N_t}(u^b)={\mathfrak L}_t(u^b)$.
\end{Proposition}


\begin{Remark}
Of course, the computations of this section depend on our modelling choices, and are unlikely to be directly adaptable to other situations. There is however a generic way of finding the lower boundary $($as well as the upper one$)$ which we give details in the next section. It amounts to solving a fictitious contract situation where the good bank hires the bad one and minimises $($maximises for the upper boundary$)$ her utility over her monitoring choices, and over all contracts for which the bad bank receives a fixed utility $u^b$. The dynamic value function of this control problem is exactly $\mathfrak L_t(u^b)$, since it corresponds to the minimal utility that the good bank can have when the bad one receives $u^b$.
\end{Remark}

\subsection{Upper boundary of the credible set}

The upper boundary of the credible set is not as simple to obtain as the lower boundary and we have to solve a specific stochastic control problem to identify it. Notice that this approach is similar to the one used in \cite{cvitanic2013dynamics}.

\vspace{0.5em}
Let us fix any contract $(\theta,D)\in\Theta\times\Dc$. We remind the reader that thanks to Proposition \ref{prop:rep}, we know that there exist $\G-$predictable integrable processes $(h^{1,g}(\theta,D),h^{2,g}(\theta,D))$ satisfying the second integrability condition in \eqref{eq:integ} and such that
\begin{align}\label{dynamic ug}
\nonumber \mathrm{d}U_s^g(\theta,D) = & \big( rU_s^g(\theta,D)-Bk_s^{\star,g}(\theta,D) \big) \mathrm{d}s-\rho_g \mathrm{d}D_s-h_s^{1,g}(\theta,D)\big(\mathrm{d}N_s-\lambda_s^{k^{\star,g}(\theta,D)}\mathrm{d}s\big)\\
&-h_s^{2,g}(\theta,D)\big(\mathrm{d}H_s-(1-\theta_s)\lambda_s^{k^{\star,g}(\theta,D)}\mathrm{d}s\big), ~ s\in[0,\tau],
\end{align}
where the optimal monitoring choice $k^{\star,g}(\theta,D)$ is given by $k_s^{\star,g}(\theta,D) = (I-N_s) {\bf 1}_{ \{ h_s^{1,g}(\theta,D)+(1-\theta_s)h_s^{2,g}(\theta,D) < b_s \} }$. Similarly, there exist $\G-$predictable processes $(h^{1,b}(\theta,D),h^{2,b}(\theta,D))$ satisfying the second integrability condition in \eqref{eq:integ} and such that
\begin{align}\label{dynamic ubc}
\nonumber \mathrm{d}U_s^b(\theta,D) =& \big( rU_s^b(\theta,D)-Bk_s^{\star,b}(\theta,D)\big) \mathrm{d}s -\rho_b \mathrm{d}D_s-h_s^{1,b}(\theta,D)\big(\mathrm{d}N_s-\lambda_s^{k^{\star,b}(\theta,D)}\mathrm{d}s\big)\\
&-h_s^{2,b}(\theta,D)\big(\mathrm{d}H_s-(1-\theta_s)\lambda_s^{k^{\star,b}(\theta,D)}\mathrm{d}s\big),  ~ s\in[0,\tau],
\end{align}
with $
k_s^{\star,b}(\theta,D) = (I-N_s)  {\bf 1}_{ \{ h_s^{1,b}(\theta,D)+(1-\theta_s)h_s^{2,b}(\theta,D) < b_s \} }. 
$ We will use the dynamics \eqref{dynamic ug}--\eqref{dynamic ubc} to define a simple set of admissible contracts in which we will reinterpret both the value functions of the agents as controlled diffusion processes, where the controls are $(D,\theta,h^{1,g},h^{2,g},h^{1,b},h^{2,b})$, and which satisfy the instantaneous conditions \eqref{eq:yz}. Obviously, doing so makes us, at least at first sight, look at a larger class of "contracts", in the sense that in the above representation of the value functions of the bank, the choice of the processes $(h^{1,g},h^{2,g},h^{1,b},h^{2,b})$ is not free, since they are completely determined by the choice of $(\theta,D)$. Nonetheless, as we will see below, this still describes exactly the same set of contracts.

\vspace{0.5em}
In the meantime, let us denote by $\mathcal{H}$ the set of non--negative, $\G-$predictable processes $h$ satisfying for some $\beta>0$
\[
\E^\P\bigg[\int_0^\tau\mathrm{e}^{(\beta\varepsilon^2-2r)s}|h_s|^2\mathrm{d}s\bigg]<+\infty.
\]
We abuse notations and define, for every $\Psi:=(D,\theta,h^{1,g},h^{2,g},h^{1,b},h^{2,b})\in{\mathcal D}\times\Theta\times{\mathcal H}^4$, the processes $U^g(\Psi)$ and $U^b(\Psi)$ which satisfy the following SDEs \begin{align}\label{eq:sde good value}
\mathrm{d}U_s^g(\Psi) &=  \big(rU_s^g(\Psi)-Bk_s^{\star,g}(\Psi)\big)\mathrm{d}s-\rho_g \mathrm{d}D_s-h_s^{1,g}\big(\mathrm{d}N_s-\lambda_s^{k^{\star,g}(\Psi)}\mathrm{d}s\big)-h_s^{2,g}\big(\mathrm{d}H_s-(1-\theta_s)\lambda_s^{k^{\star,g}(\Psi)}\mathrm{d}s\big),\\
\mathrm{d}U_s^b(\Psi) &=  \big(rU_s^b(\Psi)-Bk_s^{\star,b}(\Psi)\big)\mathrm{d}s-\rho_b \mathrm{d}D_s-h_s^{1,b}\big(\mathrm{d}N_s-\lambda_s^{k^{\star,b}(\Psi)}\mathrm{d}s\big) -h_s^{2,b}\big(\mathrm{d}H_s-(1-\theta_s)\lambda_s^{k^{\star,b}(\Psi)}\mathrm{d}s\big),
\label{eq:sde bad value}
\end{align}
where we defined
\begin{align*}
k_s^{\star,g}(\Psi)& := (I-N_s) {\bf 1}_{ \{ h_s^{1,g}+(1-\theta_s)h_s^{2,g} < b_s \} },\ k_s^{\star,b}(\Psi): = (I-N_s)  {\bf 1}_{ \{ h_s^{1,b}+(1-\theta_s)h_s^{2,b}< b_s \} }.
\end{align*}

\begin{Remark}
In the model, there is no need to consider $h^{1,g}$ and $h^{1,b}$ as positive processes and we do this just for technical reasons. Intuitively, the optimal contracts should satisfy this additional constraint because the investor does not benefit from earlier defaults and if a contract increases the banks' continuation utilities after one of the defaults, the banks should increase the default intensity as much as possible.
\end{Remark}
\begin{Remark}
It is immediate from the definition that given $\Psi:=(D,\theta,h^{1,g},h^{2,g},h^{1,b},h^{2,b})\in{\mathcal D}\times\Theta\times{\mathcal H}^4$, the dynamics of $U^g(\Psi)$ only depends on $(D,\theta,h^{1,g},h^{2,g})$, while the dynamics of $U^b(\Psi)$ only depends on $(D,\theta,h^{1,b},h^{2,b})$. We will thus sometimes also use the notations $U^g(\Psi)$, $U^b(\Psi)$, $k^{\star,g}(\Psi)$ and $k^{\star,b}(\Psi)$ when $\Psi\in {\mathcal D}\times\Theta\times{\mathcal H}^2$.
\end{Remark}

For fixed $(t,u^b,u^g)\in\R_+\times \Vc_t^2$, we define the set of contracts $\overline{\Ac}(t,u^b,u^g)$ as the set of $\Psi:=(D,\theta,h^{1,g},h^{2,g},h^{1,b},h^{2,b})\in{\mathcal D}\times\Theta\times{\mathcal H}^4$ such that \eqref{eq:sde good value} and \eqref{eq:sde bad value} have at least one weak solution\footnote{In general, all the processes in $\Psi$ could for instance functionals of the paths of $U^g(\Psi)$ and $U^b(\Psi)$, in which case wellposedness of the SDEs has to be assumed as part of the definition.}, which satisfies the first integrability condition in \eqref{eq:integ}, and in addition
\[
U_{s^-}^{i} (\Psi)= h_s^{1,i}+h_s^{2,i} ,\ U_{s^-}^{i}(\Psi)-h_s^{1,i}\geq\frac{B(I-N_s)}{r+\lambda_s^{\rm SH}},\; \forall s\in[t,\tau],\  U_t^{i}(\Psi)=u^{i},\; i\in\{b,g\}.
\]
What we claimed above is that all processes $(D,\theta)\in\Dc\times\Theta$ can  be obtained from a contract $(D,\theta,h^{1,g},h^{2,g},h^{1,b},h^{2,b})=:\Psi\in\overline{\Ac}(0,u^b,u^g)$, meaning that we are not enlarging at all the class of admissible contracts in our reformulation. Indeed, we already know by the results from Section \ref{sec:pure}, that the continuation utilities of the good and the bad bank given a contract $(D,\theta)\in\Dc\times\Theta$ were completely characterised as being the unique solutions of the corresponding BSDEs \eqref{bsde} satisfying in addition \eqref{eq:integ}. If we take some $\Psi\in \overline{\Ac}(0,u^b,u^g)$, then it is immediate that the processes $U^g(\Psi)$, $U^b(\Psi)$ solve the corresponding BSDEs \eqref{bsde}, since the dynamics is the correct one by definition, we have $U^g_\tau(\Psi)=U^b_\tau(\Psi)=0$, and all the required integrability conditions are satisfied. By uniqueness of the solution to the BSDEs, we thus must have $U^g(\Psi)=U^g(\theta,D)$, and $U^b(\Psi)=U^b(\theta,D)$.

\vspace{0.5em}
To describe the stochastic control problem for the upper boundary of the credible set, we need to introduce additional notations. For any starting time $t\in[0,\tau]$ and for every $u^b\geq B(I-N_t)/\big(r+\widehat\lambda_{I-N_t}^{\rm SH}\big)$, we let $\overline{\mathcal A}^{b}(t,u^b)$ be the set of quadruplets $\Psi=(D,\theta,h^{1,b},h^{2,b})\in{\mathcal D}\times\Theta\times{\mathcal H}^2$ such that \eqref{eq:sde bad value} has at least one weak solution, which satisfies the first integrability condition in \eqref{eq:integ} as well as

\begin{align*}
U_{s^-}^{b} (\Psi)= h_s^{1,b}+h_s^{2,b} ,\ U_{s^-}^{b}(\Psi)-h_s^{1,b}\geq\frac{B(I-N_s)}{r+\lambda_s^{I-N_s}},\; \forall s\in[t,\tau],\ U_t^{b}(\Psi)=u^{b}.
\end{align*}

We will abuse notations and also call elements of $\overline{\mathcal A}^{b}(t,u^b)$ contracts. The upper boundary $\mathfrak U_t$ solves the following control problem
\[
{\mathfrak U}_t(u^b)=\underset{(k^g,\Psi)\in \mathfrak K\times{\overline{\mathcal{A}}}^b(t,u^b)}{\rm{ess\ sup}} \ \E^{\P^{k^{g}}} \bigg[  \int_{t\wedge\tau}^\tau \mathrm{e}^{-r(s-t)} \big(\rho_g \mathrm{d}D_s+B k_s^{g} \mathrm{d}s\big) \bigg| \mathcal{G}_t \bigg],
\]
subject to the dynamics 
\begin{align*}
U_r^b(\Psi)= u^b+\int_t^r\big(  ru_s^{b}-Bk_s^{\star,b}(\Psi) +h_s^{1,b}\lambda_s^{k^{\star,b}} + h_s^{2,b}(1-\theta_s)\lambda_s^{k^{\star,b}} \big)\mathrm{d}s -\rho_b \mathrm{d}D_s-\int_t^rh_s^{1,b}\mathrm{d}N_s-\int_t^rh_s^{2,b}\mathrm{d}H_s,\; r\in[t,\tau].
\end{align*}

\vspace{0.5em}
Indeed, the above stochastic control problem corresponds to the highest value that the good bank can obtain from any admissible contract, while ensuring that when the bad bank takes said contract, she receives exactly $u^b$, which is exactly the definition of the upper boundary of the credible set. Another way to interpret this problem is that it corresponds to the (fictitious) situation where the good bank hires the bad one, when the latter wants to receive a utility of $u^b$, and maximises her utility among all contracts ensuring that this constraint is satisfied. The importance of the results of Proposition \ref{prop:rep} is that it allows us to obtain easily the dynamic behaviour of the continuation utility of the bad bank for any initial utility, which in turns allows us to express simply the constraint in the problem for the good bank through the set $\overline{\Ac}^b$ and the state variable $U^b(\Psi)$.

\vspace{0.5em}
The next subsections are devoted to first obtaining the HJB equation associated with the above problem, its resolution and then finally to the proof of a verification theorem adapted to our framework. Notice that the above is actually a singular stochastic control problem, since the control $D$ is a non--decreasing process, which is not necessarily absolutely continuous with respect to the Lebesgue measure. We refer the reader to the monograph by Fleming and Soner \cite{fleming2006controlled} for more details. In particular, this implies that the HJB equation associated to the problem will be a variational inequality with gradient constraints.

\subsubsection{HJB equation for the upper boundary}

Exactly as in the case of the lower boundary, we expect that the time dependence of of the upper boundary only comes from the current number of remaining loans. In such a case, the HJB equations that will describe the behaviour of the upper boundary necessarily form a recursive system, with the upper boundary when $j$ loans are left depending on the one with $j-1$ loans left. We will write down this system, solve it explicitly, and prove a verification theorem ensuring that our initial guess was indeed correct.

\vspace{0.5em}
Fix some $1\leq j\leq I$, and define for every $k=0,1, \cdots, j$, $\widehat{\lambda}_j^{k}:=\alpha_{j}(j+k \varepsilon).$ The system of HJB equations associated to the previous control problem is given by $\widehat{\mathcal U}_0\equiv0$, and for any $1\leq j\leq I$ and $u^b\geq \frac{Bj}{r+\widehat\lambda_j^{\rm SH}}$
\begin{equation} \label{DPequation}
\min \Bigg\{ 
- \sup_{(\theta,h^1,h^2)\in C^j} 
\left\{ 
\begin{array}{c} 
\widehat{\mathcal U}_j^\prime(u^b) \big(u^b - Bk^b + (h^1 + (1-\theta)h^2) \widehat\lambda^{k^b}_j \big) \\[0.3em] 
+ \widehat\lambda^{k^g}_j\theta \widehat{\mathcal U}_{j-1}(u^b - h^1) -(\widehat\lambda^{k^g}_j+r) \widehat{\mathcal U}_j(u^b)  +Bk^g   
\end{array} 
\right\},  \
 \widehat{\mathcal U}_j^\prime(u^b)-\frac{\rho_g}{\rho_b} 
\Bigg\}
=0,
\end{equation}
with the additional boundary condition $\widehat{\mathcal U}_j(Bj/(r+\widehat\lambda_j^{SH}))= Bj/(r+\widehat\lambda_j^{SH}),$
and where we defined for simplicity 
\[
k^b:=j{\bf 1}_{\{h^1+(1-\theta)h^2<\widehat b_j\}},\ k^g:=j{\bf 1}_{\{\widehat{\mathcal U}_j(u^b)-\theta \widehat{\mathcal U}_{j-1}(u^b-h^1)<\widehat b_j\}},
\]
as well as
\[
C^j:=\bigg\{(\theta,h^1,h^2)\in[0,1]\times\R_+^2:h^1+h^2=u^{b},~ h^2\geq\frac{B(j-1)}{r+\widehat\lambda_{j-1}^{\rm SH}}  \bigg\} .
\]

\begin{Remark} \label{remark incentive compatibility}
Notice that if our guess on the time dependence of the upper boundary is correct, we must have for any $s\geq 0$, $\mathfrak U_s=\widehat{\mathfrak U}_{I-N_s}$. Then. the incentive compatibility condition for the good bank is implicit in the {\rm HJB} equation. Indeed, at every $s\geq 0$ we have
\begin{align*}
 \widehat{\mathcal U}_{I-N_s}\big(U_s^{b}(\Psi)\big)- \widehat{\mathcal U}_{I-N_{s^-}}\big(U_{s^-}^{b}(\Psi)\big) 
= &~ \big( \widehat{\mathcal U}_{I-N_{s^-}-1}\big(U_{s^-}^{b}(\Psi)-h_s^{1,b}(\Psi)\big)- \widehat{\mathcal U}_{I-N_{s^-}}\big(U_{s^-}^{b}(\Psi)\big)\big)\Delta N_s \\
& -  \widehat{\mathcal U}_{I-N_{s^-}-1}\big(U_{s^-}^{b}(\Psi)-h_s^{1,b}(\Psi)\big)\Delta H_s,
\end{align*}
which implies that on the upper boundary $h_s^{1,g}(\Psi)= \widehat{\mathcal U}_{I-N_{s^-}}\big(U_{s^-}^{b}(\Psi)\big)- \widehat{\mathcal U}_{I-N_{s^-}-1}\big(U_{s^-}^{b}(\Psi)-h_s^{1,b}(\Psi)\big)$ and $h_s^{2,g}(\Psi)= \widehat{\mathcal U}_{I-N_{s^-}-1}\big(U_{s^-}^{b}(\Psi)-h_s^{1,b}(\Psi)\big)$. Therefore
\[
h_s^{1,g}(\Psi)+(1-\theta_s^g)h_s^{2,g}(\Psi)= \widehat{\mathcal U}_{I-N_{s^-}}\big(U_{s^-}^{b}(\Psi)\big)-\theta_s^g  \widehat{\mathcal U}_{I-N_{s^-}-1}\big(U_{s^-}^{b}(\Psi)-h_s^{1,b}(\Psi)\big).
\]
\end{Remark}

\vspace{0.5em}
At the points where $ \widehat{\mathcal U}_j^\prime(u^{b})>\rho_g/\rho_b$, the first term of the variational inequality \eqref{DPequation} must be equal to zero, so the upper boundary must satisfy the following equation 
\begin{equation}\label{HJBequation}
r \widehat{\mathcal U}_j(u^{b})=\sup_{(\theta,h^1,h^2)\in C^j} 
\Big\{ 
 \widehat{\mathcal U}_j^\prime(u^{b}) \big( ru^{b}-Bk^{b} + (h^1+(1-\theta)h^2)\widehat\lambda^{k^{b}}_j \big) +\big( \widehat{\mathcal U}_{j-1}(u^{b}-h^1)\theta- \widehat{\mathcal U}_j(u^{b})\big)\widehat\lambda^{k^g}_j +Bk^g 
\Big\} .
\end{equation}
We will refer to this equation as the diffusion equation.

\vspace{0.5em}
\hspace{3em}{$\bullet$ \bf Step 1: case of 1 loan, solving the diffusion equation}

\vspace{0.5em}
Before dealing with the variational inequality \eqref{DPequation}, we will solve the diffusion equation \eqref{HJBequation}. When $j=1$, it reduces to 
\begin{equation} \label{hjb1}
r \widehat{\mathcal U}_1(u^{b}) =  \widehat{\mathcal U}_1^\prime(u^{b})\big(ru^{b}-Bk^{b}+u^{b}\widehat\lambda^{k^{b}}_1 \big) - \widehat{\mathcal U}_1(u^{b})\widehat\lambda^{k^g}_1 +Bk^g,
\end{equation}
with $k^{b}={\bf 1}_{\{u^{b}<\widehat b_1\}},~ k^g={\bf 1}_{\{ \widehat{\mathcal U}(u^{b})<\widehat b_1\}}$.

\begin{Remark}
Notice that the boundary condition $ \widehat{\mathcal U}_1\big(\frac{B}{r+\widehat\lambda_1^1}\big)=\frac{B}{r+\widehat\lambda_1^1}$ is implicit in the equation. 
\end{Remark}
Our first result is the following, whose proof is deferred to Appendix \ref{sec:D}.
\begin{Lemma}\label{lemma:diffusion}
There is a family of continuously differentiable solutions to the diffusion equation \eqref{HJBequation}, indexed by some constant $C_0>0$, which are given by 
\[
\widehat{\mathcal U}_1^{C_0}(u^{b}):=\begin{cases}
\displaystyle
C_0^\frac{r+\widehat\lambda_1^1}{r+\widehat\lambda_1^0}\bigg(u^{b}-\ds\frac{B}{r+\widehat\lambda_1^1}\bigg)+\ds\frac{B}{r+\widehat\lambda_1^1}, \ u^{b}\in\bigg[\frac{B}{r+\widehat\lambda_1^1},x_1^{C_0,\star}\bigg),\\[0.8em] 
\displaystyle C_0{\widehat b_1}^\frac{\widehat\lambda_1^1-\widehat\lambda_1^0}{r+\widehat\lambda_1^1}\bigg(\ds\frac{r+\widehat\lambda_1^1}{r+\widehat\lambda_1^0}\bigg)^\frac{r+\widehat\lambda_1^0}{r+\widehat\lambda_1^1}\bigg(u^{b}-\ds\frac{B}{r+\widehat\lambda_1^1}\bigg)^\frac{r+\widehat\lambda_1^0}{r+\widehat\lambda_1^1}, \ u^{b}\in\big[x_1^{C_0,\star},\widehat b_1\big),\\[0.8em] 
\displaystyle C_0  u^{b}, \  u^{b}\in\big[\widehat b_1,+\infty\big),
\end{cases}
\]
where $\ds x_1^{C_0,\star}:=\bigg(\frac{1}{C_0}\bigg)^{\frac{r+\widehat\lambda_1^1}{r+\widehat\lambda_1^0}}\widehat b_1\frac{r+\widehat\lambda_1^0}{r+\widehat\lambda_1^1}+\frac{B}{r+\widehat\lambda_1^1}$.

\end{Lemma}

\vspace{0.5em}

\hspace{3em}{$\bullet$ \bf Step 2: case of 1 loan, solving the HJB equation} 

\vspace{0.5em}
In this case the variational inequality \eqref{DPequation} reduces to 
\begin{equation}\label{dp1}
\min\left\{r \widehat{\mathcal U}_1(u^{b}) -\widehat{\mathcal U}^\prime(u^{b})\big(ru^{b}-Bk^{b}+u^{b}\widehat\lambda^{k^{b}}_1 \big) +\widehat{\mathcal U}_1(u^{b})\widehat\lambda^{k^g}_1 -Bk^g,~ \widehat{\mathcal U}_1^\prime(u^{b})-\frac{\rho_g}{\rho_b}   \right\}=0.
\end{equation}
We already found the solutions of the diffusion equation inside of this variational inequality and now we will take care of the whole HJB equation. We expect the upper boundary to saturate the second term in the variational inequality for big values of $u^{b}$, so we will search for a solution of \eqref{dp1} satisfying the following condition: there exists $x^\star \in[B/(r+\widehat\lambda_1^1),\infty)$ such that 
\begin{equation}\label{saturation}
\widehat{\mathcal U}_1^\prime(x^\star)=\frac{\rho_g}{\rho_b}~\textrm{and }  \widehat{\mathcal U}_1^\prime(u^{b}) >\frac{\rho_g}{\rho_b},~ \textrm{for } u^{b}<x^\star.
\end{equation}
At first sight it could seem that by doing this we face the risk of not finding the correct solution of the dynamic programming equation. Nevertheless, this is not the case and we will prove later a verification result which assures us that the solution that we find under this condition corresponds indeed to the upper boundary of the credible set. The proof of the following Lemma will be given in Appendix \ref{sec:D}.

\vspace{0.5em}
\begin{Lemma}\label{lemma:HJB1}
The unique solution of the {\rm HJB} equation \eqref{dp1} which satisfies condition \eqref{saturation} is given by, defining $x^\star_1:=x_1^{\rho_g/\rho_b,\star}$ 
\begin{equation}\label{HJBsolution1}
\widehat{\mathcal U}_1^\star(u^{b}):=\widehat{\mathcal U}_1^{\rho_g/\rho_b}(u^{b})=\begin{cases}
\displaystyle\bigg(\frac{\rho_g}{\rho_b}\bigg)^\frac{r+\widehat\lambda_1^1}{r+\widehat\lambda^0_1}\bigg(u^{b}-\displaystyle\frac{B}{r+\widehat\lambda_1^1}\bigg) +\displaystyle\frac{B}{r+\widehat\lambda_1^1}, \ u^{b}\in\bigg[\frac{B}{r+\widehat\lambda_1^1},x_1^{\star}\bigg), \\[0.8em]
\displaystyle\frac{\rho_g}{\rho_b}{\widehat b_1}^\frac{\widehat\lambda_1^1-\widehat\lambda^0_1}{r+\widehat\lambda_1^1}\bigg(\displaystyle\frac{r+\widehat\lambda_1^1}{r+\widehat\lambda^0_1}\bigg)^\frac{r+\widehat\lambda^0_1}{r+\widehat\lambda_1^1}\bigg(u^{b}-\displaystyle\frac{B}{r+\widehat\lambda_1^1}\bigg)^\frac{r+\widehat\lambda^0_1}{r+\widehat\lambda_1^1}, \ u^{b}\in\big[x_1^\star,\widehat b_1\big), \\[0.8em] 
\displaystyle\frac{\rho_g}{\rho_b}  u^{b}, \  u^{b}\in\big[\widehat b_1,+\infty\big).
\end{cases} 
\end{equation}
\end{Lemma}

As an illustration, in Figure \ref{fig:cset1} we show the credible set which corresponds to the region delimited by its upper and lower boundaries. In this example, we considered $r=0.02$, $B=0.002$, $\varepsilon=0.25$, $\alpha_1=0.055$, $\frac{\rho_g}{\rho_b}=2$.

\begin{figure}[!ht]
\centering
\includegraphics[scale=0.9]{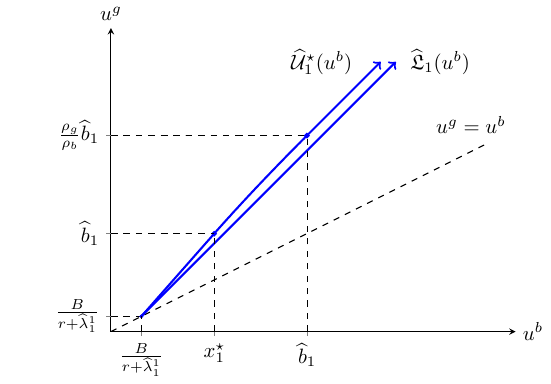}
\caption{Credible set with one loan left.}  \label{fig:cset1}
\end{figure}

\vspace{0.5em} 
\hspace{3em}{$\bullet$ \bf Step 3: solving the HJB equation in the general case} 

\vspace{0.5em} 
In the general case, when $j>1$, we can reduce the number of variables and rewrite the diffusion equation \eqref{HJBequation} in an equivalent form
\begin{equation}\label{hjbj}
r\widehat{\mathcal U}_j(u^{b})=\sup_{(\theta,h^1)\in \widehat{C}^j} 
\Big\{
\widehat{\mathcal U}_j^\prime(u^{b})\big(ru^{b}-Bk^{b}+[u^{b}-\theta(u^{b}-h^1)]\widehat\lambda^{k^{b}}_j \big) + \big(\widehat{\mathcal U}_{j-1}(u^{b}-h^1)\theta-\widehat{\mathcal U}_j(u^{b})\big)\widehat\lambda^{k^g}_j +Bk^g  
\Big\},
\end{equation}
where we recall that $k^{b}={\bf 1}_{\{u^{b}-\theta(u^{b}-h^1)<\widehat b_j\}} ,~ k^g={\bf 1}_{\{\widehat{\mathcal U}_j(u^{b})-\theta \widehat{\mathcal U}_{j-1}(u^{b}-h^1)<\widehat b_j\}}$ and the set of constraints is now given by
\begin{equation}\label{cjota}
\widehat{C}^j:=
\left\{
(\theta,h^1)\in[0,1]\times\R_+,~ u^{b}\geq h^1+\frac{B(j-1)}{r+\widehat\lambda_{j-1}^{\rm SH}}
 \right\}.
\end{equation}
When we proved that the lower boundary of the credible set is reachable we used contracts of maximum duration, which maintain the pool until the last default. This gives us the intuition that the longer the contract lasts, the smaller the difference between the utilities of the banks will be. Therefore the upper boundary of the credible set, where the difference between both utilities is maximal, should be reachable with contracts of minimum duration, which terminate the contractual relationship immediately after the first default. In the model this means that $\theta$ is equal to zero and the resulting HJB equation for the upper boundary has the same form as the one in the case with one loan left. We expect then that the solution of the diffusion equation will be the of the same form as \eqref{HJBsolution1}. The object of the next proposition is to prove our guess rigorously. We postpone the proof to Appendix \ref{sec:D}.

\begin{Proposition} \label{prop:uj functions} For any $j\geq 1$, the function $\widehat{\mathcal U}_j^\star$ defined by 
\begin{equation} \label{ujota} 
\widehat{\mathcal U}_j^\star(u^{b}):=\begin{cases}
\displaystyle\bigg(\frac{\rho_g}{\rho_b}\bigg)^\frac{r+\widehat\lambda_j^{\rm SH}}{r+\widehat\lambda_j^0}\bigg(u^{b}-\ds\frac{Bj}{r+\widehat\lambda_j^{\rm SH}}\bigg) +\ds\frac{Bj}{r+\widehat\lambda_j^{SH}}, \ u^{b}\in\bigg[\frac{Bj}{r+\widehat\lambda_j^{SH}},x_j^\star\bigg), \\[0.8em] 
\displaystyle \frac{\rho_g}{\rho_b}{\widehat b_j}^\frac{\widehat\lambda_j^{\rm SH}-\widehat\lambda_j^0}{r+\widehat\lambda_j^{\rm SH}}\bigg(\frac{r+\widehat\lambda_j^{SH}}{r+\widehat\lambda_j^0}\bigg)^\frac{r+\widehat\lambda_j^0}{r+\widehat\lambda_j^{\rm SH}}\bigg(u^{b}-\ds\frac{Bj}{r+\widehat\lambda_j^{\rm SH}}\bigg)^\frac{r+\widehat\lambda_j^0}{r+\widehat\lambda_j^{\rm SH}} ,\ u^{b}\in\big[x_j^\star,\widehat b_j\big), \\[0.8em]
\displaystyle\frac{\rho_g}{\rho_b}u^{b},\ u^{b}\in\big[\widehat b_j,+\infty\big), 
\end{cases}
\end{equation} 
where $\ds x_j^{\star}:=\bigg(\frac{\rho_b}{\rho_g}\bigg)^{\frac{r+\widehat\lambda_j^{\rm SH}}{r+\widehat\lambda_j^0}}\widehat b_j\frac{r+\widehat\lambda_j^0}{r+\widehat\lambda_j^{\rm SH}}+\frac{Bj}{r+\widehat\lambda_j^{\rm SH}}$, is a solution of the {\rm HJB} equation \eqref{DPequation}. 
\end{Proposition}

\subsubsection{Verification theorem}

According to the maximisers in equation (\ref{hjbj}) we define the following controls 
\begin{equation} \label{optimalcontrols} 
\begin{cases}
\ds \delta^{j}(u^{b}):={\bf 1}_{\{u^{b}\geq \widehat b_j\}}\frac{u^{b}(r+\widehat\lambda_j^0)}{\rho_b}, \; \ds \theta^{j}(u^{b}):=0, \\
\ds h^{1,b,j}(u^{b}):=u^{b}-\frac{B(j-1)}{r+\widehat\lambda_{j-1}^{\rm SH}}, \; h^{2,b,j}(u^{b}):=\frac{B(j-1)}{r+\widehat\lambda_{j-1}^{\rm SH}}, \\ 
\ds k^{b,j}(u^{b}):=j {\bf 1}_{\{u^{b}<\widehat b_j\}}, \ k^{g,j}(u^{b}):=j {\bf 1}_{\{\widehat{\mathcal U}_j^\star(u^{b})<\widehat b_j\}}. 
\end{cases}
\end{equation}
Before stating the verification result for the upper boundary, we make a comment about the domain of the functions $\widehat{\mathcal U}_j^\star$. Rigorously speaking, it is possible for the utilities of the banks to be zero but this happens only at time $\tau$ when all the pools are liquidated. The domain of $\widehat{\mathcal U}_j^\star$ is the set $\widehat{\mathcal{V}}_j$ but in the proof of the verification theorem it will be implicitly understood that $\widehat{\mathcal U}_j^\star(0)=0$. In any case, we do not need the functions $\widehat{\mathcal U}_j^\star$ to be defined at zero because It\^o's formula will be used on intervals which do not contain $\tau$.

\begin{Theorem} \label{verification theorem}
Consider any starting time $t\geq 0$. For any $u^{b}\geq\frac{B (I-N_t)}{r+\widehat\lambda_{I-N_t}^{\rm SH}}$, let the process $(u^{b}_s)_{s\in[t,\tau]}$ be the unique solution of the following {\rm SDE} 
\begin{equation} \label{bcutility} 
u_v^{b}=u^{b}+\ds\int_t^v\Big(\big(r+\lambda_s^{k^{b,I-N_s}}\big)u_s^{b} - Bk^{b,I-N_s}(u_s^{b})- \rho_b\delta^{I-N_s}(u_s^{b}) \Big)\mathrm{d}s-\int_t^v u_{s^-}^{b}\mathrm{d}N_s,~v\in[t,\tau]. 
\end{equation} 
Then, under the contract $\Psi^{\star}:=(D^{\star},\theta^{\star},h^{1,b,\star},h^{2,b,\star})\in\mathcal{D}\times\Theta\times\mathcal{H}^2$ defined for $s\in[t,\tau]$ by 
\[
\mathrm{d}D_s^{\star}:=\delta^{I-N_s}(u_s^{b})\mathrm{d}s,~\theta_s^{\star}\equiv 0,~ h_s^{1,b,\star}:=h^{1,b,I-N_s}(u_s^{b}),~ h_s^{2,b,\star}:=h^{2,b,I-N_s}(u_s^{b}),
\]
the value function of the bad bank is $U_t^{b}(\Psi^{\star})=u^{b}$ and the one of good bank is $U_t^g(\Psi^{\star})=\widehat{\mathcal U}_{I-N_t}^\star(u^{b})$. Moreover, $\Psi^{\star}\in\overline{\mathcal A}^{b}(t,u^{b})$ and for any other contract which belongs to $\overline{\mathcal A}^{b}(t,u^{b})$, the value function of the good bank under such a contract is less or equal to $\widehat{\mathcal U}_{I-N_t}^\star(u^{b})$. In particular, this implies that 
\[\widehat{\mathcal U}_{I-N_t}^\star(u^{b})=\widehat{\mathfrak U}_{I-N_t}(u^{b})=\mathfrak U_t(u^b).\] 
\end{Theorem}

To conclude the section, we state that $\overline{{\mathcal C}}_j$ is indeed equal to the credible set with $j$ loans left and therefore the functions $\widehat{\mathfrak U}_j$ and $\widehat{\mathfrak L}_j$ correspond to its upper and lower boundaries.

\begin{Proposition} \label{credibleset}
For every $t\geq 0$, ${{\mathcal C}_t}=\overline{\mathcal C}_t$.
\end{Proposition}

\section{Optimal contracts}\label{sec:opti}

In this section we study two kind of contracts that the investor can offer to the bank, the shutdown contract, which corresponds to a single contract designed to be accepted only by the good bank and the screening contract, corresponding to a menu of contracts, one for each type of agent, providing incentives to the bank to accept the contract designed for her true type.

\subsection{Shutdown contract}

In the so--called shutdown contract, the investor designs a contract $\Psi_g=(k^g,D^g,\theta^g)$ only for the good bank and makes sure that the bad bank will not accept it. Under these conditions the utility of the investor at time $t=0$ is
\begin{equation}\label{p utility shutdown}
{v}^{\rm g,Shut}_{0}(\Psi_g) = p_g \E^{\P^{k^g}}\bigg[ \ds\int_0^\tau  \mu(I-N_s)\mathrm{d}s - \mathrm{d}D_s^g \bigg].
\end{equation}
So the investor will offer a contract which maximises \eqref{p utility shutdown} subject to the constraints
\begin{align*}
& u_{0}^g(k^g,\theta^g,D^g)\geq R_0^g, \  \underset{k\in\mathfrak K}{\sup}\ u_{0}^b(k,\theta^g,D^g)\leq R_0^b , \ u_{0}^g(k^g,\theta^g,D^g)= \underset{k\in\mathfrak K}{\sup}\ u_{0}^g(k,\theta^g,D^g).
\end{align*}

Recalling the dynamics \eqref{dynamic ug}--\eqref{dynamic ubc}, we can rewrite the investor's maximisation problem as follows
\[ v_0^{\rm Shut}:=\underset{(\theta^g,D^g)\in{\Ac}^g_{\rm Shut}}{\sup}p_g\E^{\P^{k^{\star,g}(\theta^g,D^g)}}\bigg[\int_0^\tau\mu(I-N_s)\mathrm{d}s-\mathrm{d}D^g_s\bigg],\]
where
\[{\Ac}^g_{\rm Shut}:=\left\{(\theta^g,D^g)\in\Theta\times\Dc: U_{0}^{b,c}(\theta^g,D^g)\leq R_0^b,  U_{0}^g(\theta^g,D^g) \geq R_0^g\right\}.\]

\begin{Remark}
We will use the notation $U^{b,c}(\theta^g,D^g)$ for the value function that the bad bank gets if she does not reveal her true type and accepts the contract designed for the good bank. We make a distinction between this process and $U^b(\theta^b,D^b)$, which corresponds to the value function that the bad bank obtains if she accepts the contract designed for her by the investor. We make the same distinction between the associated processes $h^{1,b,c}(\theta,D)$, $h^{2,b,c}(\theta,D)$ and $h^{1,b}(\theta,D)$, $h^{2,b}(\theta,D)$.
\end{Remark}

As in the previous section, we will define a simple set of contracts and consider the value functions of the agents as diffussion processes controlled by $(D,\theta,h^{1,g},h^{2,g},h^{1,b,c},h^{2,b,c})$. As explained before, by doing so we do not look at a larger class of "contracts". 

\vspace{.5em}
Define for any $(t,u^g,u^{b,c})\in [0,+\infty)\times\mathcal{C}_{t}$, $\widehat{\mathcal A}^{g}(t,u^g,u^{b,c})$ to be the set of $\Psi_g=(D^g,\theta^g,h^{1,g},h^{2,g},h^{1,b,c},h^{2,b,c})\in{\mathcal D}\times\Theta\times{\mathcal H}^4$ such that \eqref{eq:sde good value} and \eqref{eq:sde bad value} have at least one weak solution, which satisfy the first integrability condition in \eqref{eq:integ}, and in addition, for any $s\in[t,\tau]$
\begin{align*}
&U_{s^-}^{g} (\Psi_g)= h_s^{1,g}+h_s^{2,g} ,\
U_{s^-}^{g}(\Psi_g)-h_s^{1,g}\geq\frac{B(I-N_s)}{r+\lambda_s^{I-N_s}},\ U_t^{g}(\Psi_g)=u^{g} ,\\ &U_{s^-}^{b,c} (\Psi_g)= h_s^{1,b,c}+h_s^{2,b,c} ,\
U_{s^-}^{b,c}(\Psi_g)-h_s^{1,b,c}\geq\frac{B(I-N_s)}{r+\lambda_s^{I-N_s}},\ U_t^{b,c}(\Psi_g)=u^{b,c}.
\end{align*}

We will also consider in the sequel the following standard control problem, for any $(u^{b,c},u^g)\in\mathcal C_0$
\[
\widehat v^g_0(u^{b,c},u^g):=\sup_{\Psi_g\in\widehat{\mathcal{A}}^g(0,u^g,u^{b,c})}  p_g\E^{\P^{k^{\star,g}(\Psi_g)}}\bigg[\ds\int_0^\tau  \mu(I-N_s)\mathrm{d}s - \mathrm{d}D_s^g   \bigg].
\]
We abuse notations and also call elements of $\widehat{\Ac}^g(t,u^g,u^{b,c})$ contracts.

\subsubsection{Value function of the investor}
\label{sec:valueinvestor}
In this section, we characterise the value function of the investor when he offers only shutdown contracts. We will start by computing the value function on the boundaries of the credible set, before explaining how it can be characterised by a specific HJB equation in the interior of the credible set, under reasonable assumptions.

\paragraph{Value function of the investor on the lower boundary}

Recall the lower boundary with $j$ loans left
\[
\widehat{\mathfrak L}_j(u^{b,c})= \begin{cases}
u^{b,c}, \ c(j,1) \leq u^{b,c} \leq C(j), \\
\displaystyle \frac{\rho_g}{\rho_b}u^{b,c}-\frac{(\rho_g-\rho_b)}{\rho_b}C(j) , \ C(j)\leq u^{b,c}<\infty.
\end{cases}
\]
Consider any starting time $t\geq 0$. For $u^{b,c}\in{\mathcal C}_{t}$,  we denote by $V^{\mathfrak L,g}(u^{b,c})$ the value function of the investor on the lower boundary, that is
\begin{equation}\label{eq:lb investor problem}
V^{\mathfrak L,g}_t(u^{b,c}) := \underset{\Psi_g\in\widehat{\mathcal A}^{g}(t,\widehat{\mathfrak L}_{I-N_t}(u^{b,c}),u^{b,c})}{\rm ess\; sup} \E^{\P^{k^{\star,g}(\Psi_g)}} \bigg[\int_t^\tau \big(\mu(I-N_s)\mathrm{d}s-\mathrm{d}D_s^g\big) \bigg|\mathcal G_t \bigg].
\end{equation}

The following two propositions are proved in Appendix \ref{sec:E} and give explicitly the value of $V_t^{\mathfrak L,g}(u^{b,c})$.
\begin{Proposition} \label{prop:value function lb1}
For every $u^{b,c}\in\mathcal C_t$, if $u^{b,c}\geq C(I-N_t)$ then the value function of the investor on the lower boundary is given by
\[
V^{\mathfrak L,g}_t(u^{b,c}) = \ds\sum_{i=N_t}^{I-1} \frac{\mu(I-i)}{\widehat\lambda_{I-i}^{\rm SH}} - \bigg( \frac{ u^{b,c} - C(I-N_t) }{\rho_b} \bigg).
\]
\end{Proposition}

\begin{Proposition}\label{prop:value function lb2}
Fix some $t\geq 0$. For every $u^{b,c}\in{\mathcal C}_{t}$, with $c(I-N_t,1)\leq u^{b,c} < C(I-N_t)$, let $\nu(u^{b,c})$ be the unique solution of the following equation in $\nu$
\[
\bigg(\frac{B(I-N_t)}{r+\widehat\lambda_{I-N_t}^{\rm SH}} - u^{b,c}\bigg) + \ds\sum_{i=N_t+1}^{I-1} \int_{s_i(\nu)}^{\infty} \bigg( \frac{B(I-i)}{r+\widehat\lambda_{I-i}^{\rm SH}}\mathrm{e}^{-rx}  \bigg)f_{\tau^i}(x)\mathrm{d}x = 0,
\]
where $f_{\tau^i}$ is the density of the law of $\tau^i$ under $\P^{k^{\rm SH}}$ and where
\[
s_i(\nu) := \begin{cases}
\displaystyle0 ,\ \nu\leq \frac{\mu(r+\widehat\lambda_{I-i}^{\rm SH})}{B\widehat\lambda_{I-i}^{\rm SH}}, \\
\displaystyle\frac{1}{r}\ln\bigg(\frac{\nu B \widehat\lambda_{I-i}^{\rm SH}}{\mu(r+\widehat\lambda_{I-i}^{\rm SH})}\bigg) ,\;  \nu\geq \frac{\mu(r+\widehat\lambda_{I-i}^{\rm SH})}{B\widehat\lambda_{I-i}^{\rm SH}}.
\end{cases}
\]
Then the value function of the investor in the lower boundary is given by
\[
V_t^{\mathfrak L,g}(u^{b,c}) = \frac{\mu(I-N_t)}{\widehat\lambda_{I-N_t}^{\rm SH}} + \ds\sum_{i=N_t+1}^{I-1}  \int_{s_i(\nu(u^{b,c}))}^\infty \frac{\mu(I-i)}{\widehat\lambda_{I-i}^{\rm SH}}  f_{\tau^i}(x)\mathrm{d}x.
\]
\end{Proposition}

\begin{Remark}
Observe that the function $V_t^{\mathfrak{L},g}$ computed in Propositions \ref{prop:value function lb1} and \ref{prop:value function lb2} depends on $t$ only through the quantity $I-N_t$. Define, for any $j=1,\dots,J$ the map 
\[
\widehat V_j^{\mathfrak L,g}(u^{b,c}) := \left\{ \begin{array}{l} 
\displaystyle  \sum_{i=1}^j \frac{\mu i}{\widehat \lambda_i^{\rm SH} } - \left( \frac{u^{b,c}- C(j)}{\rho_b} \right),~  u^{b,c} \geq C(j), \\
 \displaystyle \frac{\mu j}{\widehat \lambda_j^{\rm SH} }  + \sum_{i=1}^{j-1}  \int_{s_{I-j}(\nu(u^{b,c}))}^\infty \frac{\mu i}{\widehat\lambda_i^{\rm SH}}  f_{\tau^{I-i}}(x)\mathrm{d}x,~  u^{b,c} \in \left(c(j,1), C(j)\right).
\end{array} \right.
\]
We have therefore, that $V^{\mathfrak L,g}_t(u^{b,c})=\widehat V^{\mathfrak L,g}_{I-N_t}(u^{b,c})$.
\end{Remark} 

\paragraph{Value function of the investor on the upper boundary}

The next proposition states that the upper boundary of the credible set is absorbing in the following sense: if under any contract the pair of value functions of the banks reaches the upper boundary at some time, the pair will stay on the upper boundary until the pool is liquidated. 

\begin{Proposition}\label{ub absorbing}
Fix a triplet $(t,u^g,u^{b,c})\in[0,+\infty)\times\mathcal{C}_{t}$ such that $u^g = \widehat{\mathfrak{U}}_{I-N_t}(u^{b,c})$. For any contract $\Psi_g=(D^g,\theta^g,h^{1,g},h^{2,g},h^{1,b,c},h^{2,b,c})\in\widehat{\mathcal A}^{g}(t,u^g,u^{b,c})$, we have $U_s^g(\Psi_g)=\widehat{\mathfrak U}_{I-N_s}(U_s^{b,c}(\Psi_g))$ for every $s\in[t,\tau)$. 
\end{Proposition}

\vspace{0.5em}
The next proposition states an important property satisfied by the contracts which make the continuation utilities of the banks lie in the upper boundary of the credible set.

\begin{Proposition}\label{ub contract}
Fix a triplet $(t,u^g,u^{b,c})\in[0,+\infty)\times\mathcal{C}_{t}$ such that $u^g = \widehat{\mathfrak{U}}_{I-N_t}(u^{b,c})$. For any contract $\Psi_g=(D^g,\theta^g,h^{1,g},h^{2,g},h^{1,b,c},h^{2,b,c})\in\widehat{\mathcal A}^{g}(t,u^g,u^{b,c})$, we have
\begin{itemize}
\item[$(i)$] $\theta^g_s=0$ for every $s\in[t,\tau)$ such that $U_s^{b,c}(\Psi_g)<b_s$.
\item[$(ii)$] If $U_{s_0}^{b,c}(\Psi_g)\geq b_{s_0}$ for some $s_0\in[t,\tau)$ then $k_s^{\star,b,c}(\Psi_g)=0$ and $U_s^{b,c}(\Psi_g)\geq b_s$ for every $s\in[s_0,\tau)$.
\end{itemize}	
\end{Proposition}

\vspace{0.5em}
We are now ready to give the value function of the investor on the upper boundary of the credible set. In the last region of the upper boundary, in which both the good and the bad agent are monitoring all the loans, it coincides with the value function of the sub--problem\footnote{The authors only look at the contracts for which the agent performs the maximum effort, that is, monitors all the loans at every time.} studied in\cite{pages2012bank}, denoted by $v_j^b$. For the sake of presentation, we recall the results of \cite{pages2012bank} in Appendix \ref{app:dylan-previous}.

\begin{Proposition}\label{prop:value function ub}
Under Assumption \ref{assump}, we have that for any $t\geq 0$ and any $u^{b,c}\in\widehat{\mathcal V}_{I-N_t}$, the value function of the investor on the upper boundary, defined by
\begin{equation}\label{eq:ub investor problem}
V^{\mathfrak U,g}_t(u^{b,c}) := \underset{\Psi_g\in\widehat{\mathcal A}^{g}(t,\widehat{\mathfrak U}_{I-N_t}(u^{b,c}),u^{b,c})}{\rm ess\; sup} \E^{\P^{k^{\star,g}(\Psi_g)}} \bigg[\int_t^\tau \big(\mu(I-N_s)\mathrm{d}s-\mathrm{d}D_s^g\big) \bigg|\mathcal G_t \bigg],
\end{equation}
verifies $V^{\mathfrak U,g}_t(u^{b,c})=\widehat V^{\mathfrak U,g}_{I-N_t}(u^{b,c})$, where for any $j=1,\cdots, I$
\[
\widehat V^{\mathfrak U,g}_j(u^{b,c}): = \begin{cases}
\displaystyle \frac{\mu j}{\widehat\lambda_{j}^{\rm SH}} + \widehat{C}^j \bigg( u^{b,c} - \frac{Bj}{r + \widehat\lambda_{j}^{\rm SH}} \bigg)^\frac{\widehat\lambda_{j}^{\rm SH}}{ r + \widehat\lambda_{j}^{\rm SH} }, \; u^{b,c}<x_j^\star,  \\[.4cm]
\displaystyle\frac{\mu j}{\widehat\lambda_j^0} + \bigg( v_j^b(\widehat b_j) - \frac{\mu j}{\widehat\lambda_j^0} \bigg) \bigg( \widehat b_j\frac{r+\widehat\lambda_j^0}{r+\widehat\lambda_j^{\rm SH}} \bigg)^{-\frac{\widehat \lambda_j^0}{r+\widehat \lambda_j^{\rm SH}}} \bigg(  u^{b,c} - \frac{B j}{ r + \widehat\lambda_j^{\rm SH}} \bigg)^\frac{\widehat\lambda_j^0}{ r + \widehat\lambda_j^{\rm SH} },\; x^\star_j \leq u^{b,c} <\widehat b_j , \\[.4cm]
v_j^b(u^{b,c}) , \; u^{b,c} \geq \widehat b_j,
\end{cases}
\]
with $v^b_j$ given by \eqref{eq:vj} and
\[
\widehat{C}^j :=\bigg(   \frac{\mu j}{\widehat\lambda_j^0} - \frac{\mu j}{\widehat\lambda_j^{\rm SH}} +   \bigg( \frac{\rho_b}{\rho_g}  \bigg)^\frac{\widehat\lambda_j^0}{ r + \widehat\lambda_j^0 } \bigg( v_j^b(\widehat b_j) - \frac{\mu j}{\widehat\lambda_j^0}  \bigg) \bigg) \bigg( \frac{\rho_b}{\rho_g} \bigg)^{- \frac{\widehat\lambda_j^{\rm SH}}{ r + \widehat\lambda_j^0 } } \bigg(  \frac{ \widehat b_j (r+\widehat\lambda_j^0)}{r+\widehat\lambda_j^{\rm SH} } \bigg)^{- \frac{\widehat\lambda_j^{\rm SH}}{ r + \widehat\lambda_j^{\rm SH} } } .
\]
\end{Proposition}

\paragraph{Value function of the investor in the credible set}

We define, for any $t\geq 0$ and any $(u^{b,c},u^g)\in\widehat{\mathcal C}_{I-N_t}$, the value function of the investor in the credible set by
\begin{equation}\label{eq:investor problem}
V^{g}_t(u^{b,c},u^g) := \underset{\Psi_g\in\widehat{\mathcal A}^{g}(t,u^g,u^{b,c})}{\rm ess\; sup} \E^{\P^{k^{\star,g}(\Psi_g)}} \bigg[\int_t^\tau \big(\mu(I-N_s)ds-\mathrm{d}D_s^g\big) \bigg|\mathcal G_t \bigg].
\end{equation}

The system of HJB equations associated to this control problem is given by $\widehat V^g_0\equiv 0$, and for any $1\leq j\leq I$, on $\widehat\Vc_j\times\widehat\Vc_j$
\begin{align}\label{eq:hjb value function}
\max\left\{
 \sup_{\overline C^j} 
\left\{ 
\begin{array}{l} 
\partial _{u^{b,c}}\widehat V^g_j \big( ru^{b,c} - Bk^{b,c} + (h^{1,b,c} + (1-\theta)h^{2,b,c}) \widehat\lambda_j^{k^{b,c}} \big) \\[.3cm]
+ \partial _{u^{g}}\widehat V^g_j \big( ru^g - Bk^g + (h^{1,g} + (1-\theta)h^{2,g}) \widehat\lambda_j^{k^g} \big) \\[.3cm] 
+ \big(\widehat V^g_{j-1}(u^{b,c} - h^{1,b,c},u^g - h^{1,g}) - \widehat V^g_j(u^{b,c},u^g)\big) \widehat\lambda_j^{k^g} \\[.3cm] 
- \widehat V^g_{j-1}(u^{b,c}-h^{1,b,c},u^g - h^{1,g})(1-\theta)\widehat\lambda_j^{k^g} + \mu j   
\end{array} 
\right\}  , -\rho_b \partial _{u^{b,c}}\widehat V^g_j-\rho_g \partial _{u^{g}}\widehat V^g_j - 1  
\right\}
=0,
\end{align}
where we defined $k^{b,c}=j\cdot 1_{\{h^{1,b,c}+(1-\theta)h^{2,b,c}<\widehat b_j\}},~ k^g=j\cdot 1_{\{h^{1,g}+(1-\theta)h^{2,g}<\widehat b_j\}}$ and the set of constraints 
\[
\overline C^j=\bigg\{ (\theta,h^{1,b,c},h^{2,b,c},h^{1,g},h^{2,g})\in\R_+^5: \theta\in[0,1], u^g = h^{1,g}+h^{2,g}, u^{b,c}=h^{1,b,c}+h^{2,b,c}, h^{2,g}; h^{2,b,c}\geq\frac{B(j-1)}{r+\widehat\lambda_{j-1}^{\rm SH}}  \bigg\}.
\]
The boundary conditions of \eqref{eq:hjb value function} are given, for every $u^{b,c}\in \widehat{\mathcal{V}}_j,$ by
\begin{equation}\label{eq:boundary-conditions-PDE}
 \widehat V^g_j(u^{b,c},\widehat{\mathfrak U}_j(u^{b,c})) = \widehat V^{\mathcal U,g}_{j}(u^{b,c}), \; \widehat V^g_j(u^{b,c},\widehat{\mathfrak L}_j(u^{b,c})) = \widehat V^{\mathcal L,g}_{j}(u^{b,c}).
\end{equation}
The last step is now to use classical arguments to prove that $\widehat V^g_j$ is a viscosity solution of the above PDE for every $j=1,\dots,I$ and that the functions are sufficiently smooth (at least weakly dfferentiable) in order to obtain the optimal contract as the maximisers above.
This program can in principle be carried out using standard arguments in viscosity theory of Hamilton--Jacobi equations. However, given the length of the paper, we believe that it would not serve a specific purpose and decided to just describe the main steps that lead to this result. We list them below

\begin{itemize}[leftmargin=*]

\item[$(i)$] For $j=1$, we can use the abstract results of \cite{karoui2013capacities2} to prove that \eqref{eq:investor problem} coincides with the strong formulation of itself. This fact allows to prove directly that the value function $\widehat{V}_1^g$ is concave and therefore differentiable almost everywhere.

\item[$(ii)$] For $j>1$, let us define the penalised Hamiltonians for the diffusion equation, with $j$ loans left. Given $\widehat V^g_{j-1}$, define for instance
\begin{equation} \label{eq:penalizedH}
H_j^n(u^{b,c},u^g,v,p^{b,c},p^g) =  \sup_{\overline C^j} 
\left\{ 
\begin{array}{l} 
p^{b,c} \big( ru^{b,c} - Bk^{b,c} + (h^{1,b,c} + (1-\theta)h^{2,b,c}) \lambda_j^{k^{b,c}} \big) \\[.3cm]
+p^g \big( ru^g - Bk^g + (h^{1,g} + (1-\theta)h^{2,g}) \widehat\lambda_j^{k^g} \big) \\[.3cm] 
+ \big(\widehat V^g_{j-1}(u^{b,c} - h^{1,b,c},u^g - h^{1,g}) - v(u^{b,c},u^g)\big) \widehat\lambda_j^{k^g} \\[.3cm] 
- \widehat V^g_{j-1}(u^{b,c}-h^{1,b,c},u^g - h^{1,g})(1-\theta)\widehat\lambda_j^{k^g} + \mu j   
\end{array} 
\right\}    + n(-\rho_b p^{b,c} - \rho_g p^g - 1)^+. 
\end{equation}
Let $v_j^n$ be the value function of the penalised version of our problem, in which payments are absolutely continuous with bounded density. Then it can be argued as in \cite{elie2017on} that $v_j^n$ is a viscosity solution to $H^n_j(u,v,p)=0$, 
with appropriate credible set and boundary conditions.

\item[$(iii)$] Note that $H^n_j$ is convex in $p$, as a supremum of linear functionals and composition of convex functions. Moreover, for any $R<+\infty$ we have
\[
H^n_j(u,v_1,p) - H^n_j(u,v_2,p) \geq -\widehat \lambda_j^{\rm SH}(v_1 - v_2), \; \forall (u,p) \text{ and } R \geq v_1\geq v_2\geq -R.
\]
$H_j^n$ is also locally Lipschitz and $H_j^n(u,v,p)\longrightarrow\infty$ as $|p|\to\infty$ for any $u>\frac{B j}{r+\widehat\lambda_j^{\rm SH}}$. Finally, noticing that interior maximisers take place in the interior of the credible set (to be more precise, boundary maximisers correspond to contracts leading the agents to the absorbing boundaries of the credible set) and by using the envelope theorem, we can show that $H_j^n$ is actually strictly convex on the interior of the credible set and therefore satisfies 
\[
\forall R>0, \exists \alpha_R >0, ~ 
\left( \frac{\partial H_j^n}{\partial p}(u,v,p) - \frac{\partial H_j^n}{\partial p}(u,v,q), p-q \right) \geq \alpha_R |p-q|^2,
~ |p|,|q|,|u|\leq R, \text{ for any }u. 
\]
By Theorem 3.3 in Lions \cite{lions1982generalized}, it follows then that $v_j^n\in W^{1,\infty}_{loc}$ and $v_j^n$ is SSH (semi--super harmonic).

\item[$(iv)$] Arguing again as in \cite{elie2017on}, it can be proved the sequence $v_j^n$ converges to the value function $\widehat V_j^g$ of our problem, which is therefore SSH 
and a viscosity solution to \eqref{eq:hjb value function} with boundary condition \eqref{eq:boundary-conditions-PDE}.
\end{itemize}

Finally, since $\widehat{V}_j^g$ is differentiable almost everywhere, we can define the optimal contract through the maximisers in the Hamiltonian \eqref{eq:hjb value function}. Then, using the classical result (see for instance \cite{hynd2012eigenvalue} for related arguments) that the domain in which the diffusion equation is not saturated is bounded, it follows that the optimal controls $(h^{1,g,\star},h^{2,g,\star},h^{1,b,c,\star},h^{2,b,c,\star})$ are bounded and the corresponding SDEs admit weak solutions
\begin{align*}
\mathrm{d}U_s^{\star,g} =&\ \big(rU_s^{\star,g}-Bk_s^{\star,g}(U^{\star,b,c}_s,U^{\star,g}_s)-\rho_g \delta_{I-N_s}^{\star,g}(U^{\star,b,c}_s,U^{\star,g}_s)\big)\mathrm{d}s-h_{I-N_s}^{\star,1,g}(U^{\star,b,c}_s,U^{\star,g}_s)\big(\mathrm{d}N_s-\lambda_s^{k^{\star,g}((U^{\star,b,c},U^{\star,g}))}\mathrm{d}s\big)\\
&-h_{I-N_s}^{\star,2,g}(U^{\star,b,c}_s,U^{\star,g}_s)\big(\mathrm{d}H_s-(1-\theta_{I-N_s}^{\star,g}(U^{\star,b,c}_s,U^{\star,g}_s))\lambda_s^{k^{\star,g}((U^{\star,b,c},U^{\star,g}))}\mathrm{d}s\big),\\
\mathrm{d}U_s^{\star,b,c}=&\ \big(rU_s^{\star,b,c}-Bk_s^{\star,b,c}(U^{\star,b,c}_s,U^{\star,g}_s)-\rho_b \delta_{I-N_s}^{\star,g}(U^{\star,b,c}_s,U^{\star,g}_s)\big)\mathrm{d}s-h_{I-N_s}^{\star,1,b,c}(U^{\star,b,c}_s,U^{\star,g}_s)\big(\mathrm{d}N_s-\lambda_s^{k^{\star,b,c}((U^{\star,b,c},U^{\star,g}))}\mathrm{d}s\big)\\
&-h_{I-N_s}^{\star,2,b,c}(U^{\star,b,c}_s,U^{\star,g}_s)\big(\mathrm{d}H_s-(1-\theta_s^{\star,g}(U^{\star,b,c}_s,U^{\star,g}_s))\lambda_s^{k^{\star,b,c}((U^{\star,b,c},U^{\star,g}))}\mathrm{d}s\big).
\end{align*}
Indeed, this can be proved by noticing that in--between two jump times, the above are actually first--order ODEs, which admit weak solutions in appropriately exponentially weighted $L^1$ space (to make sure that bounded functions are integrable over the credible set), thanks to Carath\'eodory's theorem for ODEs. Thus, we have the equivalence
\[
v_0^{\rm Shut}=\underset{u^{b,c}\leq R_0^b,~ u^g\geq R_0^g}{\sup}\widehat v_{0}^g(u^{b,c},u^g)=\underset{u^{b,c}\leq R_0^b,~ u^g\geq R_0^g}{\sup} p_g \widehat V_{I}^g(u^{b,c},u^g).
\]

\subsection{Screening contract}

Recall that in the screening contract the investor designs a menu of contracts, one for each agent, and his expected utility is given by
\begin{equation}
v_{0}\big((\Psi_i)_{i\in\{g,b\}} \big)=\sum_{i\in\{g,b\}}p_i\mathbb{E}^{\mathbb{P}^{k^i}}\bigg[ \int_{0}^{\tau
}\big( I-N_s\big) \mu \mathrm{d}s-\mathrm{d}D^i_s\bigg] .  \label{investors utility eq-2}
\end{equation}
In this case, we will have to keep track of the value functions of both banks, when they choose the contract designed for them, as well as when they do not truthfully reveal their type. We will denote by $v_0$ the maximal utility that the investor can get out of the screening contract.
\[ v_0:=\underset{(\theta^g,\theta^b,D^g,D^b)\in{\Ac}_{\rm Scr}}{\sup}p_g\E^{\P^{k^{\star,g}(\theta^g,D^g)}}\bigg[\int_0^\tau\mu(I-N_s)\mathrm{d}s-\mathrm{d}D^g_s\bigg] + p_b\E^{\P^{k^{\star,b}(\theta^b,D^b)}}\bigg[\int_0^\tau\mu(I-N_s)\mathrm{d}s-\mathrm{d}D^b_s\bigg],\]
where
\[{\Ac}_{\rm Scr}:=\Big\{(\theta^g,\theta^b,D^g,D^b)\in\Theta^2\times\Dc^2: U_{0}^i(\theta^i,D^i)\geq R_0^i, U_{0}^j(\theta^j,D^j)\geq U_{0}^{j,c}(\theta^i,D^i), (i,j)\in\{g,b\}^2,\; i\neq j \Big\}.\]


Different from the study of the shutdown contract, where the investor contracts only the good bank, in order to obtain the optimal screening contract we need to characterise also the value function of the investor when he contracts the bad bank. We will therefore follow Section \ref{sec:valueinvestor}, but by replacing the good bank by the bad bank. Hence, we define similarly, for any $(t,u^b,u^{g,c})\in[0,+\infty)\times\mathcal C_{t}$ the set $\widehat{\mathcal A}^{b}(t,u^{g,c},u^b)$.
We also introduce the following stochastic control problem for any $(u^{b},u^{g,c})\in\mathcal C_I$
\[
\widehat v^{\rm b}_0(u^{b},u^{g,c}):=\sup_{\Psi_b\in\widehat{\mathcal{A}}^b(0,u^{g,c},u^{b})}  p_b\E^{\P^{k^{\star,b}(\Psi_b)}}\bigg[ \ds\int_0^\tau  \mu(I-N_s)\mathrm{d}s - \mathrm{d}D_s^b   \bigg].
\]
The aim of the next sections is to compute the function $\widehat v^{\rm b}_0(u^{g,c},u^{b})$, representing the utility of the investor when hiring the bad bank. We start by studying it on the boundary of the credible set.

\subsubsection{Boundary study} 
We denote by $V^{\mathfrak L,b}(u^{g,c})$ the value function of the investor in the lower boundary, when hiring the bad bank, defined by
\begin{equation} \label{eq:lb investor problem-b}
V^{\mathfrak L,b}_t(u^{b}) := \underset{\Psi_b\in\widehat{\mathcal A}^{b}(t,\widehat{\mathfrak L}_{I-N_t}(u^{b}),u^{b})}{\rm ess\; sup} \E^{\P^{k^{\star,b}(\Psi_b)}} \bigg[\int_t^\tau \mu(I-N_s)\mathrm{d}s-\mathrm{d}D_s^b \bigg|\mathcal G_t \bigg].
\end{equation}
The first result is that the value function of the investor on the lower boundary of the credible set is the same when hiring either the bad or the good bank. This is mainly due to the fact that both banks shirk on the lower boundary.
\begin{Proposition} \label{prop:value function lower boundary-b}
For every $u^b\in\mathcal C_{I-N_t}$, we have $ V^{\mathfrak L,b}_t(u^b) = V^{\mathfrak L,g}_t(u^b)$.
\end{Proposition}

Let us now consider the upper boundary. We denote by $V^{\mathfrak U,b}(u^b) $ the value function of the investor on the upper boundary when hiring the bad agent.
\begin{equation}\label{eq:ub investor problem-b}
V^{\mathfrak U,b}_t(u^b) := \underset{\Psi_b\in\widehat{\mathcal A}^{b}(t,\widehat{\mathfrak U}_{I-N_t}(u^b),u^b)}{\rm ess\; sup} \E^{\P^{k^{\star,b}(\Psi_b)}} \bigg[\int_t^\tau \mu(I-N_s)\mathrm{d}s-\mathrm{d}D_s^b \bigg|\mathcal G_t \bigg].
\end{equation}

We have the following result.

\begin{Proposition}\label{prop:value function ub-b}
Under Assumption \ref{assump}, for any $t\geq 0$ and any $u^b\in\widehat{\mathcal V}_{I-N_t}$, we have that $V^{\mathfrak U,b}_t(u^b)=\widehat V^{\mathfrak U,b}_{I-N_t}(u^b)$, where for any $j=1,\cdots, I$
\[
\widehat V^{\mathfrak U,b}_j(u^b): = \begin{cases}
\displaystyle \frac{\mu j}{\widehat\lambda_{j}^{\rm SH}} + \tilde{C}^j \bigg( u^b - \frac{Bj}{r+\widehat\lambda_{j}^{\rm SH}} \bigg)^\frac{\widehat\lambda_{j}^{\rm SH}}{ r + \widehat\lambda_{j}^{\rm SH} }, \ u^b<\widehat b_j, \\[.4cm]
v_j^b(u^b) , \ u^b \geq \widehat b_j,
\end{cases}
\]
with $v_j^b$ given by \eqref{eq:vj} and
\[
\tilde{C}^j = \bigg( v_j^b(\widehat b_j) - \frac{\mu j}{\widehat\lambda_j^{\rm SH}}  \bigg)	 \bigg(  \frac{ \widehat b_j (r+\widehat\lambda_j^{0})}{r+\widehat\lambda_j^{\rm SH} } \bigg)^{ \frac{- \widehat\lambda_j^{\rm SH}}{ r + \widehat\lambda_j^{\rm SH} } }.
\]
\end{Proposition}

\subsubsection{Study of the credible set} 
We define, for any $t\geq 0$ and any $(u^b,u^{g,c})\in\widehat{\mathcal C}_{I-N_t}$, the value function of the investor in the credible set when hiring the bad bank by
\begin{equation}\label{eq:investor problem-b}
V^{b}_t(u^b,u^{g,c}) := \underset{\Psi_b\in\widehat{\mathcal A}^{b}(t,u^{g,c},u^b)}{\rm ess\; sup} \E^{\P^{k^{\star,b}(\Psi_b)}} \left[\left.\int_t^\tau \left(\mu(I-N_s)ds-dD_s^b\right) \right|\mathcal G_t \right].
\end{equation}

The system of HJB equations associated to this control problem is given by $\widehat V^b_0\equiv 0$, and for any $1\leq j\leq I$
\begin{multline}\label{eq:hjb value function bad}
\max\left\{
 \sup_{\overline C^j} 
\left\{ 
\begin{array}{l} 
\partial _{u^b}\widehat V^b_j \big( ru^b - Bk^b + (h^{1,b} + (1-\theta)h^{2,b}) \widehat\lambda_j^{k^b} \big) \\[.3cm]
+ \partial _{u^{g,c}}\widehat V^b_j \big( ru^{g,c} - Bk^{g,c} + (h^{1,g,c} + (1-\theta)h^{2,g,c}) \widehat\lambda_j^{k^{g,c}} \big) \\[.3cm] 
+ \big(\widehat V^b_{j-1}(u^b - h^{1,b},u^{g,c} - h^{1,g,c}) - \widehat V^b_j\big) \widehat\lambda_j^{k^b} \\[.3cm] 
- \widehat V^b_{j-1}\big(u^b-h^{1,b},u^{g,c} - h^{1,g,c}\big)(1-\theta)\widehat\lambda_j^{k^b} + \mu j   
\end{array} 
\right\}  ,\; -\rho_b \partial _{u^b}\widehat V^b_j-\rho_g \partial _{u^{g,c}}\widehat V^b_j - 1  
\right\}
=0.
\end{multline}
With $k^b=j\cdot 1_{\{h^{1,b}+(1-\theta)h^{2,b}<\widehat b_j\}},~ k^{g,c}=j\cdot 1_{\{h^{1,g,c}+(1-\theta)h^{2,g,c}<\widehat b_j\}}$ and the same set of constraints $\overline C^j$ as in the system of HJB equations associated to the functions $\widehat V^g_j(u^{b,c},u^g)$. The boundary conditions of \eqref{eq:hjb value function bad} are given, for every $u^{b}\in \widehat{\mathcal{V}}_j$ by
\begin{align*}
 \widehat V^b_j(u^b,\widehat{\mathfrak U}_j(u^b)) = \widehat V^{\mathcal U,b}_{j}(u^b),\ \widehat V^b_j(u^b,\widehat{\mathfrak L}_j(u^b)) =  \widehat V^{\mathcal L,g}_{j}(u^b).
\end{align*}
Similarly to the shutdown contract, we can argue that the functions $\widehat{V_j^b} $ are viscosity solutions to the system \eqref{eq:hjb value function bad}, differentiable almost everywhere and the maximizers in the Hamiltonian define an admissible contract. This implies the equivalence 

\begin{align*}
v_0&=\underset{\{R_0^b\vee u^{b,c} \leq u^b, R_0^g\vee u^{g,c} \leq u^g\}}{\sup}\widehat v_{0}^g(u^{b,c},u^g)+\widehat v_{0}^{\rm b}(u^b,u^{g,c})=\underset{\{R_0^b\vee u^{b,c} \leq u^b, R_0^g\vee u^{g,c} \leq u^g\}}{\sup} p_g \widehat V_{ I}^g(u^{b,c},u^g)+ p_b \widehat V_{ I}^b(u^b,u^{g,c}).
\end{align*}

\subsection{Description of the optimal contracts}

In this section we describe the optimal contracts for the investor when he designs a contract for the good or the bad bank. We explain in detail the optimal contracts on the boundaries of the credible set, which can be obtained explicitly from the value function of the investor. In the interior of the credible set, we discuss the properties we expect the optimal contracts to have given the verification results described in the previous sections. 

\subsubsection{Optimal contracts on the boundaries of the credible set}

We start with the upper boundary of the credible set. The following result is a direct consequence of the proofs of Proposition \ref{prop:value function ub} and \ref{prop:value function ub-b}, and the optimal contract for the pure moral hazard case. In the last part of the upper boundary, in which both agents monitor all the loans, the optimal contract coincides with the one of the subproblem studied in \cite{pages2012bank}. Their main results can be found in Appendix \ref{app:dylan-previous}.

\begin{Proposition}\label{prop:optimal contracts upper boundary} Under Assumption \ref{assump}, consider for any $t\geq0$ and $(u^b,\mathfrak{U}_t(u^b))\in\mathcal{C}_{t}$ the process $(u_s^b)_{s\geq t}$ as the solution of the following {\rm SDE} on $[t,\tau)$
\begin{equation}\label{dynamic:ub upper boundary}
\mathrm{d}u^b_s=\big((ru^b_s - Bk_s^{b,\star} + \lambda_s^{k^{b,\star}}(h_s^{1,b,\star} + (1-\theta_s^\star)h_s^{2,b,\star})  \big)\mathrm{d}s -\rho_b \mathrm{d}D^{\star}_s - h_s^{1,b,\star}  \mathrm{d}N_s -h_s^{2,b,\star} \mathrm{d}H_s,
\end{equation}
with initial value $u^b$ at $t$, and with
\begin{align*}D^{\star}_s&:={\bf 1}_{\{s=t\}}\frac{(u^b-\gamma_{I-N_t}^b)^+}{\rho_b}+\int_t^s\delta^{I-N_r}(u^b_r)\mathrm{d}r,\ \theta^\star_s:=\theta^{I-N_s}(u^b_s),\ h_s^{1,b,\star} := h^{1,b,I-N_s}(u_s^b),\ k_s^{b,\star}:= k^{b,j}(u_s^b),
\end{align*}
for $s\in [t,\tau)$ and $j=1,\dots,I$, where $\gamma_j^b$ is given by \eqref{eq:gammaj} and
\begin{align*}
\delta^j(u)&:={\bf 1}_{\{u=\gamma_j^b\}}\frac{\widehat \lambda^0_j\widehat b_j+r\gamma^b_j}{\rho_i},\ \theta^j(u):={\bf 1}_{\{u\in[\widehat b_j,\widehat b_{j-1}+\widehat b_j)\}}\frac{u-\widehat b_j}{\widehat b_{j-1}}+{\bf 1}_{\{u\in[\widehat b_j+\widehat b_{j-1},\gamma_j^b)\}},\\
h^{1,b,j}(u)&:= {\bf 1}_{\{u \in[c(j,1), \widehat b_j)\}} u   + {\bf 1}_{\{u\in[\widehat b_j,\widehat b_{j-1}+\widehat b_j)\}}(u-\widehat b_{j-1})+{\bf 1}_{\{u\in[\widehat b_j+\widehat b_{j-1},\gamma_j^b)\}}\widehat b_j,\; k^{b,j}(u) = j{\bf 1}_{\{ \theta^j(u)h^{1,b,j}(u) + (1-\theta^j(u))u < \widehat b_j \}}.
\end{align*}
Then, the contract $\Psi^\star=(D^\star,\theta^\star,h^{1,b,\star},h^{2,b,\star})$ is the unique solution of problems \eqref{eq:ub investor problem} and \eqref{eq:ub investor problem-b}.
\end{Proposition}

Let us comment the optimal contract for the investor on the upper boundary of the credible set. It is the same if he designs a contract for the good or the bad bank. The state process $(u_s^b)_{s\geq t}$ defined by \eqref{dynamic:ub upper boundary} corresponds to the value function of the bad bank under the optimal contract. The optimal contract offers no payments to the banks when $u_s^b$ is smaller than $\gamma_{I-N_s}^b$. In this case the continuation utility of the bad bank is an increasing process and eventually reaches the value $\gamma_{I-N_s}^b$, if no default happens in the meantime. Payments are postponed until this moment. If the initial value for the bad agent $u^b$ is greater than $\gamma_{I-N_t}^b$, a lump-sum payment is made at $t^-$ in order to have $u_{t}^b=\gamma_{I-N_t}^b$. When $u_s^b=\gamma_{I-N_s}^b$, the banks receive constant payments which keep the value function of the bad bank constant at this level. Concerning the liquidation of the project, if, at the default time $\tau^j$, it holds that $u_{\tau^j}^b < \widehat b_j$, then  the project is liquidated. In case of $u_{\tau^j}^b \in[\widehat b_j + \widehat b_{j-1}, \gamma_j^b)$, the project will continue with probability $\theta_j\in(0,1)$ which will be closer to one as $u_{\tau^j}^b$ gets closer to $\gamma_j^b$. If $u_{\tau^j}^b \geq \gamma_j^b$, the project will be maintained. Finally, the bad bank will monitor all the loans only when her value function is greater than $\widehat b_{I-N_s}$, whereas the good bank will monitor when the value of the bad bank is greater than $x^\star_{I-N_s}$. Figure \ref{figure:optimal contract upper boundary} depicts the optimal contract of the investor on the upper boundary of the credible set, denoting $\widehat B_j:=\widehat b_j + \widehat b_{j-1}$.

\begin{figure}[H]
\centering
\includegraphics[scale=1]{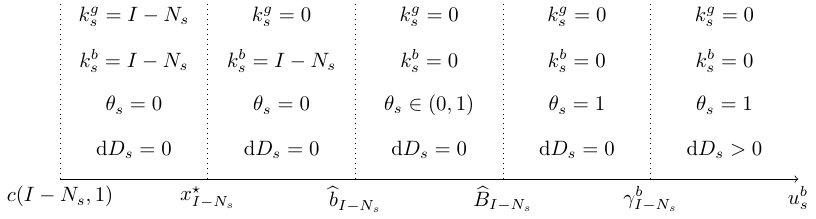}
\caption{Optimal contract on the upper boundary.}\label{figure:optimal contract upper boundary}
\end{figure}

For the lower boundary of the credible set, we have the following result.

\begin{Proposition}\label{prop:optimal contracts lower boundary} Under Assumption \ref{assump}, consider for any $t\geq0$ and $(u^b,\mathfrak{L}_t(u^b))\in\mathcal{C}_{t}$ the process $(u_s^b)_{s\geq t}$ as the solution of the following {\rm SDE} on $[t,\tau)$ 
\begin{equation}\label{dynamic:ub lower boundary}
\mathrm{d}u^b_s=\big((ru^b_s - Bk_s^{b,\star} + \lambda_s^{k^{b,\star}}(h_s^{1,b,\star} +(1-\theta_s^\star)h_s^{2,b,\star})  \big)\mathrm{d}s -\rho_b \mathrm{d}D^{\star}_s - h_s^{1,b,\star}  \mathrm{d}N_s -h_s^{2,b,\star} \mathrm{d}H_s,
\end{equation}
with initial value $u^b$ at $t$, and with
\[
D^{\star}_s:={\bf 1}_{\{s=t\}}\frac{(u^b-C(I-N_s))^+}{\rho_b},\ \theta^\star_s:={\bf 1}_{\{ u_s^b \geq C(I-N_s) \}},\]

\[
h_s^{1,b,\star} := u_s^b - C(I-N_s-1){\bf 1}_{\{ u_s^b \geq C(I-N_s)\}},\ h_s^{2,b,\star} := C(I-N_s-1){\bf 1}_{\{u_s^b \geq C(I-N_s)\}},\ k_s^{b,\star} = (I-N_s){\bf 1}_{\{ h_s^{1,b,\star}+(1-\theta_s^\star)h_s^{2,b,\star} < b_s \}},
\]
for $s\in [t,\tau)$. Then, the contract $\Psi^\star=(D^\star,\theta^\star,h^{1,b,\star},h^{2,b,\star})$ is the unique solution of \eqref{eq:lb investor problem} and \eqref{eq:lb investor problem-b}.
\end{Proposition}

\vspace{0.5em}
On the lower boundary of the credible set, the optimal contract for the investor also does not depend on the type of the bank. If the initial value of the bad bank $u^b$ is greater than $C(I-N_t)$, the banks receive a lump-sum payment such that $u^b_{t}=C(I-N_t)$. This is the only payment offered by the contract. If there is a default at some time $s$ such that $u_s^b<C(I-N_s)$, the project is liquidated. When $u_s^b=C(I-N_s)$ the contract maintains the project until the last default. Since the optimal contract does not provides incentives to the banks to monitor the loans, the good and the bad bank shirk until the liquidation of the project. Figure \ref{figure:optimal contract lower boundary} depicts the optimal contract of the investor on the lower boundary of the credible set. 

\begin{figure}[ht!]
\centering
\includegraphics[scale=0.8]{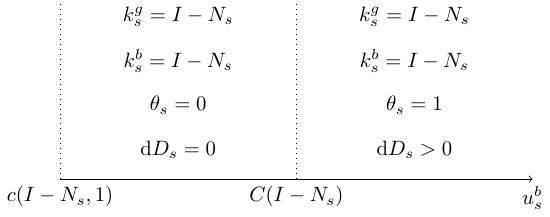}
\caption{Optimal contract on the lower boundary.}\label{figure:optimal contract lower boundary}
\end{figure}

\subsubsection{Discussion about the optimal contracts in the interior of the credible set}

Figure \ref{fig:discussion boundaries} represents the optimal contracts on the boundaries of the credible set as well as the movements of the values of the banks along these curves. The green zone corresponds to the region where the contract offers payments to the agents and the project is maintained if there is a default. The red zone corresponds to the region where there are no payments and the project is liquidated immediately after a default. Intermediate situations correspond to the yellow zone. We remark that the banks are paid only on the green zone.

\begin{figure}[!ht]
\centering
\includegraphics[scale=0.8]{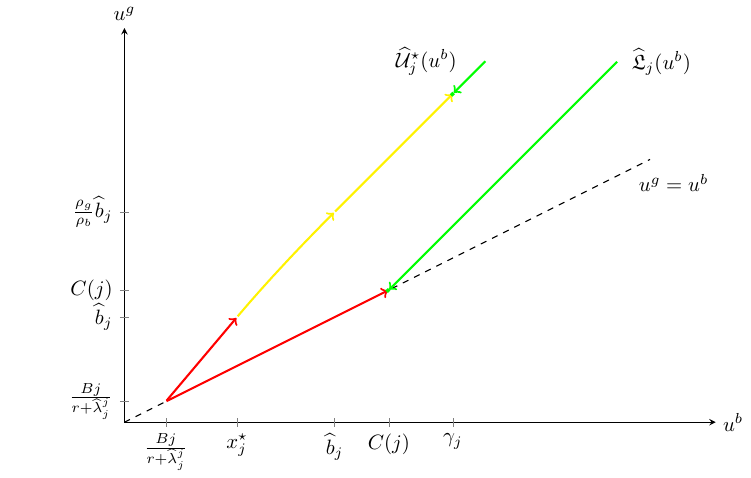}
\caption{Optimal contract on the boundaries of the credible set.}  \label{fig:discussion boundaries}
\end{figure}

\vspace{.5em}
Let us now consider the whole credible set and explain how the green and red zones on the boundaries propagate towards the interior region, 
given that the optimal contracts for problems \eqref{eq:investor problem} and \eqref{eq:investor problem-b} correspond to the maximisers in the Hamiltonian of the systems \eqref{eq:hjb value function} and \eqref{eq:hjb value function bad}. 
Recall that payments only take place when the value function of the investor saturates the gradient constraint. Therefore, if at some point of the credible set the banks are paid, this will also be the case under movements in the direction $(\rho_b,\rho_g)$. The interpretation of this property is that the green region, where the banks are paid and the project is maintained after a default, is formed by the points where the banks have a good performance and they are rewarded. A movement in the direction $(\rho_b,\rho_g)$ correspond to a better performance of both banks, so it seems unnatural to deprive them of the reward. We can do the opposite interpretation for the red region, consisting of the points where the banks receive no payments and the project is liquidated after a default. In consequence, under the optimal contracts, it is possible to identify red and green areas in the credible set, where the characteristics described in the boundaries will remain, and that will be delimited by some curves similar to those shown in Figure \ref{fig:discussion credible set} below. Mathematically, these curves are delimiting the region where the gradient constraint is saturated.

\begin{figure}[!ht]
\centering
\includegraphics[scale=0.9]{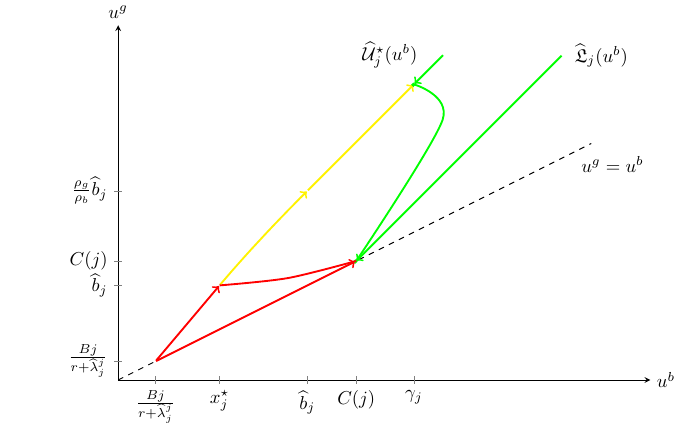}
\caption{Optimal contract on the credible set.}  \label{fig:discussion credible set}
\end{figure}

\subsection{A word on implementability of the contracts}

Any real--world application of our model requires to discuss the practical implementability of the contract. Fortunately for us, the form of the menus of contracts we obtained is completely similar to the one obtained in \cite{pages2012bank,pages2014mathematical}, in the sense that all rely on a probation zone, where stochastic liquidation may occur, and a zone of good performance, where the liquidation never occurs. The only difference is of course that in \cite{pages2012bank,pages2014mathematical} these zones are simply intervals, while they are more complex regions of the plan in our case, since we have to keep track of both the continuation and the temptation values of the Agent. Nonetheless, the practical implementation proposed by Pag\`es \cite[Proposition 7]{pages2012bank} can readily be adapted to our context. Given the length of the paper, we leave the exact detail to the reader, and simply recall how the implementation works.

\vspace{0.5em}
First, a natural way of implementing the contract is to replicate dynamically both the continuation and the temptation values of the Agent by use of two cash reserve accounts. The accounts should be managed by an independent trust, and actually serves to both provide protection to the investors, and to manage exactly the performance--based compensation scheme described in the optimal contract. The current balances reveal outright performance of both type of banks, and can be used to determine the amount and timing of fees that are released. Then, the implementation basically takes the form of a whole loan sale with monitoring retained. The reserve accounts then offer protection in the form of ABS credit default swaps (ABCDS), and serve as instruments to tie the amount and timing of compensation to performance (meaning that payments are made from the cash reserve only when the continuation and temptation values of the Agent are in the domain where the gradient constraint is satisfied). The reserve account reveals the level of underlying performance, which reduces the rent of the monitoring bank and allows it to retain risk at a lower cost than if it were funded with deposits.

\section{Extensions}
\subsection{Endogenous reservation utility}

In a standard Principal--Agent problem, it is assumed that the Agent possesses a minimum level of utility that must be provided by the Principal in order to make him accept the contract. This reservation value represents the utility that the Agent would obtain if the contract offered by the Principal was not sufficiently attractive and he made use of an outside option (see Condition \reff{const-reser} in Section \ref{sec-ppalproblem}).

\vspace{0.5em}
In this section we provide an endogenous characterisation of the reservation utilities of the banks by assuming that if they do not enter in a contractual relationship with the investor, they can manage the project by themselves. When the outside option of a bank is to manage the pool of loans on its own, we can find the explicit value of its reservation utility. Moreover, we outline an extension of our model to the case in which the bank can break the contract at any time if it can do better by itself. Different from the \emph{full-commitment} problem studied in the previous sections, the ability of the bank to break the contract makes the investor offer only the so called \emph{renegotiation-proof} contracts, which keep the utility of bank above a dynamic reservation utility until the end of the contract. 

\vspace{0.5em}
If the bank of type $\rho_i$ manages the project, it receives the cash flows from the loans and does not face the threat of liquidating the whole pool when one of the loans defaults. Consequently, its reservation utility $R_0^i$ is given by the following expression
\begin{equation} \label{eq:endogenous-reservation-utility}
R^i_0 = \sup_{k\in\mathfrak K} \mathbb E^{\P^k}\bigg[\int_0^{\tau^I} \mathrm{e}^{-rs}(\rho_i\mu(I-N_s)+Bk_s) \mathrm{d}s \bigg].
\end{equation}
The value of $R^i_0$ can be obtained as an application of the results from the previous sections, since \eqref{eq:endogenous-reservation-utility} corresponds to the utility of the bank under a contract with no liquidation at all, $\theta\equiv 1$, and with absolutely continuous payments, $\mathrm{d}D_t = \mu(I-N_t)\mathrm{d}t$. Its explicit value is provided in Proposition \ref{prop:reservationutilities}.

\begin{Proposition}\label{prop:reservationutilities}
Define the recursive sequence of numbers  
\[
\widehat{R}^i_j = \max\bigg\{ \frac{\rho_i\mu j}{r+\widehat\lambda_j^0} - \frac{r \widehat{R}^i_{j-1}}{r+\widehat\lambda_j^0} , ~ \frac{\rho_i\mu j+jB}{r+\widehat\lambda_j^{SH}} - \frac{r \widehat{R}^i_{j-1}}{r+\widehat\lambda_j^{\rm SH}}  \bigg\}, \; j\in\{1,\dots,I\},
\]
with $\widehat{R}^i_0=0$. The endogenous reservation utility of the bank of type $\rho_i$ is given by $R_0^i=\widehat{R}_I^i$. Moreover, the optimal action in Problem \eqref{eq:endogenous-reservation-utility} is constant in every interval $(\tau^{I-j},\tau^{I-j+1})$ and it is equal to $k^{\star,i}\equiv 0$ if the maximum in the definition of $\widehat{R}^i_j$  is attained at the first term, and to $k^{\star,i}\equiv j$ if the maximum is attained at the last term.
\end{Proposition}

\subsection{Renegotiation--proof contracts}\label{sec:}
Suppose that the bank of type $\rho_i$ can decide at any time to break the contract with the investors and manage the loans by itself. By doing so, the bank's utility at time $t$ would be 
\[
R^i_t := \underset{k\in\mathfrak K} {\rm ess\; sup}\; \mathbb E^{\P^k}\bigg[\int_t^{\tau^I} \mathrm{e}^{-r(s-t)}(\rho_i\mu(I-N_s)+Bk_s) \mathrm{d}s \bigg| \Gc_t \bigg].
\]
Notice that the previous expression depends on $t$ only through the number of loans left at the time. It is straightforward then that $R^i_t=\widehat{R}^i_{I-N_t}$, for every $t\in[0,\tau^I]$.

\vspace{0.5em}
In this setting, a shutdown contract $(D,\theta)$ is one which is never broken by the bank of type $\rho_g$ and is rejected by the bank of type $\rho_b$, who prefers to run the project on its own. That is, $U_t^g(D,\theta)\geq R_t^g$ for every $t\in(0,\tau)$ and $U_0^b(D,\theta)<R_0^b$. 
To find the optimal shutdown contract, we need to characterize first the new credible set which includes additional state constraints for the good bank. Let us mention immediately that the right of the bank to break the contract generates differences between the credible sets associated to the shutdown and the screening problem, which are no longer equal.

\vspace{0.5em}
Define the renegotiation--proof feasible set for the good bank with $j$ loans left
\[
\widetilde {\mathcal{V}}^g_j = \widehat{\mathcal{V}}_j \cap [\widehat{R}^g_j,\infty), \quad j=1,\dots,I.
\]

\begin{Definition}
For any time $t\geq 0$, we define the shutdown renegotiation--proof credible set $\widetilde{\mathcal C}_{t}$ as the set of $(u^{b},u^g)\in\widehat{\mathcal{V}}_{I-N_t}\times\widetilde{\mathcal{V}}_{I-N_t}^g$ such that there exists an admissible contract $(\theta,D)\in \Theta \times \mathcal D$ satisfying $U_t^b(\theta,D)=u^b$, $U_t^g(\theta,D)=u^g$ and $(U_s^b(\theta,D), U_s^g(\theta,D)) \in\widehat{\mathcal{V}}_{I-N_s}\times\widetilde{\mathcal{V}}_{I-N_s}^g$ for every $s\in [t,\tau)$, $\P-a.s.$
\end{Definition}

Given a starting time $t\geq 0$ and $u^{b}\in\widehat{\mathcal{V}}_{I-N_t}$, define the set of contracts which are not broken by the good bank and under which the value function of the bad bank at time $t$ is equal to $u^b$,

\begin{equation*}
\mathcal{A}^{{\rm SH},b}(t,u^b) = \big\{ (\theta,D)\in\Theta\times \mathcal D: U_t^b(\theta,D)=u^b, U_s^g(\theta,D) \geq R_s^g, \text{ for every }s\in[t,\tau) \big\}.
\end{equation*}

We denote by $\mathfrak{U}^{\rm SH}_t(u^{b})$ the largest value $U_t^g(\theta,D)$ that the good bank can obtain from all the contracts $(\theta,D)\in\mathcal{A}^{{\rm SH},b}(t,u^b)$ and by $\mathfrak{L}^{SH}_t(u^{b})$ the lowest one. Again, these sets can be proved to depend on $t$ only through the value of $I-N_t$ so defining $\overline{\mathfrak U}_{I-N_t}(u^{b}):=\mathfrak{U}^{SH}_t(u^{b})$ and $\overline{\mathfrak L}_{I-N_t}(u^{b}):=\mathfrak{L}^{SH}_t(u^{b})$ we finally have
\[
\overline{\mathcal C}_j := \big\{ (u^{b},u^g)\in\widehat{\mathcal{V}}_j\times\widetilde{\mathcal{V}}_j^g : \overline{\mathfrak L}_j(u^{b}) \leq u^g \leq \overline{\mathfrak U}_j(u^{b}) \big\}.
\]
As depicted in Figure \ref{fig:discussion boundaries}, the upper boundary in the problem with full commitment is absorbing and it generates a movement of the utilities of the banks in the direction $(\rho_b,\rho_g)$. We conclude that the upper boundary in this extension is the same as before and it is given by
\[
\overline{\mathfrak U}_{I-N_t}(u^{b}) = \widehat{\mathfrak U}_{I-N_t}(u^{b}), \; \text{for every } u^b \text{ such that }  \widehat{\mathfrak U}_{I-N_t}(u^{b}) \geq \widehat{R}^g_{I-N_t}.
\]
On the other hand, since the former lower boundary $\widehat{\mathfrak L}_{I-N_t}(u^{b})$ generates a movement in the direction $(-\rho_b,-\rho_g)$, it cannot be used to obtain $\overline{\mathfrak L}_{I-N_t}(u^{b})$ which is the solution to the following control problem
\[
{\mathfrak L}^{\rm SH}_t(u^b)=\underset{(k,\Psi)\in\mathfrak K\times {\overline{\mathcal{A}}}^{{\rm SH},b}(t,u^b)}{\rm{ess\ inf}} \ \E^{\P^{k(\Psi)}} \bigg[ \int_t^\tau \mathrm{e}^{-r(s-t)} (\rho_g \mathrm{d}D_s+B k_s \mathrm{d}s) \bigg| \mathcal{G}_t \bigg],
\]
subject to the dynamics, for $t\in[ r, \tau]$
\begin{align*}
U_r^b(\Psi_g)=&\ u^b+\int_t^r\big(  ru_s^{b}-Bk_s^{\star,b}(\Psi) +h_s^{1,b}\lambda_s^{k^{\star,b}} + h_s^{2,b}(1-\theta_s)\lambda_s^{k^{\star,b}} \big)\mathrm{d}s -\rho_b \int_t^r\mathrm{d}D_s-\int_t^rh_s^{1,b}\mathrm{d}N_s-\int_t^rh_s^{2,b}\mathrm{d}H_s,
\end{align*}
with
\begin{align*}
& k_s^{\star ,b}(\Psi)  = (I-N_s){\bf1}_{\{h_s^{1,b}+(1-\theta_s)h_s^{2,b}<\widehat b_{I-N_s}\}},
\end{align*}
and where $\overline{\mathcal A}^{{\rm SH},b}(t,u^b)$ is defined similarly as $\overline{\mathcal A}^{b}(t,u^b)$ in Section \ref{sec:credi}. 
Once the boundaries are determined and the credible set is found, a system of recursive HJB equations can be associated to the Principal's problem, as in the original problem, and the same kind of study explained in Section \ref{sec:opti} follows.

\vspace{0.5em}
The optimal screening renegotiation--proof problem can be studied analogously, by defining the corresponding credible set, which is no longer equivalent to the credible set for the shutdown problem but will also keep the upper boundary from the original problem with full commitment of the banks.

{\small
 \bibliographystyle{plain}
\bibliography{bibliographyDylan}}

\appendix
\small
\section{Proofs for the pure moral hazard case}\label{app:1}
We provide in this section all the proofs of the results of Section \ref{sec:pure}. We start with the

\vspace{0.5em}
\begin{proof}[Proof of Proposition \ref{prop:rep}]
Using the martingale representation theorem\footnote{We emphasise that since the filtration $\G$ is augmented and generated by point processes, the predictable martingale representation holds for any of the probability measures $(\P^k)_{k\in\mathfrak K}$, see for instance \cite[Lemma A.1]{el2018optimal}.} (recall that $D$ is supposed to be integrable and that $k$ is bounded by definition), we deduce that for any $k\in\mathfrak K$ there exist $\G-$predictable processes $h^{1,i,k}$ and $h^{2,i,k}$ such that, ${\mathbb P}-\mathrm{a.s.}$
\begin{align}\label{dyn-u}
 \mathrm{d}u_t^i(k,\theta^i,D^i)=& \big(ru_t^i(k,\theta^i,D^i)-Bk_t\big)\mathrm{d}t-\rho_i\mathrm{d}D^i_t-h^{1,i,k}_t\big(\mathrm{d}N_t-\lambda_t^k\mathrm{d}t\big)-h^{2,i,k}_t\big(\mathrm{d}H_t-(1-\theta_t^i)\lambda_t^k\mathrm{d}t\big),\; 0\leq t<\tau,
\end{align}
Let us then define
\begin{align*}
&Y^{i,k}_t:=u_t^i(k,\theta^i,D^i),\; Z_t^{i,k}:=(h^{1,i,k}_t,h^{2,i,k}_t)^\top,\; M_t:=(N_t,H_t)^\top,\ \widetilde M_t^i:=M_t-\int_0^t\lambda^0_s(1,1-\theta^i_s)^\top \mathrm{d}s,\; K_t^i:=\rho_iD^i_t,
\end{align*}
so that we can rewrite \reff{dyn-u} as follows
\[
Y_t^{i,k}=0-\int_t^\tau f^i(s,k_s,Y_s^{i,k},Z_s^{i,k})\mathrm{d}s+\int_t^\tau Z^{i,k}_s\cdot \mathrm{d}\widetilde M^i_s+\int_t^\tau \mathrm{d}K^i_s,\; 0\leq t\leq \tau,\; \mathbb P-\mathrm{a.s.}
\]
In other words, $(Y^{i,k},Z^{i,k})$ appears as a (super--)solution to a BSDE with (finite) random terminal time, as studied for instance by Peng \cite{peng1991probabilistic} or Darling and Pardoux \cite{darling1997backwards}. Notice that by direct computations, it is immediate that it is equivalent to look for a solution $(Y^i,Z^i)$ of BSDE \eqref{bsde} or to look for solution to the following BSDE
\begin{equation}\label{eq:bsde}
\widetilde Y_t^{i}=\xi_\tau-\int_t^\tau \tilde g^i\big(s,\widetilde Z^i_s\big)\mathrm{d}s+\int_t^\tau \widetilde Z^{i}_s\cdot \mathrm{d}\widetilde M^i_s,\; 0\leq t\leq \tau,\; \mathbb P-\mathrm{a.s.},
\end{equation}
where we defined
\[
\widetilde Y^i_t:=\mathrm{e}^{-rt}Y^i_t+\int_0^t\mathrm{e}^{-rs}\mathrm{d}K_s,\; \widetilde Z^i_t:=\mathrm{e}^{-rt}Z^i_t,\; t\geq 0,\; \xi_\tau:=\int_0^\tau\mathrm{e}^{-rs}\mathrm{d}K_s,\; \tilde g^i(s,z):=(I-N_s)\bigg(\alpha_{I-N_s}\varepsilon z\cdot\begin{pmatrix}1\\ 1-\theta^i_s\end{pmatrix}-B\mathrm{e}^{-rs}\bigg)^-.
\]
 By direct computations, it is easy to see that $\tilde g^i$ satisfies, for any $(t,z,z')\in\R_+\times\R^2\times\R^2$
\begin{align*}
&\tilde g^i(t,z)-\tilde g^i(t,z')\leq =\gamma_t(z,z')\lambda^0_t(z-z')\cdot (1,1-\theta^i_t)^\top, 
\end{align*}
where $
\gamma_t(z,z'):=\varepsilon{\bf 1}_{\{(z-z')\cdot (1,1-\theta^i_t)^\top> 0\}},\ \text{verifies $0\leq \gamma_t(z,z')\leq \varepsilon$}.
$ Since in addition $\tilde g^i(t,0)$ is bounded, \eqref{eq:integD} holds, $\gamma_t(z,z^\prime)$ is bounded and non--negative, the intensity of $\widetilde M^i$ is also bounded, as well as its jumps, we deduce that all the assumptions of Theorems 3.5 and 3.24 in \cite{papapantoleon2016existence} hold in our setting, proving wellposedness of \eqref{bsde} in the space described by \eqref{eq:integ}, and that we can apply a comparison theorem. Therefore, we deduce immediately that for any $k\in\mathfrak K$
\[
Y^{i,k}_t\leq Y^i_t=Y^{i,k^{\star,i}}_t,\ \mathbb P-\mathrm{a.s.},
\]
where we defined 
\[
k^{\star,i}_t:=(I-N_t){\bf 1}_{\{Z^i_t\cdot(1,1-\theta^i_t)^\top<b_t\}},\text{ \rm and }b_t:=\frac{B}{\alpha_{I-N_t}\varepsilon},\ t\geq 0.
\]
This means that $Y^i$ is the value function of the bank, and that her optimal response given $(\theta^i,D^i)\in\Theta\times\mathcal D$ is $k^{\star,i}$. 
\end{proof}

 \vspace{0.5em}
We finish with the 
 
 \vspace{0.5em}
 \begin{proof}[Proof of Lemma \ref{lemma:feasible set}]
First of all, it is clear that the bank of type $\rho_i$ can get arbitrarily large levels of utility (it suffices for the investor to set $\mathrm{d}D^i_s:=n\mathrm{d}s$ for $n$ large enough, starting from time $t$). The bank's maximal level of utility is therefore $+\infty$, which corresponds to a utility equal to $-\infty$ for the investor. Then, coming back to the definition of the bank's problem, or to the BSDE \reff{bsde}, it is clear, for instance by using the comparison theorem for super solutions to \reff{bsde} (see \cite[Theorem 2.5]{royer2006backward}), that in order to minimise the utility that the bank obtains, the investor has to set $D^i=0$. Moreover, since by definition we must always have $Y_t^i\geq 0$ and $Y^i_\tau=0$, and since the totally inaccessible jumps of $Y$ (recall that $D$ is assumed to be predictable) are given by $\Delta Y_t^i=-Z^i_t\cdot \Delta M_t,$
we must have that 
\begin{equation}\label{eq:yz}
Y_{t^-}^i=Z^i_t\cdot (1,1)^\top,\text{ \rm and }Y_{t^-}^i\geq Z^i_t\cdot (1,0)^\top,\ t>0,\ \mathbb P-\mathrm{a.s.},
\end{equation}
Indeed, the support of the laws of $\tau$ and the $\tau^j$ under $\P$ is $[0,+\infty)$. This implies in particular that we must have $Z_t^i\cdot (0,1)^\top\geq 0$, which in turn implies that the generator $g^i$ is then non--increasing with respect to $\theta^i$, and thus that the minimal utility for the bank is attained, as expected, when $\theta^i=0$. Then, if $(\theta^i,D^i)=(0,0)$ (which is obviously in $\Theta\times\mathcal D$) starting from time $t$, it is clear that the bank will never monitor and will obtain
\begin{align*}
U_{t}^i(0,0)=B(I-N_t)\mathbb E^{\mathbb P^{I-N_\cdot}}\bigg[\int_t^\tau e^{-r(s-t)}\mathrm{d}s\bigg|\mathcal G_t\bigg]&=\frac{B(I-N_t)}{r}\Big(1-\mathbb E^{\mathbb P^{I-N_\cdot}}\Big[\mathrm{e}^{-r(\tau-t)}\Big|\mathcal G_t\Big]\Big)\\
&=\frac{B(I-N_t)}{r}\bigg(1-\int_0^{+\infty}\lambda^{I-N_t}_t\mathrm{e}^{-xr-x\lambda^{I-N_t}_t}\mathrm{d}x\bigg)=\frac{B(I-N_t)}{r+\lambda^{I-N_t}_t}.
\end{align*}
Notice that this corresponds to the investor getting 
\[
\mu(I-N_t)\mathbb E^{\mathbb P^{I-N_\cdot}}[ \tau-t|\mathcal G_t]=\frac{\mu(I-N_t)}{\lambda^{I-N_t}_t}.
\]
\end{proof}



\section{Short--term contracts with constant payments}\label{sec:appcons}

In this section we first analyse the optimal responses and the value functions of the banks at a starting time $t\geq0$, under contracts with constant payments of the form $\mathrm{d}D_s=c\mathrm{d}s$, where $c$ is any $\mathcal G_t-$measurable random variable, and with $\theta\equiv 0$, so that the pool is liquidated immediately after the first default. Then, we extend the study to the case in which the payments are delayed and they happen only after a certain time $t^\star> t$.

\subsection{Contracts with no delay} 
\label{constantpayment}

\begin{Proposition}
For any $t\geq0$, consider the contract $(\theta,D)\in\Theta\times\mathcal{D}$ such that
\[
\theta_s=0, \; \mathrm{d}D_s = c \mathrm{d}s, \; \forall s\geq t,
\]
where $c$ is any $\mathcal G_t-$measurable random variable. For $i\in\{g,b\}$, define $\bar{c}_i:=\ds\frac{b_{I-N_t}(r+\lambda_t^{0})}{\rho_i}$. The optimal effort of the bank of type $\rho_i$ and her expected utility under the contract are
\begin{itemize}
\item If $c \leq \bar{c}_i$ then $k_s^{\star,i}(\theta,D)=k_s^{\rm SH}, \forall s\in[t,\tau)$ and  $U_t^i(\theta,D)=\ds\frac{\rho_i c+B(I-N_t)}{r+\lambda_t^{k^{\rm SH}}}$.
\item If $c \geq \bar{c}_i$ then $k_s^{\star,i}(\theta,D)=0, \forall s\in[t,\tau)$ and $U_t^i(\theta,D)=\ds\frac{\rho_i c}{r+\lambda_t^{0}}$.
\end{itemize}
\end{Proposition}

\begin{proof}
\vspace{0.5em}
$(i)$ If the bank of type $\rho_i$ always monitors, we have
\[
u_t^i(0,\theta,D)=\E^{\mathbb P}\bigg[\int_t^\tau \mathrm{e}^{-r(s-t)}\rho_i c\mathrm{d}s \bigg| \Gc_t \bigg]=\frac{\rho_i c}{r+\lambda_t^{0}}.
\]
Hence, the continuation utility is constant in time and if the payment $c$ is exactly equal to $u^i(r+\lambda_{t}^{0})/\rho_i,$ for some $u^i\geq 0$,
then the bank receives exactly $u^i$. In this case, the strategy of always monitoring is incentive compatible if and only if $u^i\geq b_{I-N_t}$. The minimum payment such that the bank of type $\rho_i$ will always work is therefore
\[
\overline{c}_i=\frac{b_{I-N_t}(r+\lambda_t^{0})}{\rho_i}.
\]
$(ii)$ If the bank of type $\rho_i$ always shirks, her continuation utility is constant and equal to
\[
u_t^i\big(k^{\rm SH},\theta,D\big)=\E^{\mathbb P^{k^{\rm SH}}}\bigg[\int_t^\tau \mathrm{e}^{-r(s-t)}(\rho_i c+B)\mathrm{d}s  \bigg| \Gc_t\bigg]=\frac{\rho_i c+B(I-N_t)}{r+\lambda_t^{k^{\rm SH}}}.
\]
Then, if for some $u_i\geq 0$ one takes $c$ equal to
\[
\frac{u^i\big(r+\lambda_t^{k^{\rm SH}}\big)-B(I-N_t)}{\rho_i},
\]
the bank receives $u^i$. Therefore $k^{\rm SH}$ is incentive compatible if and only if $u^i<b_{I-N_t}$. Nevertheless, since the payment $c$ must be positive, $u^i$ must be greater than $B(I-N_t)/(r+\lambda_t^{k^{\rm SH}})$. The supremum of the payments such that the bank of type $\rho_i$ will always shirk is therefore equal to
\[
\frac{b_{I-N_t}\big(r+\lambda_t^{k^{\rm SH}}\big)-B(I-N_t)}{\rho_i}=\frac{b_{I-N_t}(r+\lambda_t^{0})}{\rho_i}=\overline{c}_i.
\]
\end{proof}

\subsection{Contracts with delayed payments}
\label{contractswithdelay}

\begin{Proposition}
For any $t\geq0$, consider the contract $(\theta,D)\in\Theta\times\mathcal{D}$ such that
\[
\theta_s=0, ~ \mathrm{d}D_s = c {\bf 1}_{s\geq t^\star} \mathrm{d}s, \; \forall s\geq t,
\]
where $c$ is any $\mathcal G_t-$measurable random variable and $t^\star>t$ is a fixed constant. For $i\in\{g,b\}$, define the time 
\[\bar{t}_i(c):=t+\ds\frac{1}{r+\lambda_t^{0}}\log\bigg(\frac{\rho_i c}{b_{I-N_t}(r+\lambda_t^{0})}\bigg).
\] The optimal effort of the bank of type $\rho_i$ and her expected utility under the contract are
\begin{itemize}
\item If $c \leq \bar{c}_i$, then $k_s^{\star,i}(0,D)=k_s^{\rm SH}, \forall s\in[t,\tau)$ and  $U_t^i(0,D)=\ds \exp\Big(-(r+\lambda_t^{k^{\rm SH}})(t^\star-t)\Big)\frac{\rho_i c}{r+\lambda_t^{k^{\rm SH}}}+\frac{B(I-N_t)}{r+\lambda_t^{k^{\rm SH}}}$.
\item If $c > \bar{c}_i,~ t^\star\leq \bar{t}_i(c)$, then $k_s^{\star,i}(0,D)=0, \forall s\in[t,\tau)$ and $U_t^i(0,D)=\ds \exp\Big(-(r+\lambda_t^{0})(t^\star-t)\Big)\frac{\rho_i c}{r+\lambda_t^{0}}$.
\item If $c > \bar{c}_i,~ t^\star> \bar{t}_i(c)$, then $k_s^{\star,i}(0,D)=k_s^{\rm SH}{\bf 1}_{\{s<\bar{t}_i(c)\}},\; \forall s\in[t,\tau)$ and \[
U_t^i(0,D)=\exp\Big(-(r+\lambda_t^{k^{\rm SH}})(t^\star-t)\Big)\bigg(\frac{\rho_i c}{b_{I-N_t}(r+\lambda_t^{0})}\bigg)^\frac{r+\lambda_t^{k^{\rm SH}}}{r+\lambda_t^{0}}\frac{b_{I-N_t}(r+\lambda_t^{0})}{r+\lambda_t^{k^{\rm SH}}}+\frac{B(I-N_t)}{r+\lambda_t^{k^{\rm SH}}}.
\]
\end{itemize}
\end{Proposition}

\begin{proof}
\vspace{0.5em}
$(i)$ If the bank of type $\rho_i$ always works, at any time $t\leq s<t^\star$, her continuation utility is, noticing that since $\theta=0$, we have that $(\lambda_{u}^{0})_{u\geq t}$ is constant
\[
u_s^i(0,0,D)=\E^{\mathbb{P}^0}\bigg[\int_{t^\star\wedge\tau}^\tau \mathrm{e}^{-r(u-s)}\rho_i c\mathrm{d}u\bigg| \Gc_s\bigg]=\frac{\mathrm{e}^{-(r+\lambda_t^{0})(t^\star-s)}\rho_i c}{r+\lambda_t^{0}}=u_t^i(0,0,D)\mathrm{e}^{(r+\lambda_t^{0})(s-t)}.
\]
Therefore, at $s=t^\star$ the continuation utility of the bank is $
u_{t^\star}^i(0,0,D)=u_t^i(0,0,D)\mathrm{e}^{(r+\lambda_t^{0})(t^\star-t)}.
$
Next, for any $s>t^\star$, the continuation utility of the bank will be 
\[
u_s^i(0,0,D)=\E^{\mathbb{P}^0}\bigg[\int_s^\tau \mathrm{e}^{-r(u-s)}\rho_i c\mathrm{d}s \bigg| \Gc_s \bigg]=\frac{\rho_i c}{r+\lambda_t^{0}}.
\]
Then, we see that once the bank starts being paid, her continuation utility becomes constant and it must be equal to $u_{t^\star}^i(0,0,D)$. Then, if for some $u^i\geq 0$, one chooses $c$ equal to
\begin{equation} \label{firstpayment} 
\frac{u^i\mathrm{e}^{(r+\lambda_t^{0})t^\star}(r+\lambda_t^{0})}{\rho_i},
\end{equation} 
the continuation utility of the bank will be an increasing process with initial value $u^i$. Therefore, $0$ is incentive compatible if and only if $u^i\geq b_{I-N_t}$. The minimum payment and delay such that the bank always works are $t^\star=0$ and $\overline{c}_i=\frac{b_{I-N_t}(r+\lambda_t^{0})}{\rho_i}.$	

\medskip
$(ii)$ If the bank of type $\rho_i$ always shirks, at any time $t\leq s<t^\star$, her continuation utility is
\[
u_s^i\big(k^{\rm SH},0,D\big)=\E^{\mathbb{P}^{k^{\rm SH}}}\bigg[\int_{t^\star\wedge\tau}^\tau \mathrm{e}^{-r(u-s)}\rho_i c \mathrm{d}u+\int_s^\tau B\mathrm{e}^{-r(u-s)}(I-N_t)\mathrm{d}u\bigg|\mathcal G_s\bigg]=\frac{\mathrm{e}^{-\big(r+\lambda_t^{k^{\rm SH}}\big)(t^\star-s)}\rho_i c}{r+\lambda_t^{k^{\rm SH}}}+\frac{B(I-N_t)}{r+\lambda_t^{k^{\rm SH}}}.
\]
Therefore
\[
u_s^i\big(k^{\rm SH},0,D\big)=\mathrm{e}^{\big(r+\lambda_t^{k^{\rm SH}}\big)(s-t)}\bigg(u_t^i(k^{\rm SH},0,D)-\frac{B(I-N_t)}{r+\lambda_t^{k^{\rm SH}}}\bigg)+\frac{B(I-N_t)}{r+\lambda_t^{k^{\rm SH}}},
\]
and the continuation utility is an increasing process. Recall that $k^{\rm SH}$ is incentive compatible if and only if $u_s^i\big(k^{\rm SH},0,D\big)<b_{I-N_t}$ for every $s\geq t$. However, if $t^\star$ is large, there will exist $t_w$ such that $u_{t_w}^i\big(k^{\rm SH},0,D\big)=b_{I-N_t}$ and the bank will start to work. More precisely, $t_w$ depends on the initial value $u_t^i(k^{\rm SH},0,D)$ and is given by 
\[
t_w:=t+\frac{1}{r+\lambda_t^{k^{\rm SH}}}\log\bigg(\frac{b_{I-N_t}(r+\lambda_t^{k^{\rm SH}})-B(I-N_t)}{u_t^i(k^{\rm SH},0,D)(r+\lambda_t^{k^{\rm SH}})-B(I-N_t)}\bigg).
\]
Notice that $t_w\geq t$ if and only if $b_{I-N_t}\geq u_t^i(k^{\rm SH},0,D).$ Therefore, $k^{\rm SH}$ is incentive compatible if and only if $t^\star<t_w$. Under this condition, at $t=t^\star$ the continuation utility of the bank is
\[
u_{t^\star}^i\big(k^{\rm SH},0,D\big)=\mathrm{e}^{(r+\lambda_t^{k^{\rm SH}})(t^\star-t)}\bigg(u_t^i(k^{\rm SH},\theta,D)-\frac{B(I-N_t)}{r+\lambda_t^{k^{\rm SH}}}\bigg)+\frac{B(I-N_t)}{r+\lambda_t^{k^{\rm SH}}}<b_{I-N_t}.
\]
Once the bank starts being paid her continuation utility is constant and equal to 
\[
u_s^i\big(k^{\rm SH},0,D\big)=\E^{\mathbb{P}^{k^{\rm SH}}}\bigg[\int_s^\tau \mathrm{e}^{-r(u-s)}(\rho_i c+B(I-N_t)) \mathrm{d}u\bigg|\Gc_s\bigg]=\frac{\rho_i c+B(I-N_t)}{r+\lambda_t^{k^{\rm SH}}}.
\]
So if for some $u^i\geq 0$ the payment $c$ is equal to 
\begin{equation}\label{secondpayment} 
\frac{\mathrm{e}^{(r+\lambda_t^{k^{\rm SH}})(t^\star-t)}\big(u^i(r+\lambda_t^{k^{\rm SH}})-B(I-N_t)\big)}{\rho_i},
\end{equation} 
the expected payoff of the bank at time $t$ is $u^i$. The supremum of the delays and payments such that the bank always shirks are respectively $t_w$ and
\[
\frac{\mathrm{e}^{(r+\lambda_t^{k^{\rm SH}})(t_w-t)} \big(b_{I-N_t}(r+\lambda_t^{k^{\rm SH}})-B(I-N_t)\big)}{\rho_i}=\frac{b_{I-N_t}(r+\lambda_t^{0})}{\rho_i}=\overline{c}_i.
\]

$(iii)$ Finally, consider the case when $t^\star$ is greater than $t_w$. Under this contract, the bank will shirk until time $t_w$ and will work afterwards. Indeed, from the previous analysis we know that this strategy is incentive compatible. At time $t_w$ we have that $u_{t_w}^i(k^{\rm SH},0,D)=b_{I-N_t}$ and for $s\in[t_w,t^\star)$ the continuation utility is given by 
\begin{align*}
u_s^i(0,0,D) &= \E^{\mathbb{P}}\bigg[\int_{t^\star\wedge\tau}^\tau \mathrm{e}^{-r(u-s)}\rho_i c\mathrm{d}u\bigg| \Gc_s\bigg]=\frac{\mathrm{e}^{-(r+\lambda_t^{0})(t^\star-s)}\rho_i c}{r+\lambda_t^{0}}= \mathrm{e}^{(r+\lambda_t^{0})(s-t_w)}u_{t_w}^i(k^{\rm SH},0,D)=b_{I-N_t}\mathrm{e}^{(r+\lambda_t^{0})(s-t_w)}.
\end{align*}
Therefore, at $t=t^\star$ the continuation utility of the bank is 
\[
u_{t^\star}^i(0,0,D)=b_{I-N_t}\mathrm{e}^{(r+\lambda_t^{0})(t^\star-t_w)},
\]
and for any $s>t^{\star}$, the continuation utility of the bank is constant and equal to 
\[
u_s^i(0,0,D)=\E^{\mathbb P}\bigg[\int_s^\tau \mathrm{e}^{-r(u-s)}\rho_i c\mathrm{d}u \bigg| \Gc_s \bigg]=\frac{\rho_i c}{r+\lambda_t^{0}}.
\]
So if for some $u^i\geq 0$ the payment $c$ is equal to 
\begin{equation} \label{thirdpayment} 
\frac{b_{I-N_t}(r+\lambda_t^{0})\mathrm{e}^{(r+\lambda_t^{0})(t^\star-t)}}{\rho_i}\bigg(\frac{u^i(r+\lambda_t^{k^{\rm SH}})-B(I-N_t)}{b_{I-N_t}(r+\lambda_t^{0})}\bigg)^{\frac{r+\lambda_t^{0}}{r+\lambda_t^{k^{\rm SH}}}},
\end{equation} 
the expected payoff of the bank at time $t$ is $u^i$. The minimum payment and delay such that the bank shirks first and works afterwards are $t^\star=t_w$ and $\overline{c}_i=\frac{b_{I-N_t}(r+\lambda_t^{0})}{\rho_i}.$ Notice that $\overline{t_i}(c)$ in the statement of the Proposition is the corresponding expression for $t_w$ as a function of the payments $c$. 
\end{proof}

We conclude this section with the following result, saying that every point in the upper boundary of the credible set can be attained by short--term contracts with delay. 

\begin{Proposition}
Fix some $t\geq 0$. For any point $(u^b,u^g)$ in the upper boundary of the credible set $\mathcal{C}_{t}$, there exists a $\Gc_t-$ measurable payment $c$, and $t^\star\geq t$ such that the contract $(\theta,D)$ with $\theta_s=0, ~ \mathrm{d}D_s = c {\bf 1}_{s\geq t^\star} \mathrm{d}s, \; \forall s\geq t,$ is such that $U_t^b(\theta,D)=u^b$ and $U_t^g(\theta,D)=u^g$.
\end{Proposition}


\begin{proof}
$(i)$ Let $c>\bar{c}_b>\bar{c}_g$ and $t^\star \leq \bar{t}_b(c)<\bar{t}_g(c)$. Then $k^{\star,b}(\theta,D)=k^{\star,g}(\theta,D)=0$ and the values of the banks are 
\[
U_t^g(\theta,D)=\frac{\rho_g c}{r+\lambda_t^{0}}e^{-(r+\lambda_t^{0})(t^\star-t)}, ~U_t^b(\theta,D)=\frac{\rho_b c}{r+\lambda_t^{0}}e^{-(r+\lambda_t^{0})(t^\star-t)}.
\]
Therefore the utilities satisfy
\[
U_t^g(\theta,D)=\frac{\rho_g}{\rho_b}U_t^b(\theta,D),\ \text{with }U_t^g(\theta,D)\in\left[\frac{\rho_g}{\rho_b}b_{I-N_t},\infty\right),~ U_t^b(\theta,D)\in\left[b_{I-N_t},\infty\right).
\]

\vspace{0.5em}
$(ii)$ If $c>\bar{c}_b$ and $\bar{t}_b(c)<t^{\star}\leq\bar{t}_g(c)$, we have that the good bank will always work and the bad bank will start working at time $\bar{t}_b(c)$. Their value functions are 
\begin{align*}
U_t^g(\theta,D) &= \frac{\rho_g c}{r+\lambda_t^{0}}\mathrm{e}^{-(r+\lambda_t^{0})(t^\star-t)}, \; U_t^b(\theta,D) =  \mathrm{e}^{-(r+\lambda_t^{k^{\rm SH}})(t^\star-t)}\bigg(\frac{\rho_b c}{b_{I-N_t}(r+\lambda_t^{0})}\bigg)^\frac{r+\lambda_t^{k^{\rm SH}}}{r+\lambda_t^{0}}\frac{b_{I-N_t}(r+\lambda_t^{0})}{r+\lambda_t^{k^{\rm SH}}}+\frac{B(I-N_t)}{r+\lambda_t^{k^{\rm SH}}},
\end{align*} 
so they belong to the curve 
\[
U_t^g(\theta,D) =\frac{\rho_g}{\rho_b}b_{I-N_t}^\frac{\lambda_t^{k^{\rm SH}}-\lambda_t^{0}}{r+\lambda_t^{k^{\rm SH}}}\bigg(U_t^b(\theta,D)- \frac{B(I-N_t)}{r+\lambda_t^{k^{\rm SH}}}\bigg)^\frac{r+\lambda_t^{0}}{r+\lambda_t^{k^{\rm SH}}}\bigg(\frac{r+\lambda_t^{k^{\rm SH}}}{r+\lambda_t^{0}}\bigg)^{\frac{r+\lambda_t^{0}}{r+\lambda_t^{k^{\rm SH}}}},\]
and take values in the sets (recall the definition of $x_j^\star$ in Proposition \ref{prop:uj functions})
\[
U_t^g(\theta,D) \in\left[b_{I-N_t},\frac{\rho_g}{\rho_b}b_{I-N_t}\right),~U_t^b(\theta,D) \in [x_{I-N_t}^\star,b_{I-N_t}).
\]

\vspace{0.5em}
$(iii)$ If $c>\bar{c}_b$ and $\overline{t}_g(c)<t^\star$, the good bank will start working at time $\bar{t}_g(c)$ and the bad bank will start to work at time $\bar{t}_b(c)$. Their value functions are
\begin{align*}
U_t^g(\theta,D)&=\mathrm{e}^{-(r+\lambda_t^{k^{\rm SH}})(t^\star-t)}\bigg(\frac{\rho_g c}{b_{I-N_t}(r+\lambda_t^{0})}\bigg)^\frac{r+\lambda_t^{k^{\rm SH}}}{r+\lambda_t^{0}}\frac{b_{I-N_t}(r+\lambda_t^{0})}{r+\lambda_t^{k^{\rm SH}}}+\frac{B(I-N_t)}{r+\lambda_t^{k^{\rm SH}}},\\
U_t^b(\theta,D)&=\mathrm{e}^{-(r+\lambda_t^{k^{\rm SH}})(t^\star-t)}\bigg(\frac{\rho_b c}{b_{I-N_t}(r+\lambda_t^{0})}\bigg)^\frac{r+\lambda_t^{k^{\rm SH}}}{r+\lambda_t^{0}}\frac{b_{I-N_t}(r+\lambda_t^{0})}{r+\lambda_t^{k^{\rm SH}}}+\frac{B(I-N_t)}{r+\lambda_t^{k^{\rm SH}}},
\end{align*}
so they belong to the line 
\[
U_t^g(\theta,D) =\bigg(\frac{\rho_g}{\rho_b}\bigg)^\frac{r+\lambda_t^{k^{\rm SH}}}{r+\lambda_t^{0}}\bigg(U_t^b(\theta,D) -\frac{B(I-N_t)}{r+\lambda_t^{k^{\rm SH}}}\bigg)+\frac{B(I-N_t)}{r+\lambda_t^{k^{\rm SH}}},
\]
with 
\[U_t^g(\theta,D) \in\bigg[\frac{B(I-N_t)}{r+\lambda_t^{k^{\rm SH}}},b_{I-N_t}\bigg),~U_t^b(\theta,D) \in\bigg[\frac{B(I-N_t)}{r+\lambda_t^{k^{SH}}},x_{I-N_t}^\star\bigg).\] 
\end{proof}

\subsection{Initial lump--sum payment} \label{lumpsum}
Take any point $(u^b,u^g)$ in the credible set at time $t$. We know that there exists an admissible contract $(\theta,D)$, such that $U_t^b(\theta,D)=u^b$ and $U_t^g(\theta,D)=u^g$. Consider the payments $D^\ell$ which differ from $D$ only at time $t$, where a lump-sum payment of size $\ell>0$ is made. This added lump-sum payment will not change the banks' incentives and the new value functions at time $t$ will be 
\[
U_t^g(\theta,D^\ell)=u^g+\rho_g\ell,~U_t^b(\theta,D^\ell)=u^b+\rho_b\ell.
\]
Hence, the new pair of values of the banks belong to the line with slope $\frac{\rho_g}{\rho_b}$ which passes through the point $(u^b,u^g)$. Since in our setting there is no upper bound on the payment, by increasing the value of $\ell$ it is possible to reach every point of the ray which starts at $(u^b,u^g)$ and goes in the positive direction. 


\subsection{Credible region under contracts with delay} 
\label{credibleregiondelay}

From the previous subsection we know that for every point $(u^b,u^g)$ on the upper boundary there exists a pair $(c,t^\star)$, with $c>\bar{c}_b$, such that under the contract $(\theta\equiv 0,\; \mathrm{d}D_s=c{\bf 1}_{\{s\geq t^\star\}}\mathrm{d}s)$ we have $U_t^b(\theta,D)=u^b$ and $U_t^g(\theta,D)=u^g$. As explained in Section \ref{lumpsum}, if we consider the contract $(\theta, D^\ell)$ with an additional initial lump--sum payment, the incentives of the banks will not change and the new value functions of the agents will be $U_t^b(\theta,D^\ell)=u^b+\rho_b \ell$, $U_t^g(\theta,D)=u^g+\rho_g \ell$. Therefore under short--term contracts with delay which reach the upper boundary and lump--sum payments, all the subregion of the credible set delimited by the lines shown in Figure \ref{crdelay} can be reached. We will not enter into details but it can be proved that under all the short--term contracts with delay (not only the ones who reach the upper boundary) and lump-sum payments, the subregion of the credible set which can be reached is exactly the same. When there is only one loan left, this region is equal to the whole credible set but when $j>1$ the credible set is strictly bigger due to the pair of utilities that can be achieved in situations when $\theta \not\equiv 0$.

\begin{figure}[H]
\centering
\includegraphics[scale=0.8]{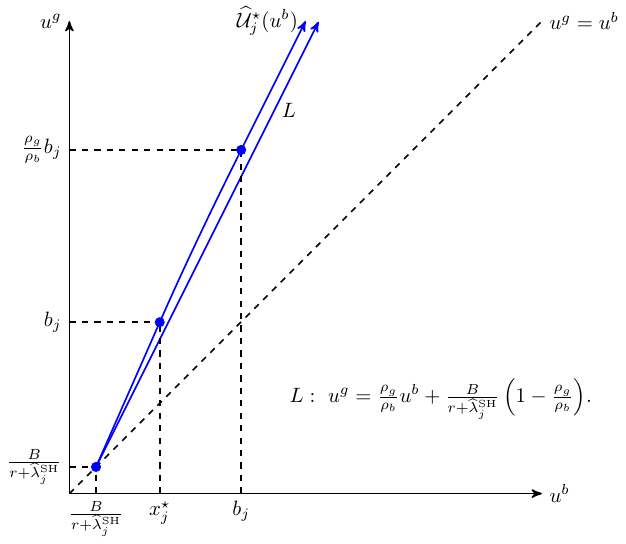}
\caption{Credible region under short-term contracts with delay and lump-sum payment.}\label{crdelay}
\end{figure}

\section{Technical results for the lower boundary}\label{sec:lower}
We begin this section with the 

\vspace{0.5em}
\begin{proof}[Proof of Proposition \ref{prop:utility-shirking}]
Observe first that we have
\[
\E^{\mathbb P^{k^{\rm SH}}} \big[  \mathrm{e}^{-r(\tau^{N_t+1} - t)}  \big| \mathcal G_t  \big]  =  \int_0^\infty \mathrm{e}^{-r x} \widehat\lambda_{I-N_t}^{\rm SH}  \mathrm{e}^{-\widehat\lambda_{I-N_t}^{\rm SH} x} \mathrm{d}x = \frac{\widehat \lambda_{I-N_t}^{\rm SH}}{r+\widehat\lambda_{I-N_t}^{\rm SH}},
\]
and for any $\ell\in\{N_t+1,\dots,I-1\}$
\[
\E^{\mathbb P^{k^{\rm SH}}} \big[   \mathrm{e}^{-r(\tau^{\ell+1} - \tau^\ell)}  \big| \mathcal G_t   \big]  =  \int_0^\infty  \mathrm{e}^{-r x} \widehat\lambda_{I-\ell}^{\rm SH}  \mathrm{e}^{-\widehat\lambda_{I-\ell}^{\rm SH} x} \mathrm{d}x = \frac{\widehat\lambda_{I-\ell}^{\rm SH}}{r+\widehat\lambda_{I-\ell}^{\rm SH}}.
\]
Thus, the utility that the bank gets from shirking (without considering the payments in the contract) is 
\begin{align*}
u_t(k^{\rm SH},\theta,0) & =  \E^{\mathbb P^{k^{\rm SH}}}\bigg[\int_t^\tau  \mathrm{e}^{-r(s-t)} B(I-N_s) \mathrm{d}s \bigg| \mathcal G_t \bigg] \\
& =  \E^{\mathbb P^{k^{\rm SH}}}\bigg[\int_t^{\tau^{N_t+1}}  \mathrm{e}^{-r(s-t)} B(I-N_t) \mathrm{d}s + \ds\sum_{i=N_t+1}^{N_t+m-1} \int_{\tau^i}^{\tau^{i+1}}  \mathrm{e}^{-r(s-t)} B(I-i) \mathrm{d}s \bigg| \mathcal G_t  \bigg] \\
& =  \frac{B(I-N_t)}{r} \E^{\mathbb P^{k^{\rm SH}}} \big[1 -  \mathrm{e}^{-r(\tau^{N_t+1} - t)} \big| \mathcal G_t \big]  + \ds\sum_{i=N_t+1}^{N_t+m-1} \frac{B(I-i)}{r} \E^{\mathbb P^{k^{\rm SH}}}\big[ \mathrm{e}^{-r(\tau^i - t)} -  \mathrm{e}^{-r(\tau^{i+1} - t)} \big| \mathcal G_t  \big] \\
& =  \frac{B(I-N_t)}{r+\widehat\lambda_{I-N_t}^{\rm SH}} + \ds\sum_{i=N_t+1}^{N_t+m-1} \frac{B(I-i)}{r} \E^{\mathbb P^{k^{\rm SH}}} \bigg[ \big( 1 -  \mathrm{e}^{-r(\tau^{i+1} - \tau^i)} \big) \prod_{\ell=N_t}^{i-1}  \mathrm{e}^{-r(\tau^{\ell+1} - \tau^\ell)} \bigg| \mathcal G_t  \bigg].
\end{align*}
Therefore, by independence we have
\begin{align*}
u_t(k^{\rm SH},\theta,0) & =   \frac{B(I-N_t)}{r+\widehat\lambda_{I-N_t}^{\rm SH}} + \ds\sum_{i=N_t+1}^{N_t+m-1} \frac{B(I-i)}{r + \widehat\lambda_{I-i}^{\rm SH}}   \prod_{\ell=N_t}^{i-1} \ds\frac{\widehat\lambda_{I-\ell}^{\rm SH}}{r+\widehat\lambda_{I-\ell}^{\rm SH}}  =   \frac{B(I-N_t)}{r+\widehat\lambda_{I-N_t}^{\rm SH}} + \ds\sum_{i=I-N_t-m+1}^{I-N_t-1} \frac{Bi}{r +\widehat \lambda_{i}^{\rm SH}}   \prod_{\ell=i+1}^{I-N_t} \ds\frac{\widehat\lambda_{\ell}^{\rm SH}}{r+\widehat\lambda_{\ell}^{\rm SH}}.
\end{align*}
\end{proof}

We proceed with the

\begin{proof}[Proof of Lemma \ref{lemma:inequalities}]
The value functions of the banks under $\Psi:=(\theta,D)$ are given by
\begin{align*}
U_t^g(\Psi) & =  \E^{\P^{k^{\star,g}(\Psi)}}\bigg[\int_t^\tau  \mathrm{e}^{-r(s-t)}(\rho_g  \mathrm{d}D_s+Bk_s^{\star,g}(\Psi) \mathrm{d}s)\bigg| \mathcal{G}_t\bigg], \; U_t^{b}(\Psi) & = \E^{\P^{k^{\star,b}(\Psi)}}\bigg[\int_t^\tau  \mathrm{e}^{-r(s-t)}(\rho_b  \mathrm{d}D_s+Bk_s^{\star,b}(\Psi) \mathrm{d}s)\bigg|\mathcal{G}_t\bigg].
\end{align*}
Thus, we first have, $\P-$a.s.
\begin{align*}
U_t^g(\Psi) & \geq  \E^{\P^{k^{\star,b}(\Psi)}}\bigg[ \int_t^\tau  \mathrm{e}^{-r(s-t)}(\rho_g  \mathrm{d}D_s+Bk_s^{\star,b}(\Psi) \mathrm{d}s) \bigg| \mathcal{G}_t \bigg]  \geq  \E^{\P^{k^{\star,b}(\Psi)}}\bigg[\int_t^\tau  \mathrm{e}^{-r(s-t)}(\rho_b  \mathrm{d}D_s^g+Bk_s^{\star,b}(\Psi) \mathrm{d}s)\bigg| \mathcal{G}_t\bigg] = U_t^{b}(\Psi).
\end{align*}
But we also have
\begin{align*}
U_t^g(\Psi)  \geq  \E^{\P^{k^{\star,b}(\Psi)}}\bigg[ \int_t^\tau  \mathrm{e}^{-r(s-t)}(\rho_g  \mathrm{d}D_s+Bk_s^{\star,b}(\Psi) \mathrm{d}s)\bigg| \mathcal{G}_t\bigg] & =  U_t^{b}(\Psi) + (\rho_g-\rho_b) \E^{\P^{k^{\star,b}(\Psi)}}\bigg[\int_t^\tau  \mathrm{e}^{-r(s-t)} \mathrm{d}D_s \bigg| \mathcal{G}_t\bigg] \\ 
& =  \frac{\rho_g}{\rho_b}U_t^{b}(\Psi)-\frac{(\rho_g-\rho_b)}{\rho_b} \E^{\P^{k^{\star,b}(\Psi)}}\bigg[\int_t^\tau  \mathrm{e}^{-r(s-t)}Bk_s^{\star,b}(\Psi)  \mathrm{d}s\bigg|\mathcal{G}_t\bigg].
\end{align*}
Observe next that 
\[
\sup_{k\in{\mathfrak{K}}} \E^{\P^k}\bigg[\int_t^\tau  \mathrm{e}^{-r(t-s)} B k_s  \mathrm{d}s\bigg| \mathcal{G}_t\bigg]=\E^{\P^{k^{\rm SH}}}\bigg[\int_t^\tau  \mathrm{e}^{-r(t-s)} B k_s^{\rm SH}  \mathrm{d}s \bigg| \mathcal{G}_t \bigg],
\]
because the left--hand side is the value function of a bank who is offered a contract with no payments. Therefore, we have that
\begin{align*}
U_t^g(\Psi)& \geq \frac{\rho_g}{\rho_b}U_t^{b}(\Psi)-\frac{(\rho_g-\rho_b)}{\rho_b} \E^{\P^{k^{\rm SH}}}\bigg[\int_t^\tau  \mathrm{e}^{-r(s-t)}Bk_s^{\rm SH}  \mathrm{d}s \bigg| \mathcal{G}_t\bigg] \geq \frac{\rho_g}{\rho_b}U_t^{b}(\Psi) -\frac{(\rho_g-\rho_b)}{\rho_b} C(I-N_t),
\end{align*}
because the utility that the banks get from shirking is non--decreasing with respect to the process $\theta$ and its maximum value is equal to $C(I-N_t)$, attained when $\theta\equiv 1$ (see \eqref{eq:C}).
\end{proof}

\vspace{0.5em}
We continue this section with the

\vspace{0.5em}
\begin{proof}[Proof of Proposition \ref{lowerboundary}]
Thanks to Lemma \ref{lemma:inequalities}, it suffices to prove the existence of contracts under which the value functions of the banks satisfy the equalities.

\vspace{0.5em}
\hspace{2em}{$\bullet$ \bf Step 1:} First, fix some $t\geq 0$, take any $u^{b}\in\left[c(I-N_t,1),C(I-N_t)\right]$ and fix $m\in\{1,\dots,I-N_t-1\}$ such that $
c(I-N_t,m)\leq u^{b} \leq c(I-N_t,m+1).
$
Next, take $\theta^0_t(u^{b})\in[0,1]$ such that $u^{b} = c(I-N_t,m) + \theta^0_t(u^{b}) (c(I-N_t,m+1)-c(I-N_t,m)).$ Then, there is a contract $(\theta,D)\in\Theta\times\mathcal D$ such that $U_t^g(\theta,D)=U_t^{b}(\theta,D)=u^{b}$. Such a contract can be defined as follows
\[
 \mathrm{d}D_s:=0,\ \theta_s:={\bf 1}_{ \left\{t\leq s\leq\tau^{N_t+m}\right\}}+(1-\theta^0_t(u^{b})){\bf 1}_{\left\{\tau^{N_t+m}<s\leq\tau^{N_t+m+1}\right\}},\ \text{ for every $s\geq t$}.
\]
The contract has no payments, it always maintains the pool after the first $m$ defaults, maintains the pool with probability $\theta^0$ after default $m+1$, and liquidates the pool at default $m+2$. It is clear that under this contract both banks always shirk in $[t,\tau]$, since they are not paid, and their value functions satisfy 
\begin{align*}
U_t^g(\theta,D)=U_t^{b}(\theta,D) &= \E^{\P^{k^{\rm SH}}}\bigg[\int_t^\tau  \mathrm{e}^{-r(s-t)}Bk_s^{\rm SH} \mathrm{d}s \bigg| \mathcal{G}_t\bigg] =c(I-N_t,m) + \theta_t^0(u^{b}) (c(I-N_t,m+1)-c(I-N_t,m))=u^{b}.
\end{align*}

\hspace{2em}{$\bullet$ \bf Step 2:} Fix again some $t\geq 0$, and choose now any $u^{b}\geq C(I-N_t)$ and define $u^g:=\frac{\rho_g}{\rho_b}u^{b}-\frac{(\rho_g-\rho_b)}{\rho_b}C(I-N_t).$ Let $\ell_t:=(u^{b}-C(I-N_t))/\rho_b$ and consider the admissible contract satisfying, $\theta_s=1,\  \mathrm{d}D_s=\ell_t{\bf 1}_{\left\{s=t\right\}},$ for every $s\geq t$. The optimal strategy for both banks under this contract is to always shirk and then
\begin{align*}
U_t^{b}(\theta,D) & =\E^{\P^{k^{\rm SH}}}\bigg[ \int_t^\tau  \mathrm{e}^{-r(s-t)}\big(\rho_b  \mathrm{d}D_s+Bk_s^{\rm SH} \mathrm{d}s\big)\bigg|\mathcal{G}_t\bigg]=\rho_b \ell_t+C(I-N_t)=u^{b}, \\
 U_t^g(\theta,D) & =\E^{\P^{k^{\rm SH}}}\bigg[\int_t^\tau  \mathrm{e}^{-r(s-t)}\big(\rho_g  \mathrm{d}D_s+Bk_s^{\rm SH} \mathrm{d}s)\bigg|\mathcal{G}_t\bigg]=\rho_g \ell_t+C(I-N_t)=u^g.
\end{align*}
\end{proof}

We conclude this section by proving some useful results that will be used in Section \ref{sec:valueinvestor} in the study of the value function of the investor on the lower boundary. We show that there are several ways of reaching the lower boundary and that all the contracts which can achieve it have the same structure as the ones used in the proof of Proposition \ref{lowerboundary}. 

\begin{Lemma} \label{lemma:contracts lb1}
Consider any $(t,u^b, u^g)\in[0,\tau]\times\widehat{\mathcal V}_{I-N_t}\times\widehat{\mathcal V}_{I-N_t}$ such that in addition $u^{b}=u^g$. Any contract $\Psi=(\theta,D)\in\Theta\times D$ such that $U_t^b(\Psi)=u^b$ and $U_t^g(\Psi)=u^g$, has no payments on $[t,\tau]$ and consequently both banks always shirk under $\Psi$.
\end{Lemma}

\begin{proof} Looking at the proof of \eqref{ineq1} we deduce that necessarily 
$
k^{\star,g}_s(\Psi)=k^{\star,b}_s(\Psi),~  dD_s=0, ~\forall s\geq t.
$ 
Since there are no payments, we have that $k_s^{\star,g}(\Psi)=k_s^{\star,b}(\Psi)=k_s^{\rm SH}$ for $s\in[t,\tau]$ and indeed have
\[
U_t^g(\Psi)= U_t^{b}(\Psi) = \E^{\P^{k^{\rm SH}}}\bigg[\int_t^\tau  \mathrm{e}^{-r(s-t)} B(I-N_s) \mathrm{d}s\bigg|\mathcal{G}_t\bigg].
\]
\end{proof}

\begin{Lemma} \label{lemma:contracts lb2}
Consider any $(t,u^g, u^{b})\in\mathbb R_+\times\widehat{\mathcal V}_{I-N_t}\times\widehat{\mathcal V}_{I-N_t}$ such that in addition
\[
u^g = \ds\frac{\rho_g}{\rho_b}u^{b}-\frac{(\rho_g-\rho_b)}{\rho_b}C(I-N_t).
\]
Under any contract $\Psi=(\theta,D)\in\Theta\times D$ such that $U_t^b(\Psi)=u^b$ and $U_t^g(\Psi)=u^g$, the pool is not liquidated until the last default $(\tau=\tau^I)$ and both banks always shirk on $[t,\tau]$.
\end{Lemma}

\begin{proof}
Looking at the proof of \eqref{ineq2}, we deduce that necessarily $
k_s^{\star,g}(\Psi)=k_s^{\star,b}(\Psi)=k_s^{\rm SH},~ \theta_s=1$, for every $s\geq t.$ Thus, the value functions of the banks are given by
\begin{align*}
U_t^g(\Psi) & = \rho_g \E^{\P^{k^{\rm SH}}}\bigg[\int_t^{\tau^I}  \mathrm{e}^{-r(s-t)} \mathrm{d}D_s\bigg|\mathcal{G}_t\bigg] + C(I-N_t), \ U_t^{b}(\Psi)  = \rho_b \E^{\P^{k^{\rm SH}}}\bigg[\int_t^{\tau^I}  \mathrm{e}^{-r(s-t)} \mathrm{d}D_s\bigg| \mathcal{G}_t\bigg]+ C(I-N_t).
\end{align*}
\end{proof}

\section{Technical results for the upper boundary}\label{sec:D}

\begin{Lemma} \label{lemma:theta=0} 
For every $j\geq 1$, $x_{j}^{\star}>\ds\frac{\rho_b}{\rho_g} \widehat b_j$. 
\end{Lemma}
\begin{proof}
For any $j\geq 1$, define the functions $g,h:\R\longrightarrow\R$ by 
\[
g(x):=x^{\frac{r+\widehat\lambda_j^{\rm SH}}{r+\widehat\lambda_j^0}}\widehat b_j\frac{r+\widehat\lambda_j^0}{r+\widehat\lambda_j^{\rm SH}}+\frac{Bj}{r+\widehat\lambda_j^{\rm SH}},~ h(x):=\widehat b_jx.
\]
Then $g$ is strictly convex in $\R_+$ and we have that $g(1)=h(1)=\widehat b_j$ and $g^\prime(1)=h^\prime(1)=\widehat b_j$. Thus, $h$ is the tangent line to $g$ at $x=1$ so $g(x)>h(x)$ for every $x\neq 1$ and therefore
\[
x_j^{\star}=g\bigg(\frac{\rho_b}{\rho_g}\bigg)>h\bigg(\frac{\rho_b}{\rho_g}\bigg)=\frac{\rho_b}{\rho_g}\widehat b_j.
\] 
\end{proof}

\begin{Proposition} 
For every $j\geq 1$, the function $\widehat{\mathcal U}_j^\star$ defined by \eqref{ujota} satisfies
\[
\frac{\widehat{\mathcal U}_j^\star(x)}{x}\leq \frac{\rho_g}{\rho_b}, \ \forall x\geq \frac{Bj}{r+\widehat\lambda_j^{\rm SH}}.
\]
Moreover, equality holds if and only if $x \geq \widehat b_j$.
\end{Proposition} 

\begin{proof} Define $A(x):=\ds\frac{\widehat{\mathcal U}_j^\star(x)}{x}$. If $x \geq \widehat b_{j-1}$ then $A(x)=\rho_g/\rho_b$. If now $x\in[x_j^{\star},\widehat b_j)$, we have 
\[
A(x)=\ds\frac{\rho_g}{\rho_b}(\widehat b_j)^{\frac{\widehat\lambda_j^{SH}-\widehat\lambda_j^0}{r+\widehat\lambda_j^{SH}}}\bigg(\frac{r+\widehat\lambda_j^{SH}}{r+\widehat\lambda_j^0}\bigg)^\frac{r+\widehat\lambda_j^0}{r+\widehat\lambda_j^{SH}}\frac1x \bigg(x-\frac{Bj}{r+\widehat\lambda_j^{SH}}\bigg)^\frac{r+\widehat\lambda_j^0}{r+\widehat\lambda_j^{SH}}.
\]
This function is decreasing so that $A$ reaches its maximum value over $[x_j^{\star},\widehat b_j)$ at $x_j^{\star}$. Next, we have
\[
A(x_j^{\star})=\frac{\widehat b_j}{x_j^{\star}}<\frac{\rho_g}{\rho_b}~\iff~ x_j^{\star}>\frac{\rho_b}{\rho_g}\widehat b_j,
\]
and the last inequality holds as a consequence of Lemma \ref{lemma:theta=0}. Finally, if $x\in\big[\frac{Bj}{r+\widehat\lambda_j^{\rm SH}},x_j^{\star}\big)$ then
\[
A(x)=\frac1x \bigg(\frac{\rho_g}{\rho_b}\bigg)^\frac{r+\widehat\lambda_j^{\rm SH}}{r+\widehat\lambda_j^0}\bigg(x-\frac{Bj}{r+\widehat\lambda_j^{\rm SH}}\bigg)+\frac1x \frac{Bj}{r+\widehat\lambda_j^{\rm SH}}.
\]
This function is increasing, hence $
A(x) \leq A(x_j^{\star})<\frac{\rho_g}{\rho_b}, \ \forall x\in\Big[ \frac{Bj}{r+\widehat\lambda_j^{SH}}, x^\star_j\Big].
$
\end{proof}

\begin{Corollary} \label{corollary:theta}
Let $j\geq 2$ and $\widehat{\mathcal U}_j^\star,$ $\widehat{\mathcal U}_{j-1}^\star$ defined by \eqref{ujota}, and assume that $\widehat\lambda_j^{k^g}\leq\widehat\lambda_j^{k^{b}}$. Then, for any $u^{b}\geq h^{1,b}+ \frac{B(j-1)}{r+\widehat\lambda_{j-1}^{SH}}$ we have 
\[
\widehat{\mathcal U}_{j-1}^\star(u^{b}-h^{1,b})\widehat\lambda_j^{k^g}-\big(\widehat{\mathcal U}_j^\star\big)^\prime(u^{b})\widehat\lambda_j^{k^{b}}(u^{b}-h^{1,b})\leq 0.
\]
Furthermore, equality holds if and only if $u^{b}-h^{1,b}\geq \widehat b_j$, $u^{b} \geq \widehat b_j$ and $\widehat\lambda_j^{k^{b}}=\widehat\lambda_j^{k^g}$.
\end{Corollary}

\begin{proof}
Under the conditions of the corollary, the following allows us to conclude immediately
\[
\frac{\widehat{\mathcal U}_{j-1}^\star(u^{b}-h^{1,b})}{u^{b}-h^{1,b}}\leq \frac{\rho_g}{\rho_b}\leq \big(\widehat{\mathcal U}_j^\star\big)^\prime(u^{b}).
\]
\end{proof}

\begin{Corollary}\label{corollary:ordered lambdas}   
For $j\geq 1$, let $\widehat{C}_j$ and $\widehat{\mathcal U}_j^\star$ be defined by \eqref{cjota} and \eqref{ujota} respectively. If $(\theta,h^{1,b})\in\widehat{C}_j$ is such that $u^{b}-\theta(u^{b}-h^{1,b})\geq \widehat b_j$ then  $\widehat{\mathcal U}_j^\star(u^{b})-\theta \widehat{\mathcal U}_{j-1}^\star(u^{b}-h^{1,b})\geq \widehat b_j$. As a consequence, in the context of equation \eqref{hjbj}, for every $(\theta,h^{1,b})\in\widehat{C}_j$ we have $k^g\leq k^{b}$ and $\widehat\lambda_j^{k^g}\leq\widehat\lambda_j^{k^{b}}$. 
\end{Corollary}

\begin{proof}
First observe that $u^{b}-\theta(u^{b}-h^{1,b})\geq \widehat b_j$ implies $u^{b}\geq \widehat b_j$. Then we have 
\[
\widehat{\mathcal U}_{j}^\star(u^{b})-\widehat b_j \geq \frac{\rho_g}{\rho_b}(u^{b}-\widehat b_j)\geq \frac{\widehat{\mathcal U}^\star_{j-1}(u^{b}-h^{1,b})}{u^{b}-h^{1,b}}(u^{b}-\widehat b_j).
\]
Also, $\theta\leq\ds\frac{u^{b}-\widehat b_j}{u^{b}-h^{1,b}}$ and thus
\[
\widehat{\mathcal U}_j^\star(u^{b})-\theta \widehat{\mathcal U}^\star_{j-1}(u^{b}-h^{1,b}) \geq \widehat{\mathcal U}_j^\star(u^{b}) - \bigg(\frac{u^{b}-\widehat b_j}{u^{b}-h^{1,b}}\bigg)\widehat{\mathcal U}^\star_{j-1}(u^{b}-h^{1,b}) \geq \widehat b_j.
\]
\end{proof}

\vspace{0.5em}
We now proceed with the

\vspace{0.5em}
\begin{proof}[Proof of Lemma \ref{lemma:diffusion}]
We start with the region $u^{b}<\widehat b_1,~ \widehat{\mathcal U}_1(u^{b})<\widehat b_1$. For these points, we have that $k^{b}=k^g=1$, so \eqref{hjb1} can be solved easily and leads to, for some $C_1\in\R$
\begin{align*}
 \widehat{\mathcal U}_1(u^{b}) & = C_1\bigg(u^{b}-\ds\frac{B}{r+\widehat\lambda_1^1}\bigg)+\ds\frac{B}{r+\widehat\lambda_1^1} .
\end{align*}
If $u^{b}<\widehat b_1$ and $ \widehat{\mathcal U}_1(u^{b})\geq \widehat b_1$, then $k^{b}=1$, $k^g=0$ and we can solve \eqref{hjb1} to obtain for some $C_2\in\R$
\begin{align*}
 \widehat{\mathcal U}_1(u^{b}) & = C_2\bigg(u^{b}-\ds\frac{B}{r+\widehat\lambda_1^1}\bigg)^\frac{r+\widehat\lambda_1^0}{r+\widehat\lambda_1^1}.
\end{align*}
Finally, when $u^{b}\geq \widehat b_1$ and $ \widehat{\mathcal U}(u^{b})\geq \widehat b_1$ the optimal strategies are $k^{b}=k^g=0$ and we have for some $C_3\in\R$, $\widehat{\mathcal U}_1(u^{b}) = C_3  u^{b}.$ We are interested in smooth solutions of \eqref{hjb1}. Denote by $ \widehat{\mathcal U}_1^{(1)}$, $ \widehat{\mathcal U}_1^{(2)}$ and $ \widehat{\mathcal U}_1^{(3)}$ the following functions
\begin{align*}
\widehat{\mathcal U}_1^{(1)}(u^{b})&:= C_1\bigg(u^{b}-\ds\frac{B}{r+\widehat\lambda_1^1}\bigg)+\ds\frac{B}{r+\lambda_1^1}, \
\widehat{\mathcal U}_1^{(2)}(u^{b}):= C_2\bigg(u^{b}-\ds\frac{B}{r+\widehat\lambda_1^1}\bigg)^\frac{r+\widehat\lambda_1^0}{r+\widehat\lambda_1^1}, \
\widehat{\mathcal U}_1^{(3)}(u^{b}):=C_3 u^{b}.
\end{align*}
We will determine the relations between the constants which allow the smooth fitting of $\widehat{\mathcal U}_1$. First we impose $\widehat{\mathcal U}_1^{(2)}(\widehat b_1)=\widehat{\mathcal U}_1^{(3)}(\widehat b_1)$ and we get  
\[
C_2\bigg(\widehat b_1\ds\frac{r+\widehat\lambda_1^0}{r+\widehat\lambda_1^1}\bigg)^{\frac{r+\widehat\lambda_1^0}{r+\widehat\lambda_1^1}}=C_3 \widehat b_1.
\]
It can be checked that this relation between $C_1$ and $C_2$ ensures also that $({\widehat{\mathcal U}_1^{(2)}})^{\prime}(\widehat b_1)=(\widehat{\mathcal U}_1^{(3)})^{\prime}(\widehat b_1)$. Next, define $x_1$ as the point such that $\widehat{\mathcal U}_1^{(1)}(x_1)=\widehat b_1$, {\it i.e.}
\[
x_1=\frac{\widehat b_1}{C_1}\bigg( \frac{r+\widehat\lambda_1^0}{r+\widehat\lambda_1^1} \bigg) +\ds\frac{B}{r+\widehat\lambda_1^1}.
\]
Also, define $x_2$ as the point such that $\widehat{\mathcal U}_1^{(2)}(x_2)=\widehat b_1$, {\it i.e.}
\[
x_2=\bigg(\frac{\widehat b_1}{C_2}\bigg)^\frac{r+\widehat\lambda_1^1}{r+\widehat\lambda_1^0}+\ds\frac{B}{r+\widehat\lambda_1^1}.
\]
We impose $x_1=x_2$ and we get 
\[
\frac{\widehat b_1}{C_1}\bigg( \frac{r+\widehat\lambda_1^0}{r+\widehat\lambda_1^1} \bigg) = \bigg(\frac{\widehat b_1}{C_2}\bigg)^{\frac{r+\widehat\lambda_1^1}{r+\widehat\lambda_1^0}},
\]
and this relation ensures also that $(\widehat{\mathcal U}_1^{(1)})^\prime(x_1)=(\widehat{\mathcal U}_1^{(2)})^\prime(x_2)$. Expressing both $C_1$ and $C_2$ in terms of $C_3$ we get $\widehat{\mathcal U}_1^{(3)}(u^{b})  = C_3  u^{b}$, and
\begin{align*}
\widehat{\mathcal U}_1^{(1)}(u^{b}) & =  {C_3}^\frac{r+\widehat\lambda_1^1}{r+\widehat\lambda_1^0}\bigg(u^{b}-\ds\frac{B}{r+\widehat\lambda_1^1}\bigg) +\ds\frac{B}{r+\widehat\lambda_1^1}, \ \widehat{\mathcal U}_1^{(2)}(u^{b})  =  C_3  {\widehat b_1}^\frac{\widehat\lambda_1^1-\widehat\lambda_1^0}{r+\widehat\lambda_1^1}\bigg(\ds\frac{r+\widehat\lambda_1^1}{r+\widehat\lambda_1^0}\bigg)^\frac{r+\widehat\lambda_1^0}{r+\widehat\lambda_1^1}\bigg(u^{b}-\ds\frac{B}{r+\widehat\lambda_1^1}\bigg)^\frac{r+\widehat\lambda_1^0}{r+\widehat\lambda_1^1}.
\end{align*}
\end{proof}

\vspace{0.5em}
We pursue with the

\vspace{0.5em}
\begin{proof}[Proof of Lemma \ref{lemma:HJB1}]
For $C_0>0$, define the following modification $\widehat{\mathcal U}^{C_0,\star}_1$ of $\widehat{\mathcal U}^{C_0}_1$
\[
\widehat{\mathcal U}^{C_0,\star}_1(u^{b}):=\begin{cases}
\displaystyle \widehat{\mathcal U}_1^{C_0}(u^{b}) , \  u^{b} \leq x_1^{C_0,\star} , \\[0.5em] 
\displaystyle \frac{\rho_g}{\rho_b}(u^{b}-x_1^{C_0,\star})+ \widehat{\mathcal U}_1^{C_0}(x_1^{C_0,\star}),\ u^{b} \geq x_1^{C_0,\star},
\end{cases}
\]
where
\[
x_1^{C_0,\star}:=\ds\inf\bigg\{u^{b}\in\bigg[\frac{B}{r+\widehat\lambda_1^1},+\infty\bigg): \big(\widehat{\mathcal U}^{C_0}_1\big)^\prime(u^{b})\leq\frac{\rho_g}{\rho_b}\bigg\} .
\]
The function $\widehat{\mathcal U}^{C_0,\star}_1$ is continuously differentiable, solves the diffusion equation in $[B/(r+\widehat\lambda_1^1),x_1^{C_0,\star})$ and satisfies $\big({\widehat{\mathcal U}_1^{C_0,\star}}\big)^{\prime}=\rho_g/\rho_b$ in $(x_1^{C_0,\star},\infty)$. In the following we will study for which values of $C_0$ this function indeed solves the HJB equation. 

\vspace{0.5em}
\hspace{2em}$-$ First of all, if $C_0^\frac{r+\widehat\lambda_1^1}{r+\widehat\lambda_1^0}\leq \frac{\rho_g}{\rho_b}$, we have that 
\[
x_1^{C_0,\star}=\frac{B}{r+\widehat\lambda_1^1},\ \widehat{\mathcal U}_1^{C_0,\star}(u^{b})=\frac{\rho_g}{\rho_b}\left(u^{b}-\frac{B}{r+\widehat\lambda_1^1}\right)+\frac{B}{r+\widehat\lambda_1^1},\ \left({\widehat{\mathcal U}_1^{C_0,\star}}\right)^\prime(u^{b})\rho_b-\rho_g=0,
\]
so that we need to check that for every $u^{b}$ in $[B/(r+\widehat\lambda_1^1),\infty)$
\[
r  \widehat{\mathcal U}_1^{C_0,\star}(u^{b}) -\big({ \widehat{\mathcal U}_1^{C,\star}}\big)^\prime(u^{b})\big(ru^{b}-Bk^{b}+u^{b}\widehat\lambda_1^{k^{b}} \big) + \widehat{\mathcal U}_1^{C,\star}(u^{b})\widehat\lambda_1^{k^g} -Bk^g \geq0.
\]
Take $u^{b}>\widehat b_1$. Then $k^g=k^{b}=0$, and we have
\begin{align*}
& r  \widehat{\mathcal U}_1^{C_0,\star}(u^{b}) -\big({ \widehat{\mathcal U}_1^{C_0,\star}}\big)^\prime(u^{b})\big(ru^{b}-Bk^{b}+u^{b}\widehat\lambda_1^{k^{b}} \big) + \widehat{\mathcal U}_1^{C_0,\star}(u^{b})\widehat\lambda_1^{k^g} -Bk^g \\
&=  r\bigg(\frac{\rho_g}{\rho_b}\bigg(u^{b}-\frac{B}{r+\widehat\lambda_1^1}\bigg)+\frac{B}{r+\widehat\lambda_1^1}\bigg )-\frac{\rho_g}{\rho_b}\big((r+\widehat\lambda_1^0)u^{b}\big)+\widehat\lambda_1^0\bigg(\frac{\rho_g}{\rho_b}\bigg(u^{b}-\frac{B}{r+\widehat\lambda_1^1}\bigg)+\frac{B}{r+\widehat\lambda_1^1} \bigg) =  (r+\widehat\lambda_1^0)\frac{B}{r+\widehat\lambda_1^1}\bigg(1-\frac{\rho_g}{\rho_b} \bigg) <0 .
\end{align*}
Hence $ \widehat{\mathcal U}_1^{C_0,\star}$ is not a solution of \eqref{dp1}. 

\vspace{0.5em}
\hspace{2em}$-$ If $\left(\frac{\rho_g}{\rho_b}\right)^\frac{r+\widehat\lambda^0_1}{r+\widehat\lambda_1^1}<C_0\leq \frac{\rho_g}{\rho_b} $, then $
x_1^{C_0,\star}=\widehat b_1\frac{r+\widehat\lambda^0_1}{r+\widehat\lambda_1^1}\big(C_0\rho_b/\rho_g\big)^\frac{r+\widehat\lambda_1^1}{\widehat\lambda_1^1-\widehat\lambda^0_1}+\frac{B}{r+\widehat\lambda_1^1}.
$
Take $u^{b}>\widehat b_1$, then $k^g=k^{b}=0$ and
\begin{align*}
 r  \widehat{\mathcal U}_1^{C,\star}(u^{b}) -\big({ \widehat{\mathcal U}_1^{C_0,\star}}\big)^\prime(u^{b})\big(ru^{b}-Bk^{b}+u^{b}\widehat\lambda_1^{k^{b}}\big) + \widehat{\mathcal U}_1^{C,\star}(u^{b})\widehat\lambda_1^{k^g} -Bk^g 
& =  (r+\widehat\lambda_1^0)\bigg(\widehat b_1 C^\frac{r+\widehat\lambda_1^1}{\widehat\lambda_1^1-\widehat\lambda^0_1}\bigg(\frac{\rho_b}{\rho_g}\bigg)^\frac{r+\widehat\lambda^0_1}{\widehat\lambda_1^1-\widehat\lambda^0_1}\frac{\widehat\lambda_1^1-\widehat\lambda^0_1}{r+\widehat\lambda_1^1}-\frac{\rho_g}{\rho_b}\frac{B}{r+\widehat\lambda_1^1}\bigg) \\
& \leq  (r+\widehat\lambda^0_1)\bigg(\widehat b_1\frac{\rho_g}{\rho_b}\frac{\widehat\lambda_1^1-\widehat\lambda^0_1}{r+\widehat\lambda_1^1}-\frac{B}{r+\widehat\lambda_1^1}\frac{\rho_g}{\rho_b}\bigg) = 0.
\end{align*}
The inequality is strict if $C_0<\frac{\rho_g}{\rho_b}$ so the only value of $C_0$ such that $ \widehat{\mathcal U}_1^{C_0,\star}$ solves the HJB equation is $C_0=\frac{\rho_g}{\rho_b}$. 

\vspace{0.5em}
\hspace{2em}$-$ For large values of $C_0$, {\it i.e.} $C_0>\frac{\rho_g}{\rho_b}$, we have that $x_1^{C_0,\star}=+\infty$ and then $ \widehat{\mathcal U}_1^{C_0,\star}=\widehat{\mathcal U}_1^{C_0}$. We exclude this case because these functions do not satisfy condition \eqref{saturation}. 
\end{proof}

\vspace{0.5em}
We end this section with the

\vspace{0.5em}
\begin{proof}[Proof of Proposition \ref{prop:uj functions}]
The proof is by induction. For $j=1$ the result is proved in Step 2, so we take any $j>1$ and assume that $\widehat{\mathcal U}_{j-1}^\star$ solves its corresponding diffusion equation. We will need to consider three different cases to prove that $\widehat{\mathcal U}_j^\star$ solves the equation (\ref{hjbj}). In each one of them we prove that the supremum in the right--hand side of (\ref{hjbj}) is attained with $\theta=0$, so that the diffusion equation takes the same form as the one in the case with one loan left. Then, it follows from the analysis in Step 2 that its solution satisfies also the variational inequality \eqref{DPequation}. 

\vspace{0.5em}
\hspace{2em}{$-$ \bf Case $1$:} $u^{b}<\widehat b_j$,  $\widehat{\mathcal U}_j^\star(u^{b})<\widehat b_j$. 

\vspace{0.5em}
In this case for any $(\theta,h^1)\in\widehat{C}^j$, we have that $k^g=k^{b}=j$. To ease notations, define $c_j(u^{b}):=\big(\widehat{\mathcal U}_j^\star\big)^\prime(u^{b})
\big(ru^{b} - Bj + u^{b}\widehat\lambda_j^{\rm SH} \big)$. Then the term inside the supremum in (\ref{hjbj}) becomes 
\[ 
c_j(u^{b})-\widehat{\mathcal U}_j^\star(u^{b})\widehat\lambda_j^{\rm SH}+Bj +\theta\widehat\lambda_j^{\rm SH}\big( \widehat{\mathcal U}_{j-1}^\star(u^{b}-h^1)-\big(\widehat{\mathcal U}_j^\star\big)^\prime(u^{b})(u^{b}-h^1)\big),
\]
and the optimal choice of $\theta$ in this case is $0$ (uniquely) because thanks to Corollary \ref{corollary:theta} we have
\[
\widehat{\mathcal U}_{j-1}^\star(u^{b}-h^1)-\big(\widehat{\mathcal U}_j^\star\big)^\prime(u^{b})(u^{b}-h^1)< 0.
\]

\hspace{2em}{$-$ \bf Case $2$:} $u^{b}<\widehat b_j$, $\widehat{\mathcal U}_j^\star(u^{b})\geq \widehat b_j$. 

\vspace{0.5em}
In this case $k^{b}=j$ for every $(\theta,h^1)\in\widehat{C}^j$. The term inside the supremum in (\ref{hjbj}) becomes 
\[
c_j(u^{b}) - \widehat{\mathcal U}_j^\star(u^{b})\widehat\lambda_j^{k^g}+Bk^g + \theta \big( \widehat{\mathcal U}_{j-1}^\star(u^{b}-h^1)\widehat\lambda_j^{k^g}-\big(\widehat{\mathcal U}_j^\star\big)^\prime(u^{b})\widehat\lambda_j^{\rm SH}(u^{b}-h^1) \big).
\]
Define the following sets
\[
\widehat{C}_j^0:=\big\{(\theta,h^1)\in\widehat{C}^j: \widehat{\mathcal U}_j^\star(u^{b})-\theta \widehat{\mathcal U}_{j-1}^\star(u^{b}-h^1)\geq \widehat b_j\big\},\;
\widehat{C}_j^j:=\big\{(\theta,h^1)\in\widehat{C}^j: \widehat{\mathcal U}_j^\star(u^{b})-\theta \widehat{\mathcal U}_{j-1}^\star(u^{b}-h^1)<\widehat b_j\big\},
\]
and note that $k^g=0$ for every $(\theta,h^1)\in\widehat{C}_j^0$ and $k^g=j$ for every $(\theta,h^1)\in\widehat{C}_j^j$. Also, the pair $(0,h^1)$ belongs to $\widehat{C}_j^0$ for every feasible $h^1$. 

\vspace{0.5em}
\hspace{3em}$\bullet$ If $(\theta,h^1)\in\widehat{C}_j^0$ we have 
\begin{align*}
& c_j(u^{b}) - \widehat{\mathcal U}_j^\star(u^{b})\widehat\lambda^{k^g}_j+Bk^g+\theta\big( \widehat{\mathcal U}_{j-1}^\star(u^{b}-h^1)\widehat\lambda^{k^g}_j-\big(\widehat{\mathcal U}_j^\star\big)^\prime(u^{b})\widehat\lambda_j^{\rm SH}(u^{b}-h^1)\big) \\
&=  c_j(u^{b}) - \widehat{\mathcal U}_j^\star(u^{b})\widehat\lambda_j^0+\theta\big( \widehat{\mathcal U}_{j-1}^\star(u^{b}-h^1)\widehat\lambda_j^0-\big(\widehat{\mathcal U}_j^\star\big)^\prime(u^{b})\widehat\lambda_j^{\rm SH}(u^{b}-h^1)\big)  \leq   c_j(u^{b}) - \widehat{\mathcal U}_j^\star(u^{b})\widehat\lambda_j^0,
\end{align*}
where the inequality is due to Corollary \ref{corollary:theta}.  

\vspace{0.5em}
\hspace{3em}$\bullet$ If $(\theta,h^1)\in\widehat{C}_j^j$ we have
\begin{align*}
& c_j(u^{b}) - \widehat{\mathcal U}_j^\star(u^{b})\widehat\lambda^{k^g}_j+Bk^g+\theta\big( \widehat{\mathcal U}_{j-1}^\star(u^{b}-h^1)\widehat\lambda_j^{k^g}-\big(\widehat{\mathcal U}_j^\star\big)^\prime(u^{b})\widehat\lambda_j^{\rm SH}(u^{b}-h^1)\big) \\
&=  c_j(u^{b}) - \widehat{\mathcal U}_j^\star(u^{b})\widehat\lambda_j^{\rm SH}+Bj+\theta\big( \widehat{\mathcal U}_{j-1}^\star(u^{b}-h^1)\widehat\lambda_j^{\rm SH}-\big(\widehat{\mathcal U}_j^\star\big)^\prime(u^{b})\widehat\lambda_j^{\rm SH}(u^{b}-h^1)\big) \\
&<  c_j(u^{b}) - \widehat{\mathcal U}_j^\star(u^{b})\widehat\lambda_j^{\rm SH} + Bj \\
&=  c_j(u^{b}) - \widehat{\mathcal U}_j^\star(u^{b})\widehat\lambda_j^{\rm SH} + \widehat b_j(\widehat\lambda_j^{SH}-\widehat\lambda_j^0) \leq  c_j(u^{b}) - \widehat{\mathcal U}_j^\star(u^{b})\widehat\lambda_j^{\rm SH} + \widehat{\mathcal U}_j^\star(u^{b})\big(\widehat\lambda_j^{\rm SH}-\widehat\lambda_j^0\big) =  c_j(u^{b}) - \widehat{\mathcal U}_j^\star(u^{b})\widehat\lambda_j^0,
\end{align*}
where the first inequality is a consequence of Corollary \ref{corollary:theta} and the second one holds because $\widehat{\mathcal U}_j^\star(u^{b})\geq \widehat b_j$.  
So we conclude that the optimal value for $\theta$ in this case is also 0 (uniquely). 

\vspace{0.5em}
\hspace{2em}{$-$ \bf Case $3$:} $u^{b}\geq \widehat b_j$, $\widehat{\mathcal U}_j^\star(u^{b})\geq \widehat b_j$. 

\vspace{0.5em}
Thanks to Proposition \ref{corollary:ordered lambdas} , we know that there are only three possibilities for the value of $(k^{b},k^g)$.  Define the sets 
\begin{align*}
\widehat{C}_j^{0,0} & :=\left\{(\theta,h^1)\in\widehat{C}^j: u^{b}-\theta(u^{b}-h^1)\geq \widehat b_j,~ \widehat{\mathcal U}_j^\star(u^{b})-\theta \widehat{\mathcal U}^\star_{j-1}(u^{b}-h^1)\geq \widehat b_j\right\}, \\
\widehat{C}_j^{j,0} & :=\left\{(\theta,h^1)\in\widehat{C}^j: u^{b}-\theta(u^{b}-h^1)<\widehat b_j,~ \widehat{\mathcal U}_j^\star(u^{b})-\theta \widehat{\mathcal U}^\star_{j-1}(u^{b}-h^1)\geq \widehat b_j\right\}, \\
\widehat{C}_j^{j,j} & :=\left\{(\theta,h^1)\in\widehat{C}^j: u^{b}-\theta(u^{b}-h^1)< \widehat b_j,~ \widehat{\mathcal U}_j^\star(u^{b})-\theta \widehat{\mathcal U}^\star_{j-1}(u^{b}-h^1)<\widehat b_j\right\}.
\end{align*}
Then, $(k^{b},k^g)=(0,0)$ for every $(\theta,h^1)\in\widehat{C}_j^{0,0}$, $(k^{b},k^g)=(j,0)$ for every $(\theta,h^1)\in\widehat{C}_j^{j,0}$ and $(k^{b},k^g)=(j,j)$ for every $(\theta,h^1)\in\widehat{C}_j^{j,j}$. Also, $(0,h^1)$ belongs to $\widehat{C}_j^{0,0}$ for any feasible $h^1$. 

\vspace{0.5em}
\hspace{3em}$\bullet$ If $(\theta,h^1)\in\widehat{C}_j^{0,0}$ then the term inside the supremum in (\ref{hjbj}) is, because of Corollary \ref{corollary:theta}, equal to
\begin{align*}
& \big(\widehat{\mathcal U}_j^\star\big)^\prime(u^{b})u^{b}\big(r+\widehat\lambda_j^{0} \big)-\widehat{\mathcal U}_j^\star(u^{b})\widehat\lambda_j^{0}+\theta\widehat\lambda_j^{0}\big( \widehat{\mathcal U}^\star_{j-1}(u^{b}-h^1)-\big(\widehat{\mathcal U}_j^\star\big)^\prime(u^{b})(u^{b}-h^1)\big)  \leq   \big(\widehat{\mathcal U}_j^\star\big)^\prime(u^{b})u^{b}\big(r+\widehat\lambda_j^{0} \big)-\widehat{\mathcal U}_j^\star(u^{b})\widehat\lambda_j^{0},
\end{align*}

\hspace{3em}$\bullet$ If $(\theta,h^1)\in\widehat{C}_j^{j,0}$, then $h^1<\widehat b_j$ and 
$
\frac{u^{b}-\widehat b_j}{u^{b}-h^1}<\theta\leq\frac{\widehat{\mathcal U}_j^\star(u^{b})-\widehat b_j}{\widehat{\mathcal U}^\star_{j-1}(u^{b}-h^1)}$. The term in the supremum in (\ref{hjbj}) is
\begin{align*}
 & c_j(u^{b}) - \widehat{\mathcal U}_j^\star(u^{b})\widehat\lambda_j^0+\theta\big( \widehat{\mathcal U}^\star_{j-1}(u^{b}-h^1)\widehat\lambda_j^0-\big(\widehat{\mathcal U}_j^\star\big)^\prime(u^{b})\widehat\lambda_j^{\rm SH}(u^{b}-h^1)\big) \\
&<  c_j(u^{b}) - \widehat{\mathcal U}_j^\star(u^{b})\widehat\lambda_j^0+\bigg(\frac{u^{b}-\widehat b_j}{u^{b}-h^1}\bigg)\big( \widehat{\mathcal U}^\star_{j-1}(u^{b}-h^1)\widehat\lambda_j^0-\big(\widehat{\mathcal U}_j^\star\big)^\prime(u^{b})\widehat\lambda_j^{\rm SH}(u^{b}-h^1)\big) \\
&\leq  c_j(u^{b})-\widehat{\mathcal U}_j^\star(u^{b})\lambda_j^0+(u^{b}-\widehat b_j)\big( \big(\widehat{\mathcal U}_j^\star\big)^\prime(u^{b})\widehat\lambda_j^0-\big(\widehat{\mathcal U}_j^\star\big)^\prime(u^{b})\widehat\lambda_j^{\rm SH}\big) =  \big(\widehat{\mathcal U}_j^\star\big)^\prime(u^{b})\big(ru^{b}+u^{b}\widehat\lambda_j^{0} \big)-\widehat{\mathcal U}_j^\star(u^{b})\widehat\lambda_j^0.
\end{align*}
Both inequalities are direct consequences of Corollary \ref{corollary:theta}. 

\vspace{0.5em}
\hspace{3em}$\bullet$ Finally, if $(\theta,h^1)\in\widehat{C}_j^{j,j}$, note that $h^1<\widehat b_j$, $\widehat{\mathcal U}_j^\star(u^{b})-\widehat{\mathcal U}_{j-1}^\star(u^{b}-h^1)<\widehat b_j$ and 
\[
\frac{u^{b}-\widehat b_j}{u^{b}-h^1}\leq\frac{\widehat{\mathcal U}_j^\star(u^{b})-\widehat b_j}{\widehat{\mathcal U}^\star_{j-1}(u^{b}-h^1)}<\theta.
\]
Then, the term inside the sup in (\ref{hjbj}) becomes 
\begin{align*}
& c_j(u^{b}) - \widehat{\mathcal U}_j^\star(u^{b})\widehat\lambda_j^{SH}+Bj+\theta\widehat\lambda_j^{\rm SH}\big( \widehat{\mathcal U}^\star_{j-1}(u^{b}-h^1)-\big(\widehat{\mathcal U}_j^\star\big)^\prime(u^{b})(u^{b}-h^1)\big) \\[-0.05em]
&\leq  c_j(u^{b}) -\widehat{\mathcal U}_j^\star(u^{b})\widehat\lambda_j^{\rm SH}+Bj + \frac{\widehat{\mathcal U}_j^\star(u^{b})-\widehat b_j}{\widehat{\mathcal U}^\star_{j-1}(u^{b}-h^1)}\widehat\lambda_j^{\rm SH}\left( \widehat{\mathcal U}^\star_{j-1}(u^{b}-h^1)-\big(\widehat{\mathcal U}_j^\star\big)^\prime(u^{b})(u^{b}-h^1)\right)  \\[-0.05em]
&\leq  c_j(u^{b}) - \widehat{\mathcal U}_j^\star(u^{b})\widehat\lambda_j^{\rm SH}+Bj+\widehat\lambda_j^{\rm SH}\bigg( \widehat{\mathcal U}_j^\star(u^{b})-\widehat b_j-\big(\widehat{\mathcal U}_j^\star\big)^\prime(u^{b})\frac{\widehat{\mathcal U}_j^\star(u^{b})-\widehat b_j}{\frac{\rho_g}{\rho_b}}\bigg)  \\[-0.05em]
&=  c_j(u^{b}) - \widehat b_j\widehat\lambda_j^0+\widehat\lambda_j^{\rm SH}\bigg(-\frac{\rho_b}{\rho_g}\big(\widehat{\mathcal U}_j^\star\big)^\prime(u^{b})\widehat{\mathcal U}_j^\star(u^{b})+\frac{\rho_b}{\rho_g}\big(\widehat{\mathcal U}_j^\star\big)^\prime(u^{b})\widehat b_j \bigg)  \\[-0.05em]
&= \widehat \lambda_j^{\rm SH}\big(\widehat{\mathcal U}_j^\star\big)^\prime(u^{b})\bigg(u^{b}-\frac{\rho_b}{\rho_g}\widehat{\mathcal U}_j^\star(u^{b})\bigg)+\big(\widehat{\mathcal U}_j^\star\big)^\prime(u^{b})\bigg(ru^{b}+\frac{\rho_b}{\rho_g}\widehat\lambda_j^{\rm SH}\widehat b_j-Bj\bigg)-\widehat\lambda_j^0\widehat b_j. 
\end{align*} 

The first inequality comes from Corollary \ref{corollary:theta} and the second one from the fact that the map $h^1\longmapsto\widehat{\mathcal U}^\star_{j-1}(u^{b}-h^1)/(u^{b}-h^1)$ is non--decreasing and constant for large values of $h^1$, which implies that $\widehat{\mathcal U}^\star_{j-1}(u^{b}-h^1)/(u^{b}-h^1)\leq\rho_g/\rho_b$. Now we use the explicit form of $\widehat{\mathcal U}_j^\star$ and compute
\begin{align*}
& \widehat\lambda_j^{\rm SH}\big(\widehat{\mathcal U}_j^\star\big)^\prime(u^{b})\bigg(u^{b}-\frac{\rho_b}{\rho_g}\widehat{\mathcal U}_j^\star(u^{b})\bigg)+\big(\widehat{\mathcal U}_j^\star\big)^\prime(u^{b})\bigg(ru^{b}+\frac{\rho_b}{\rho_g}\widehat\lambda_j^{\rm SH}\widehat b_j-Bj\bigg)-\widehat\lambda_j^0\widehat b_j \\
& =  \ds\frac{\rho_g}{\rho_b} r u^{b}+ \widehat\lambda_j^{\rm SH}\widehat b_j - \frac{\rho_g}{\rho_b}  Bj- \widehat\lambda_j^0 \widehat b_j  =  \frac{\rho_g}{\rho_b} r u^{b}+Bj\bigg(1-\frac{\rho_g}{\rho_b}\bigg)  <  \frac{\rho_g}{\rho_b} r u^{b}.
\end{align*}
The term in the last line corresponds to $\big(\widehat{\mathcal U}_j^\star\big)^\prime(u^{b})\big(ru^{b}+u^{b}\widehat\lambda_j^{0} \big)-\widehat{\mathcal U}_j^\star(u^{b})\widehat\lambda_j^0$ and therefore the optimal $\theta$ in this case is also $0$. Observe that in this case every $(\theta,h^1)\in \widehat{C}_j^{0,0}$ such that $u^{b}-h^1 \geq \widehat b_j$ is optimal.  
\end{proof}

\vspace{0.5em}
We next continue with the

\vspace{0.5em}
\begin{proof}[Proof of Theorem \ref{verification theorem}] We divide the proof in $3$ steps.

\vspace{0.5em}
\hspace{3em}{$\bullet$ \bf Step $1$:} Let us prove first that the SDE (\ref{bcutility}) has a unique solution, keeping in mind that $\Psi^\star$ liquidates the pool immediately after the first default. We consider two cases: if $u^{b}<\widehat b_{I-N_t}$, by right--continuity we can find for every solution of (\ref{bcutility}) some $\varepsilon\in(0,\tau-t)$ such that $u_s^{b}<\widehat b_{I-N_t}$ for $s\in[t,t+\varepsilon]$. Consequently $u^{b}$ solves the ODE 
\[
du_s^{b}=\big((r+\widehat\lambda_{I-N_t}^{\rm SH})u^{b}_s-B(I-N_t)\big)\mathrm{d}s,~s\in[t,t+\varepsilon],
\] 
whose unique solution is given by 
\[
u_s^{b}=\ds \mathrm{e}^{(r+\widehat\lambda_{I-N_t}^{SH})(s-t)}\bigg(u^{b}-\frac{B(I-N_t)}{r+\widehat\lambda_{I-N_t}^{\rm SH}}\bigg)+\frac{B(I-N_t)}{r+\widehat\lambda_{I-N_t}^{\rm SH}},~s\in[t,t+\varepsilon].
\] 
So, as long as there is no default and the project keeps running $u_s^{b}$ will be deterministic until it reaches the value $\widehat b_{I-N_t}$. That will eventually happen at time 
\[
t^{\star}(u^{b}):=t+\frac{1}{r+\widehat\lambda_{I-N_t}^{\rm SH}}\log\bigg(\frac{\widehat b_{I-N_t}(r+\widehat\lambda_{I-N_t}^0)}{u^{b}(r+\widehat\lambda_{I-N_t}^{\rm SH})-B(I-N_t)}\bigg),
\] 
and we see from \eqref{bcutility} that at time $t^{\star}(u^{b})$ we will have $\mathrm{d}u_s^{b}=0$, so $u_s^{b}=\widehat b_{I-N_t}$ for every $s\in[t^{\star}(u^{b}),\tau)$. In the second case, if $u^{b}\geq \widehat b_{I-N_t}$ then \eqref{bcutility} becomes $
\mathrm{d}u_s^{b}=-u_{s^-}^{b} \mathrm{d}N_s,~s\in[t,\tau],
$
and necessarily $u_s^{b}=u^{b}$ for every $s\in[t,\tau)$. This proves the existence and uniqueness of the solution of \eqref{bcutility} in both cases. It is then immediate that the first integrability condition in \eqref{eq:integ} is satisfied.

\vspace{0.5em}
\hspace{3em}{$\bullet$ \bf Step $2$:} Now we turn to the values of the banks under $\Psi^\star$. If $u^{b}\geq \widehat b_{I-N_t}$, we know from the previous analysis that $u_s^{b}=u^{b}\geq \widehat b_{I-N_t}$ for every $s\in[t,\tau)$, so in this case $\Psi^\star$ is a short--term contract with constant payment, see Section \ref{constantpayment}. Using the notations of that section, since the payment $c=\frac{u^{b}(r+\widehat\lambda_j^0)}{\rho_b}$ is such that $c\geq\bar{c}_b\geq\bar{c}_g$ both banks will always work, the value function of the bad bank is $U_t^{b}(\Psi^\star)=\rho_b c/(r+\widehat\lambda_{I-N_t}^0)=u^{b}$ and the one of the good bank is $U_t^g(\Psi^\star)=\rho_g c/(r+\widehat\lambda_{I-N_t}^0)=\rho_g/\rho_bu^{b}=\widehat{\mathcal U}_{I-N_t}^{\star}(u^{b})$. 

\vspace{0.5em}
In the case where $u^{b}<\widehat b_{I-N_t}$, $\Psi^\star$ is a short--term contract with delay $t^{\star}(u^{b})$ and constant payment, see Section \ref{contractswithdelay}. Using the notations of that section, since $c=\bar{c}_b$ the bad bank will always shirk and her value function is 
\[
U_t^{b}(\Psi^\star)=\rho_b c \frac{\mathrm{e}^{-(r+\widehat\lambda_{I-N_t}^{\rm SH})t^{\star}(u^{b})}}{r+\widehat\lambda_{I-N_t}^{\rm SH}}+\frac{B}{r+\widehat\lambda_{I-N_t}^{\rm SH}}=u^{b}.
\] 

For the good bank we have two sub--cases. First, if $u^{b}\in[x_{I-N_t}^{\star},\widehat b_{I-N_t})$ then $\bar{t}_g(c)\geq t^{\star}(u^{b})$, so the good bank will always work and her value function is
\begin{align*}
U_t^g(\Psi^\star)  
& = \frac{\rho_g}{\rho_b}{\widehat b_{I-N_t}}^\frac{\widehat\lambda_{I-N_t}^{\rm SH}-\widehat\lambda_{I-N_t}^0}{r+\widehat\lambda_{I-N_t}^{\rm SH}}\bigg(\frac{r+\widehat\lambda_{I-N_t}^{\rm SH}}{r+\widehat\lambda_{I-N_t}^0}\bigg)^\frac{r+\widehat\lambda_{I-N_t}^0}{r+\widehat\lambda_{I-N_t}^{\rm SH}}\bigg(u^{b}-\frac{B(I-N_t)}{r+\widehat\lambda_{I-N_t}^{\rm SH}}\bigg)^\frac{r+\widehat\lambda_{I-N_t}^0}{r+\widehat\lambda_{I-N_t}^{\rm SH}}=\widehat{\mathcal U}_{I-N_t}^{\star}(u^{b}).
\end{align*}
If $u^{b}\in\big[\frac{B}{r+\widehat\lambda_{I-N_t}^{\rm SH}},x_{I-N_t}^{\star}\big)$ then $\bar{t}_g(c)< t^{\star}(u^{b})$, so the good bank will start working at time $t^{\star}(u^{b})$ and her value function is 
\begin{align*}
U_t^g(\Psi^\star) 
 & =  \frac{\rho_g}{\rho_b}^\frac{r+\widehat\lambda_{I-N_t}^{\rm SH}}{r+\widehat\lambda_{I-N_t}^0}\bigg(u^{b}-\ds\frac{B(I-N_t)}{r+\widehat\lambda_{I-N_t}^{\rm SH}}\bigg) +\ds\frac{B(I-N_t)}{r+\widehat\lambda_{I-N_t}^{SH}} = \widehat{\mathcal U}^\star_{I-N_t}(u^{b}). 
\end{align*}  

\vspace{0.5em}
\hspace{3em}{$\bullet$ \bf Step $3$:} Since $U_t^{b}(\Psi^\star)=u^{b}$, we have $\Psi^{\star}\in\overline{\mathcal A}^{b}(t,u^{b})$. Consider now a contract $\Psi=(D,\theta,h^{1,b},h^{2,b})\in\overline{\mathcal A}^{b}(t,u^{b})$. We recall that the value function of the bad bank under $\Psi$ satisfies
\begin{equation*}
\mathrm{d}U_s^{b}(\Psi)=\Big( rU_s^{b}(\Psi)-Bk_s^{\star,b}(\Psi) + \big(h_s^{1,b} + h_s^{2,b}(1-\theta_s)\big)\lambda_s^{k^{\star,b}(\Psi)} \Big)\mathrm{d}s -\rho_b \mathrm{d}D_s-h_s^{1,b}\mathrm{d}N_s-h_s^{2,b}\mathrm{d}H_s,
\end{equation*}
with $
k_s^{\star,b}(\Psi) = {\bf 1}_{\{ h_s^{1,b} + (1-\theta_s)h_s^{2,b} < b_s \}}.
$ Define the process 
\[
G_w:=\int_t^w \mathrm{e}^{-r(s-t)}\big( \rho_g \mathrm{d}D_s + k_s^{\star,g}(\Psi) B \mathrm{d}s \big) + \mathrm{e}^{-r(w-t)} \widehat{\mathcal U}^\star_{I-N_w}(U_w^{b}(\Psi)),\ w\in[t,\tau]. 
\] 
Observe we can rewrite the second term in the following form (with the convention $\tau^{N_t}=t$, $\tau^{N_w+1}=w$)
\begin{align*}
\mathrm{e}^{-r(w-t)}\widehat{\mathcal U}^\star_{I-N_w}(U_w^{b}(\Psi)) =& \sum_{i=N_t}^{N_w} \mathrm{e}^{-r(\tau^{i+1}-t)}\widehat{\mathcal U}^\star_{I-i}\Big(U_{(\tau^{i+1})^-}^{b}(\Psi)\Big) - \mathrm{e}^{-r(\tau^i-t)} \widehat{\mathcal U}^\star_{I-i}\Big(U_{\tau^i}^{b}(\Psi)\Big) \\
& + \sum_{i=N_t}^{N_w-1} e^{-r(\tau^{i+1}-t)}\Big(  \widehat{\mathcal U}^\star_{I-(i+1)}\Big(U_{\tau^{i+1}}^{b}(\Psi)\Big) - \widehat{\mathcal U}^\star_{I-i}\Big(U_{(\tau^{i+1})^-}^{b}(\Psi)\Big) \Big)  + \widehat{\mathcal U}^\star_{I-N_t}\Big(U_t^{b}(\Psi)\Big).
\end{align*}
Since the functions $\widehat{\mathcal U}^\star_j$ are $C^1$, we can apply It\^o's formula on the intervals $[\tau^i\wedge\tau,\tau^{i+1}\wedge\tau)$ with $i\in\{N_t,\dots,N_w\}$ to obtain an integral expression for the first sum. Regarding the second sum, observe that
\begin{align*}
&  \widehat{\mathcal U}^\star_{I-(i+1)}\Big(U_{\tau^{i+1}}^{b}(\Psi)\Big) - \widehat{\mathcal U}^\star_{I-i}\Big(U_{(\tau^{i+1})^-}^{b}(\Psi)\Big) \\
&=  \Big( \widehat{\mathcal U}^\star_{I-(i+1)}\Big(U_{(\tau^{i+1})^-}^{b}(\Psi)-h_{\tau^{i+1}}^{1,b}\Big) - \widehat{\mathcal U}^\star_{I-i}\Big(U_{(\tau^{i+1})^-}^{b}(\Psi)\Big)\Big) \Delta N_{\tau^{i+1}}  - \widehat{\mathcal U}^\star_{I-(i+1)}\Big(U_{(\tau^{i+1})^-}^{b}(\Psi)-h_{\tau^{i+1}}^{1,b}\Big)\Delta H_{\tau^{i+1}} \\
 &=  \int_{\tau^i}^{\tau^{i+1}} \Big( \widehat{\mathcal U}^\star_{I-(i+1)}\Big(U_{s^-}^{b}(\Psi)-h_s^{1,b}\Big) - \widehat{\mathcal U}^\star_{I-i}\Big(U_{s^-}^{b}(\Psi)\Big)\Big)  \mathrm{d}N_s  - \int_{\tau^i}^{\tau^{i+1}} \widehat{\mathcal U}^\star_{I-(i+1)}\Big(U_{s^-}^{b}(\Psi)-h_s^{1,b}\Big)\mathrm{d}H_s. 
\end{align*}
Hence
\begin{align*}
G_{\tau\wedge v}  =&\   \widehat{\mathcal U}^\star_{I-N_t}(u^{b}) + \ds\sum_{i=N_t}^{I-1}\int_{\tau^i\wedge v}^{\tau^{i+1}\wedge v} \mathrm{e}^{-r(s-t)}\Big(\rho_g-\rho_b \big(\widehat{\mathcal U}^\star_{I-i}\big)^\prime\big(U_s^{b}(\Psi)\big)\Big)\mathrm{d}D_s  +  \ds\sum_{i=N_t}^{I-1}\int_{\tau^i\wedge v}^{\tau^{i+1}\wedge v} \mathrm{e}^{-r(s-t)}\Big(k_s^{\star,g}(\Psi)B-r\widehat{\mathcal U}^\star_{I-i}\big(U_s^{b}(\Psi)\big) \Big)\mathrm{d}s \\
& +  \ds\sum_{i=N_t}^{I-1}\int_{\tau^i\wedge v}^{\tau^{i+1}\wedge v} \mathrm{e}^{-r(s-t)} \lambda_s^{k^{\star,g}(\Psi)}\Big(  \theta_s\widehat{\mathcal U}^\star_{I-i-1}\big(U_{s^-}^{b}(\Psi)-h_s^{1,b}\big)-\widehat{\mathcal U}^\star_{I-i}\big(U_s^{b}(\Psi)\big) \Big) \mathrm{d}s \\
& +  \sum_{i=N_t}^{I-1}\int_{\tau^i\wedge v}^{\tau^{i+1}\wedge v} \mathrm{e}^{-r(s-t)}\widehat{\mathcal U}^{\star\prime}_{I-i}\big(U_s^{b}(\Psi)\big)\Big( rU_s^{b}(\Psi)-Bk_s^{\star,b}(\Psi)+\lambda_{s}^{k^{\star,b}(\Psi)}(h_s^{1,b}+(1-\theta_s)h_s^{2,b}) \Big)\mathrm{d}s \\
& +  \sum_{i=N_t}^{I-1}\int_{\tau^i\wedge v}^{\tau^{i+1}\wedge v} \mathrm{e}^{-r(s-t)}\Big(  \widehat{\mathcal U}^\star_{I-i-1}\big(U_{s^-}^{b}(\Psi)-h_s^{1,b}\big)-\widehat{\mathcal U}^\star_{I-i}\big(U_{s^-}^{b}(\Psi)\big) \Big) \big(\mathrm{d}N_s-\lambda_s^{k^{\star,g}(\Psi)}\mathrm{d}s\big) \\
& -  \sum_{i=N_t}^{I-1}\int_{\tau^i\wedge v}^{\tau^{i+1}\wedge v} \mathrm{e}^{-r(s-t)}\widehat{\mathcal U}^\star_{I-i-1}\big(U_{s^-}^{b}(\Psi)-h_s^{1,b}\big)\big(\mathrm{d}H_s-\lambda_s^{k^{\star,g}(\Psi)}(1-\theta_s)\mathrm{d}s\big).
\end{align*}
We know that the derivative of every $\widehat{\mathcal U}^\star_j$ is greater than $\rho_g/\rho_b$ by definition, and since $D$ is non--decreasing, the first sum of integrals is non--positive. Also, the functions $\widehat{\mathcal U}^\star_j$ are solutions of the system of HJB equations, which implies that for any admissible contract the second and the third sum of integrals are also non--positive. We deduce 
\begin{align}\label{corcho17} 
\nonumber G_{\tau\wedge v}  \leq & ~ \widehat{\mathcal U}^\star_{I-N_t}(u^{b})  + \ds\sum_{i=N_t}^{I-1}\int_{\tau^i\wedge v}^{\tau^{i+1}\wedge v} \mathrm{e}^{r(t-s)}\Big( \widehat{\mathcal U}^\star_{I-i-1}\big(U_{s^-}^{b}(\Psi)-h_s^{1,b}\big)-\widehat{\mathcal U}^\star_{I-i}\big(U_{s^-}^{b}(\Psi)\big) \Big)\big(\mathrm{d}N_s-\lambda_s^{k^{\star,g}}\mathrm{d}s\big) \\ 
& -  \ds\sum_{i=N_t}^{I-1}\int_{\tau^i\wedge v}^{\tau^{i+1}\wedge v} \mathrm{e}^{-r(s-t)}\widehat{\mathcal U}^\star_{I-i-1}\big(U_{s^-}^{b}(\Psi)-h_s^{1,b}\big)\big(\mathrm{d}H_s-\lambda_s^{k^{\star,g}(\Psi)}(1-\theta_s)\mathrm{d}s\big).
\end{align}  
Define $\lambda:=\max_{1\leq j\leq I} \widehat\lambda_j^{\rm SH}$. For every $i$ we have that, recalling that the functions $\widehat{\mathcal U}^\star_j$ are non--decreasing and null at $0$
\begin{align*}
\E^{\P^{k^{\star,g}}}\bigg[ \int_t^\tau \mathrm{e}^{-r(s-t)}\Big|\widehat{\mathcal U}^\star_{I-i-1}\big(U_s^{b}(\Psi)-h_s^{1,b}\big)\Big|\mathrm{d}s \bigg| \mathcal G_t\bigg] 
\leq \E^{\P^{k^{\star,g}}}\bigg[ \int_t^\tau \mathrm{e}^{-r(s-t)}\frac{\rho_g}{\rho_b}u^{b} \mathrm{e}^{(r+\lambda)s}\mathrm{d}s \bigg| \mathcal G_t\bigg],
\end{align*} 
which is finite. Indeed, we have between two consecutive jump times of $N$
\begin{align*} 
\mathrm{d}U_s^{b}(\Psi) & =  \Big(rU_s^{b}(\Psi)-Bk_s^{\star,b}(\Psi)+(h_s^{1,b}+(U_s^{b}(\Psi)-h_s^{1,b})(1-\theta_s))\lambda_s^{k^{\star,b}(\Psi)}\Big)\mathrm{d}s-\rho_b\mathrm{d}D_s \\ 
& \leq  \Big(rU_s^{b}(\Psi)+h_s^{1,b}\lambda_s^{k^{\star,b}(\Psi)}+(U_s^{b}(\Psi)-h_s^{1,b})(1-\theta_s)\lambda_s^{k^{\star,b}(\Psi)}\Big)\mathrm{d}s \\ 
& =  U_s^{b}(\Psi)\left(r+(1-\theta_s)\lambda_s^{k^{\star,b}(\Psi)}\right)ds+h_s^{1,b}\theta_s\lambda_s^{k^{\star,b}(\Psi)}\mathrm{d}s  \leq  U_s^{b}(\Psi)\left(r+\lambda_s^{k^{\star,b}(\Psi)}\right)\mathrm{d}s,
\end{align*}
where we used the facts that $h_s^{1,b}\in[0,U_s^{b}(\Psi)]$, the functions $\widehat{\mathcal U}^\star_j$ are non--decreasing and $U_s^{b}(\Psi)$ is bounded from below and has positive jumps. Similarly 
\begin{align*}
& \E^{\P^{k^{\star,g}}}\bigg[ \ds\int_t^\tau \mathrm{e}^{-r(s-t)}\Big|\widehat{\mathcal U}^\star_{I-i-1}\big(U_{s^-}^{b}(\Psi)-h_s^{1,b}\big)-\widehat{\mathcal U}^\star_{I-i}\big(U_{s^-}^{b}(\Psi)\big)\Big|\mathrm{d}s \bigg| \mathcal G_t\bigg] \\ 
\leq &~ \E^{\P^{k^{\star,g}}}\bigg[ \ds\int_t^\tau \mathrm{e}^{-r(s-t)}\Big|\widehat{\mathcal U}^\star_{I-i-1}\big(U_{s^-}^{b}(\Psi)-h_s^{1,b}\big)\Big|\mathrm{d}s \bigg| \mathcal G_t \bigg]+\E^{\P^{k^{\star,g}}}\bigg[\ds\int_t^\tau \mathrm{e}^{-r(s-t)}\Big|\widehat{\mathcal U}^\star_{I-i}\big(U_{s^-}^{b}(\Psi)\big)\Big|\mathrm{d}s \bigg| \mathcal G_t \bigg] \\
\leq &~ \E^{\P^{k^{\star,g}}}\bigg[\ds\int_t^\tau \mathrm{e}^{-r(s-t)}\frac{\rho_g}{\rho_b}U_s^{b}(\Psi)\mathrm{d}s \bigg| \mathcal G_t\bigg]+\E^{\P^{k^{\star,g}}}\bigg[\ds\int_t^\tau \mathrm{e}^{-r(s-t)}\frac{\rho_g}{\rho_b}U_s^{b}(\Psi)\mathrm{d}s \bigg| \mathcal G_t\bigg]   \leq 2\E^{\P^{k^{\star,g}}}\bigg[\ds\int_t^\tau \mathrm{e}^{-r(s-t)}\frac{\rho_g}{\rho_b}u^b \mathrm{r}^{(r+\lambda)s}\mathrm{d}s\bigg| \mathcal G_t \bigg]<\infty. 
\end{align*}
The stochastic integrals appearing above are martingales, and taking conditional expectation in (\ref{corcho17}) we get $\E^{\P^{k^{\star,g}}}[ G_{\tau\wedge v}| \mathcal G_t] \leq \widehat{\mathcal U}^\star_{I-N_t}(u^{b})$ and from Fatou's Lemma we obtain 
\[
\widehat{\mathcal U}^\star_{I-N_t}(u^{b})\geq \underset{v\to\infty}{\underline{\rm{lim}}} \E^{\P^{k^{\star,g}}}\big[ G_{\tau\wedge v}\big| \mathcal G_t\big]\geq\E^{\P^{k^{\star,g}}}\bigg[\underset{v\to\infty}{\underline{\rm{lim}}}G_{\tau\wedge v}\bigg| \mathcal G_t\bigg] = U_t^g(\Psi),
\] 
where we used that, $\P^{k^{\star,g}}-$a.s.
\begin{align*} 
\underset{v\to\infty}{\underline{\lim}}G_{\tau\wedge v}  = & \lim_{v\to\infty} \int_t^{\tau\wedge v} \mathrm{e}^{-r(s-t)}\big(\rho_g\mathrm{d}D_s+k_s^{\star,g}(\Psi)B\mathrm{d}s\big)+ {\bf 1}_{\{v<\tau\}}\mathrm{e}^{-r(v-t)}\widehat{\mathcal U}^\star_{I-N_v}(U_v^{b}(\Psi))  =  \int_t^{\tau} \mathrm{e}^{-r(s-t)}\big(\rho_g\mathrm{d}D_s+k_s^{\star,g}(\Psi)B\mathrm{d}s\big).
\end{align*}
\end{proof}

\vspace{0.5em}
We end this section with the

\vspace{0.5em}
\begin{proof}[Proof of Proposition \ref{credibleset}]
Notice that the inclusion ${\mathcal C}_t \subseteq  \overline{\mathcal C}_t$ holds by definition and therefore we only need to prove the reverse inclusion. We will make use of contracts with lump--sum payments to prove that every point in $ \overline{\mathcal C}_t$ belongs to the credible set ${\mathcal C}_t$. We start by defining the line with slope $\rho_g/\rho_b$ which passes through the point $(u^{b},u^g)=\big(\frac{B(I-N_t)}{r+\lambda_t^{\rm SH}},\frac{B(I-N_t)}{r+\lambda_t^{\rm SH}}\big)$
\[
\overline {\mathfrak M}_t(u^{b}):=\frac{\rho_g}{\rho_b}u^{b}+\frac{Bj}{r+\lambda_t^{\rm SH}}\bigg(1-\frac{\rho_g}{\rho_b}\bigg),
\]
and the sets
\begin{align*}
 \overline{{\mathcal C}}_t^1 &:= \left\{ (u^{b},u^g)\in  {\mathcal{V}}_t \times {\mathcal{V}}_t: \overline {\mathfrak M}_t(u^{b})\leq u^g \leq {\mathfrak U}_{t}(u^{b}) \right\}, \  \overline{{\mathcal C}}_t^2 := \left\{ (u^{b},u^g)\in  {\mathcal{V}}_t \times {\mathcal{V}}_t:\ {\mathfrak L}_t(u^{b})\leq u^g \leq \overline {\mathfrak M}_t(u^{b}) \right\}.
\end{align*}
From Section \ref{credibleregiondelay} in the Appendix, we know that $ \overline{{\mathcal C}}_t^1\subseteq {\mathcal C}_j$. Indeed, every point from the upper boundary $ {\mathfrak U}_t$ belongs to the credible set, and if we perturb a contract $\Psi=(\theta,D)$ only by adding a lump--sum payment $\varepsilon$ at time $t$, that is $\mathrm{d}D_s^{\Psi'}={\bf 1}_{\left\{s=t\right\}}\varepsilon + \mathrm{d}D_s^{\Psi}$, then the values of the banks under $\Psi'$ are $U_t^g(\Psi')=u^g+\varepsilon\rho_g$ and $U_t^b(\Psi')=u^{b}+\varepsilon\rho_b$, so $(U_t^{b}(\Psi'),U_t^g(\Psi'))=(u^b,u^g)+\varepsilon(\rho_b,\rho_g)$. We use this idea to prove also that 
$ \overline{{\mathcal C}}_t^2\subseteq {\mathcal C}_j$. From Proposition \ref{lowerboundary}, we know that the graph of ${\mathfrak L}_t$ is contained in ${\mathcal C}_t$. Therefore any point of the following form belongs to ${\mathcal C}_t$
\begin{equation}\label{desplazamiento-lumpsum}
(\widehat u^{b},\widehat u^g) = (u^{b},u^g) + \ell (\rho_b,\rho_g),~\ell\geq 0,~ u^g = {\mathfrak L}_t(u^{b}).
\end{equation} 
By the geometry of the lower boundary ${\mathfrak L}_t$, the set of points of the form \eqref{desplazamiento-lumpsum} is exactly $ \overline{{\mathcal C}}_t^2$.
\end{proof}

\section{Principal's value function on the boundary of the credible set}\label{sec:E}

\subsection{The optimal full--monitoring contract in pure moral hazard}
\label{app:dylan-previous}

The full-monitoring problem studied in \cite{pages2014mathematical}, considers that the only acceptable behaviour for the bank, from the social point of view, is that she never shirks away from her monitoring responsibilities. In other words, only contracts with a recommendation of $k=0$ are allowed. In this section there is no adverse selection, so there is only one type of bank, and the main result stands for both $i=b,g$, a good bank or a bad bank. The value function of the investor in this sub-problem is given by
\[
V_{t}^{{\rm pm},0}(R_0):=\underset{(D^i,\theta^i)\in\mathcal A^{0,i}(t,R_0)}{\rm{ess\ sup}}\ \mathbb E^{\mathbb P^{0}}\left[\left.\int_{t\wedge\tau}^\tau(I-N_s)\mu \mathrm{d}s-\mathrm{d}D^i_s\right|\mathcal G_t\right],
\]
where the set of admissible contracts $\mathcal A^{0,i}(t,R_0)$ is defined for $R_0\geq b_t$, by
\begin{align*}
\mathcal A^{0,i}(t,R_0):=\left\{(\theta^i,D^i)\in\Theta\times\mathcal D,\text{ s.t. $(\theta^i,D^i)$ enforces $k=0$ and $U_{t}^i(\theta^i,D^i)\geq R_0$}\right\}.
\end{align*}
Define for $x>0$ the functions
\[\phi(x):=\left(\frac{1+x}{1+2x}\right)^{\frac1x-1},\ \psi(x):=\frac{\phi(x)-x}{(1-x)\phi(x)}.\]
Let us then define some family of concave functions, unique solutions to the following system of ODEs
\begin{equation}\label{eq:vj}
\begin{cases} 
 \displaystyle \left( ru+\widehat \lambda_j^0 \widehat b_j \right) (v_j^i)^{\prime }(u)+j\mu -\widehat \lambda
_j^0\left( v^i_j(u)-\frac{u-\widehat b_j}{\widehat b_{j-1}}v^i_{j-1}(\widehat b_{j-1})\right) =0,\
u\in \left( \widehat b_j,\widehat b_j+\widehat b_{j-1}\right],\\[0.5em]
\displaystyle \left( ru+\widehat \lambda_j^0 \widehat b_j\right) (v^i_j)^{\prime }(u)+j\mu -\widehat \lambda_j^0 \left( v_{j}^i(u)-v_{j-1}^i(u-\widehat b_{j})\right) =0,\ u\in \left(\widehat b_j+\widehat b_{j-1},\gamma^i_{j}\right], \\[0.5em]
\displaystyle \rho_i(v^i_j)^{\prime}(u)+1=0,\ u>\gamma^i_j,
\end{cases}
\end{equation}
with initial values $\gamma^i_1:=\widehat b_1$ and
\[
v^i_1(u):=\overline{v}^i_1-\frac{1}{\rho_i}(u-\widehat b_1), u\geq \widehat b_1,\ \overline{v}^i_1:=\frac{\mu}{\widehat \lambda^0_1}-\frac{\widehat b_1(r+\widehat \lambda^0_1)}{\rho_i\widehat \lambda^0_1},
\]
and where for $j\geq 2$, $\gamma^i_j$ is defined recursively by 
\begin{equation}\label{eq:gammaj}
r/\widehat \lambda^0_j-1\in\partial v^i_{j-1}(\gamma^i_j-\widehat b_j),
\end{equation}
where $\partial v^i_{j-1}$ is the super--differential of the concave function $v^i_{j-1}$. The main result of \cite{pages2014mathematical} is 
\begin{Theorem}\label{thr:dylan-previous}
Assume that the $\big(\widehat \lambda^0_j\big)_{1\leq j\leq I}$ satisfy the following recursive conditions for $j\geq 2$
\[
\frac{r}{\widehat \lambda^0_j}-1\leq \frac{v^i_{j-1}\big(\widehat b_{j-1}\big)}{\widehat b_{j-1}}\text{ and }\left((v^i_{j-1})^{\prime}\Big(\widehat b_{j-1}\Big)\right)^+\frac{\widehat b_{j-1}}{v^i_{j-1}\Big(\widehat b_{j-1}\Big)}\leq \psi\bigg(\frac{r}{\widehat \lambda_j^0}\bigg).\]

Then, under Assumption \ref{assump}, the system \reff{eq:vj} is well--posed and we have
\[
V_t^{{\rm pm},0}(R_0)=\underset{u_{t}\geq R_0}{\sup}v^i_{I-N_t}\left(u_{t}\right),
\]
where $(u_s)_{s\geq t}$ is defined as the unique solution to the SDE on $[t,\tau)$
\begin{align*}
\mathrm{d}u_s=&\ \left(ru_s+\lambda^0_{I-N_s}\widehat b_{I-N_s}\right)\mathrm{d}s -\rho_i \mathrm{d}D^{\star,i}_s  \\
&-\left({\bf 1}_{\{u_s\in[\widehat b_{I-N_s},\widehat b_{I-N_s-1}+\widehat b_{I-N_s})\}}(u_s-\widehat b_{I-N_s-1})+\widehat b_{I-N_s}{\bf 1}_{\{u_s\in[\widehat b_{I-N_s}+\widehat b_{I-N_s-1},\gamma_{I-N_s}^i)\}}\right)\mathrm{d}N_s\\
&-\left({\bf 1}_{\{u_s\in[\widehat b_{I-N_s},\widehat b_{I-N_s-1}+\widehat b_{I-N_s})\}}\widehat b_{I-N_s-1}+(u_s-\widehat b_{I-N_s}){\bf 1}_{\{u_s\in[\widehat b_{I-N_s}+\widehat b_{I-N_s-1},\gamma_{I-N_s}^i)\}}\right)\mathrm{d}H_s,
\end{align*}
with initial value $u_{t}$ at $t$, and where we defined for $s\in [t,\tau)$ and $j=1,\dots,I$
\begin{align*}
&D^{\star,i}_s:={\bf 1}_{\{s=t\}}\frac{(u_t-\gamma_{I-N_t}^i)^+}{\rho_i}+\int_t^s\delta^{I-N_r}_i(u_r)\mathrm{d}r,\ \theta^\star_s:=\theta^{I-N_s}_i(u_s),\\
&\delta^j_i(u):={\bf 1}_{\{u=\gamma_j^i\}}\frac{\widehat \lambda^0_j\widehat b_j+r\gamma^i_j}{\rho_i},\ \theta^j_i(u):={\bf 1}_{\{u\in[\widehat b_j,\widehat b_{j-1}+\widehat b_j)\}}\frac{u-\widehat b_j}{\widehat b_{j-1}}+{\bf 1}_{\{u\in[\widehat b_j+\widehat b_{j-1},\gamma_j^i)\}}.
\end{align*}
\end{Theorem}

\subsection{Proofs of the main results}

We start this section with the

\vspace{0.5em}
\begin{proof}[Proof of Proposition \ref{prop:value function lb1}]
Consider any time $t\geq 0$ and take any $u^{b,c}\geq C(I-N_t)$, as well as
some $\Psi_g\in\widehat{\mathcal A}^{g}(t,\widehat{\mathfrak L}_{I-N_t}(u^{b,c}),u^{b,c})$. From Lemma \ref{lemma:contracts lb2}, we know that the components of $\Psi_g$ must satisfy $\theta^g\equiv1$ and that both banks shirk under $\Psi_g$. The payments determine the utility of the banks and the following holds by definition
\begin{align*}
 \E^{\P^{k^{\rm SH}}}\bigg[\int_t^{\tau^I} \mathrm{e}^{-r(s-t)} \mathrm{d}D_s^g \bigg|\mathcal G_t\bigg]=\frac{ u^{b,c} - C(I-N_t) }{\rho_b} . 
\end{align*}
Besides, the utility of the investor under the contract $\Psi_g$ is
\[
\E^{\P^{k^{\rm SH}}}\bigg[\int_t^{\tau^I} (\mu (I-N_s)\mathrm{d}s - \mathrm{d}D_s^g) \bigg| \mathcal G_t\bigg]= \ds\sum_{i=N_t}^{I-1} \frac{\mu(I-i)}{\widehat\lambda_{I-i}^{\rm SH}} - \E^{\P^{k^{\rm SH}}} \bigg[\int_t^{\tau^I} \mathrm{d}D_s^g \bigg| \mathcal G_t\bigg].
\]
Now, observe that
\[
\E^{\P^{k^{\rm SH}}} \bigg[\int_t^{\tau^I} \mathrm{d}D_s^g\bigg| \mathcal G_t \bigg] \geq \E^{\P^{k^{\rm SH}}} \bigg[\int_t^{\tau^I} \mathrm{e}^{-r(s-t)} \mathrm{d}D_s^g \bigg| \mathcal G_t\bigg]= \frac{ u^{b,c} - C(I-N_t) }{\rho_b},
\]
and the equality holds if and only if $D^g$ has a jump at time $t$ of size $\frac{ u^{b,c} - C(I-N_t) }{\rho_b}$ and $dD_s^g=0$ for every $s>t$. That means that it is optimal for the investor to use a contract with an initial lump--sum payment and to pay nothing afterwards. Consequently, the value function of the investor on the lower boundary is given by
\[
V^{\mathfrak L,g}_t(u^{b,c}) = \ds\sum_{i=N_t}^{I-1} \frac{\mu(I-i)}{\widehat\lambda_{I-i}^{\rm SH}} - \bigg( \frac{ u^{b,c} - C(I-N_t) }{\rho_b} \bigg).
\]
\end{proof}

\vspace{0.5em}
We continue this section with the

\vspace{0.5em}
\begin{proof}[Proof of Proposition \ref{prop:value function lb2}] 
Consider any time $t\geq 0$. Take any $u^{b,c}\in [c(I-N_t,1),C(I-N_t))$, and $\Psi_g\in\widehat{\mathcal A}^{g}(t,u^{b,c},u^{b,c})$. From Lemma \ref{lemma:contracts lb1}, we know that the components of $\Psi_g$ must satisfy $\mathrm{d}D_s^g=0$ for all $s\geq t$ and that both banks will shirk under this contract. Then, $\theta^g$ determines the continuation utilities of the banks in the following way
\[
u^{b,c} = \E^{\P^{k^{\rm SH}}}\bigg[ \int_t^\tau \mathrm{e}^{-r(s-t)} B(I-N_s)\mathrm{d}s \bigg|\mathcal G_t\bigg],
\]
so in this case, the problem (\ref{eq:lb investor problem}) reduces to
\[
(P)\ \ \sup_{\theta\in\Theta} \E^{\P^{k^{\rm SH}}}\bigg[\int_t^\tau \mu(I-N_s)\mathrm{d}s  \bigg|\mathcal G_t\bigg],\  
\textrm{s.t} \; \E^{\P^{k^{\rm SH}}}\bigg[\int_t^\tau \mathrm{e}^{-r(s-t)}B(I-N_s)\mathrm{d}s \bigg|\mathcal G_t\bigg]= u^{b,c}.
\]
Next, we rewrite the objective function in a more convenient way
\begin{align*}
&  \E^{\P^{k^{\rm SH}}}\bigg[ \int_t^\tau \mu(I-N_s)\mathrm{d}s \bigg|\mathcal G_t\bigg] \\
& =  \mu(I-N_t) \E^{\P^{k^{\rm SH}}}\big[ \tau^{N_t+1} - t \big|\mathcal G_t\big] + \ds\sum_{i=N_t+1}^{I-1} \mu(I-i) \E^{\P^{k^{\rm SH}}}\big[ {\bf1}_{\{ \tau > \tau^i \}}(\tau^{i+1}-\tau^i) \big|\mathcal G_t\big]  \\
& =  \frac{\mu(I-N_t)}{\widehat\lambda_{I-N_t}^{\rm SH}} + \ds\sum_{i=N_t+1}^{I-1} \mu(I-i) \E^{\P^{k^{\rm SH}}}\bigg[\E^{\P^{k^{\rm SH}}}\big[{\bf1}_{\{ \tau > \tau^i \}} \big| \Gc_{\tau^i}\big]     \E^{\P^{k^{\rm SH}}}\big [ \tau^{i+1}-\tau^i \big| \mathcal G_{\tau^i}\big]\bigg| \mathcal G_t\bigg]   =  \frac{\mu(I-N_t)}{\widehat\lambda_{I-N_t}^{\rm SH}} + \ds\sum_{i=N_t+1}^{I-1} \frac{\mu(I-i)}{\widehat\lambda_{I-i}^{\rm SH}} \E^{\P^{k^{\rm SH}}}[\theta_{\tau^i}| \mathcal G_t].
\end{align*}
We do the same with the constraint
\begin{align*}
  \E^{\P^{k^{\rm SH}}}\bigg[\int_t^\tau \mathrm{e}^{-r(s-t)}B(I-N_s)\mathrm{d}s\bigg| \mathcal G_t\bigg]
 &=  \E^{\P^{k^{\rm SH}}}\bigg[ \int_{t}^{\tau^{N_t+1}} B(I-N_t)\mathrm{e}^{-r(s-t)}\mathrm{d}s + \ds\sum_{i=N_t+1}^{I-1} {\bf 1}_{\{ \tau > \tau^i \}} \int_{\tau^i}^{\tau^{i+1}} \mathrm{e}^{-r(s-t)}B(I-i)\mathrm{d}s \bigg| \mathcal G_t\bigg] \\
 &=  \frac{B(I-N_t)}{r+\widehat\lambda_{I-N_t}^{\rm SH}} + \ds\sum_{i=N_t+1}^{I-1} \frac{B(I-i)}{r} \E^{\P^{k^{\rm SH}}}\bigg[\E^{\P^{k^{\rm SH}}}\Big[{\bf 1}_{\{ \tau > \tau^i \}}\big(\mathrm{e}^{-r(\tau^i-t)}-\mathrm{e}^{-r(\tau^{i+1}-t)}\big) \Big| \Gc_{\tau^i}\Big]\bigg| \mathcal G_t\bigg]  \\
 &=  \frac{B(I-N_t)}{r+\widehat\lambda_{I-N_t}^{\rm SH}} + \ds\sum_{i=N_t+1}^{I-1} \frac{B(I-i)}{r+\widehat\lambda_{I-i}^{\rm SH}} \E^{\P^{k^{\rm SH}}}\big[\theta_{\tau^i} \mathrm{e}^{-r(\tau^i-t)} \big| \mathcal G_t\big].
\end{align*}
So we obtain the following expression for our problem
\[
(P)\begin{cases}
\ds\sup_{\theta\in\Theta} \ds\frac{\mu(I-N_t)}{\widehat\lambda_{I-N_t}^{\rm SH}} + \ds\sum_{i=N_t+1}^{I-1} \frac{\mu(I-i)}{\widehat\lambda_{I-i}^{\rm SH}} \E^{\P^{k^{\rm SH}}}[\theta_{\tau^i}| \mathcal G_t],\\[0.5em]
\textrm{s.t}  \ds\frac{B(I-N_t)}{r+\widehat\lambda_{I-N_t}^{\rm SH}} + \ds\sum_{i=N_t+1}^{I-1} \frac{B(I-i)}{r+\widehat\lambda_{I-i}^{\rm SH}} \E^{\P^{k^{\rm SH}}}\big[\theta_{\tau^i} \mathrm{e}^{-r(\tau^i-t)} \big| \mathcal G_t\big] = u^{b,c}.
\end{cases}
\]
We do not know how to solve $(P)$ directly, so we will define its dual problem, characterise its solution and show that the duality gap is zero. In order to do that, we define the Lagrangian function $L:\Theta\times\R\times\Omega\longrightarrow\R$ as follows
\begin{align*}
L(\theta,\nu,\omega) := & -\frac{\mu(I-N_t(\omega))}{\widehat\lambda_{I-N_t(\omega)}^{\rm SH}}  - \ds\sum_{i=N_t(\omega)+1}^{I-1} \frac{\mu(I-i)}{\widehat\lambda_{I-i}^{\rm SH}}\E^{\P^{k^{\rm SH}}}[\theta_{\tau^i}| \mathcal G_t](\omega) \\
& + \nu \bigg( \frac{B(I-N_t(\omega))}{r+\widehat\lambda_{I-N_t(\omega)}^{\rm SH}} + \ds\sum_{i=N_t(\omega)+1}^{I-1} \frac{B(I-i)}{r+\widehat\lambda_{I-i}^{\rm SH}} \E^{\P^{k^{S\rm H}}}\big[\theta_{\tau^i} \mathrm{e}^{-r(\tau^i-t)} \big| \mathcal G_t\big](\omega) - u^{b,c}\bigg),
\end{align*}
and also define the dual function and the dual problem respectively as
\[
g(\nu,\omega)  :=  \inf_{\theta\in\Theta} L(\theta,\nu,\omega),~ (D)\; 
\ds\sup_{\nu\in\R}  g(\nu,\omega)
\]
Then, we have the weak duality inequality (where val denotes the value of the optimisation problem)
\[
- \textrm{val}(P) = \inf_{\theta\in\Theta} \sup_{\nu\in\R} L(\theta,\nu,\omega) \geq \sup_{\nu\in\R} \inf_{\theta\in\Theta} L(\theta,\nu,\omega) = \textrm{val}(D).
\]
We rewrite the dual function as follows
\begin{align*}
g(\nu,\omega)  =  -\frac{\mu(I-N_t(\omega))}{\widehat\lambda_{I-N_t(\omega)}^{\rm SH}} + \nu \bigg(\frac{B(I-N_t(\omega))}{r+\widehat\lambda_{I-N_t(\omega)}^{\rm SH}} - u^{b,c}\bigg) + \inf_{\theta\in\Theta} \ds\sum_{i=N_t(\omega)+1}^{I-1} \int_\Omega \theta_{\tau^i}(\widetilde\omega) \bigg(\nu \frac{B(I-i)}{r+\widehat\lambda_{I-i}^{\rm SH}}\mathrm{e}^{-r(\tau^i(\widetilde\omega)-t)} - \frac{\mu(I-i)}{\widehat\lambda_{I-i}^{\rm SH}} \bigg) \mathrm{d}\P_{t,\omega}^{\rm SH}(\widetilde\omega),
\end{align*}
where $\P_{t,\omega}^{SH}$ is a regular conditional probability distribution for the conditional expectation with respect to the raw (that is to say not completed) version of $\Gc_t$. We have easily that it is optimal to set the optimal control $\theta^\nu$ to be $
\theta_{\tau^i}^\nu(\widetilde\omega) := \mathbf{1}_ {\widetilde\omega \in A_\nu^i}(\widetilde\omega)$, where the set $A_\nu^i$ is defined by
\[
A_\nu^i :=\begin{cases}
\displaystyle \Omega,\ \text{if }\nu < \frac{\mu}{B}\frac{r+\widehat\lambda_{I-i}^{\rm SH}}{\widehat\lambda_{I-i}^{\rm SH}},\\
\displaystyle\bigg\{ \widetilde \omega:   \tau^i(\widetilde\omega) - t > \frac{1}{r}\ln\bigg( \frac{\nu B \widehat\lambda_{I-i}^{\rm SH}}{ \mu(r+\widehat\lambda_{I-i}^{\rm SH}) } \bigg)    \bigg\},\ \text{if }\nu \geq \frac{\mu}{B}\frac{r+\widehat\lambda_{I-i}^{\rm SH}}{\widehat\lambda_{I-i}^{\rm SH}}.
\end{cases}
\]
Therefore, for any $\nu\in\R$ the dual function has the following form,
 using that the conditional law of $\tau^i-t$ given $\Gc_t$ is the same as the law of $\tau^i$
\begin{align}
  g(\nu,\omega) =  -\frac{\mu(I-N_t(\omega))}{\widehat\lambda_{I-N_t(\omega)}^{\rm SH}} + \nu \bigg(\frac{B(I-N_t(\omega))}{r+\widehat\lambda_{I-N_t(\omega)}^{\rm SH}} - u^{b,c}\bigg) + \ds\sum_{i=N_t(\omega)+1}^{I-1}  \int_{s_i(\nu)}^{\infty} \bigg(  \frac{\nu B(I-i)\mathrm{e}^{-rx}}{r+\widehat\lambda_{I-i}^{\rm SH}} - \frac{\mu(I-i)}{\widehat\lambda_{I-i}^{\rm SH}} \bigg) f_{\tau^i}(x)\mathrm{d}x.\label{g}
\end{align}
It is not difficult to see that $g$ is a continuous and differentiable function. As we want to maximise $g$ in the dual problem, we compute its derivative with respect to $\nu$ and we get
\begin{align*}
g^\prime(\nu,\omega)  
 = & \frac{B(I-N_t(\omega))}{r+\widehat\lambda_{I-N_t}^{\rm SH}} - u^{b,c} + \ds\sum_{i=N_t+1}^{I-1} \int_{s_i(\nu)}^{\infty}  \frac{B(I-i)}{r+\widehat\lambda_{I-i}^{\rm SH}}\mathrm{e}^{-rx}   f_{\tau^i}(x)\mathrm{d}x.
\end{align*}
Since $\nu\longmapsto s_i(\nu)$ is non--decreasing for any $i=1,\dots,I$, $g^\prime$ is non--increasing in $\nu$. Furthermore, since $u^{b,c}\geq c(I-N_t,1)$, we have the limit at $+\infty$ of $g^\prime$ is non--positive, and that its value for small $\nu$ is positive because $u^{b,c}< C(I-N_t)$ and 
\[
\frac{B(I-N_t(\omega))}{r+\widehat\lambda_{I-N_t}^{\rm SH}} + \ds\sum_{i=N_t+1}^{I-1} \int_{0}^{\infty}  \frac{B(I-i)}{r+\widehat\lambda_{I-i}^{\rm SH}}\mathrm{e}^{-rx}   f_{\tau^i}(x)\mathrm{d}x = C(I-N_t).
\]
Therefore, there is a unique value of $\nu$ that makes $g^\prime$ equal to $0$. Now, we compute for any $\nu$ the value of the constraint from the primal problem for the control $\theta^\nu$
\begin{align*}
\ds\sum_{i=N_t+1}^{I-1} \frac{B(I-i)}{r+\widehat\lambda_{I-i}^{\rm SH}} \E^{\P^{k^{\rm SH}}}\big[\theta_{\tau^i}^\nu \mathrm{e}^{-r(\tau^i-t)} \big|\mathcal G_t\big] =  \ds\sum_{i=N_t+1}^{I-1} \int_{s_i(\nu)}^{\infty} \frac{B(I-i)}{r+\widehat\lambda_{I-i}^{\rm SH}}\mathrm{e}^{-rx}   f_{\tau^i}(x)\mathrm{d}x,   
\end{align*}
so $\theta^\nu$ is feasible in problem $(P)$ if and only if $g^\prime(\nu,\omega)=0$. Next, we compute for $\theta^\nu$ the value of the objective function in the primal (minimisation) problem
\[
-\frac{\mu(I-N_t)}{\widehat\lambda_{I-N_t}^{\rm SH}} - \ds\sum_{i=N_t+1}^{I-1} \frac{\mu(I-i)}{\widehat\lambda_{I-i}^{\rm SH}} \E_t^{\P^{k^{\rm SH}}}\big[\theta_{\tau^i}^\nu\big] = -\frac{\mu(I-N_t)}{\widehat\lambda_{I-N_t}^{\rm SH}} - \ds\sum_{i=N_t+1}^{I-1}  \int_{s_i(\nu)}^\infty \frac{\mu(I-i)}{\widehat\lambda_{I-i}^{\rm SH}}  f_{\tau^i}(x)\mathrm{d}x.
\]
If this quantity is equal to $g(\nu,\cdot)$, the duality gap is zero. From (\ref{g}) we see that this happens if and only if 
\[
 \nu \bigg(\frac{B(I-N_t)}{r+\widehat\lambda_{I-N_t}^{\rm SH}} - u^{b,c} + \ds\sum_{i=N_t+1}^{I-1}  \int_{s_i(\nu)}^{\infty}  \frac{B(I-i)}{r+\widehat\lambda_{I-i}^{\rm SH}}\mathrm{e}^{-rx}  f_{\tau^i}(x)\mathrm{d}x \bigg) = 0 
 \iff \nu g^\prime(\nu,\cdot) = 0.
\]

We conclude that if $\nu\in\R$ is such that $g^\prime(\nu)=0$ then the control $\theta^\nu$ is optimal in the primal problem. 
\end{proof}

\vspace{0.5em}
We continue with the 

\vspace{0.5em}

\begin{proof}[Proof of Proposition \ref{ub absorbing}]
Define the process $\ell_s=\widehat{\mathfrak U}_{I-N_s}(U_s^{b,c}(\Psi_g))-U_s^g(\Psi_g)$ and note that $\ell_s\geq 0$ for every $s\geq 0$. We will prove that $\ell_t=0$ implies $\ell_v=0$ for every $v\geq t$. Assume thus that $\ell_t=0$. Following the same idea as in the proof of Theorem \ref{verification theorem}, we have for $v\geq t$ 
\begin{align*}
\ell_v  = & \  \ds\sum_{i=N_t}^{I-1}\int_{\tau^i\wedge v}^{\tau^{i+1}\wedge v} - \Big( rU_s^g(\Psi_g) - Bk_s^{\star,g}(\Psi_g) + [h_s^{1,g} + (1-\theta_s^g)h_s^{2,g}]\lambda_s^{k^{\star,g}(\Psi_g)} \Big)\mathrm{d}s \\
& +  \sum_{i=N_t}^{I-1}\int_{\tau^i\wedge v}^{\tau^{i+1}\wedge v} \widehat{\mathfrak U}_{I-i}^\prime(U_s^{b,c}(\Psi_g))\Big(rU_s^{b,c}(\Psi_g)-Bk_s^{\star,b,c}(\Psi_g)+\lambda_{I-i}^{k^{\star,b,c}(\Psi_g)}(h_s^{1,b,c}+(1-\theta^g_s)h_s^{2,b,c}) \Big)\mathrm{d}s \\
& +  \sum_{i=N_t}^{I-1}\int_{\tau^i\wedge v}^{\tau^{i+1}\wedge v} \Big(h_s^{1,g} + \widehat{\mathfrak U}_{I-i-1}(U_{s^-}^{b,c}(\Psi_g)-h_s^{1,b,c})-\widehat{\mathfrak U}_{I-i}(U_{s^-}^{b,c}(\Psi_g)) \Big) \mathrm{d}N_s \\
& +  \sum_{i=N_t}^{I-1}\int_{\tau^i\wedge v}^{\tau^{i+1}\wedge v} \Big( h_s^{2,g} - \widehat{\mathfrak U}_{I-i-1}(U_{s^-}^{b,c}(\Psi_g)-h_s^{1,b,c}) \Big)  \mathrm{d}H_s + \Big(\rho_g-\rho_b \widehat{\mathfrak U}_{I-i}^\prime(U_s^{b,c}(\Psi_g))\Big)\mathrm{d}D^g_s.
\end{align*}
Since the functions $\widehat{\mathfrak U}_i$ solve the system of HJB equations \eqref{DPequation}, and $\big(\rho_g-\rho_b\widehat{\mathfrak U}_i^\prime(U_s^{b,c}(\Psi_g))\big) \mathrm{d}D^g_s\leq0$ for every $s$, we have
\begin{align*}
\ell_v \leq & \ds\sum_{i=N_t}^{I-1}\int_{\tau^i\wedge v}^{\tau^{i+1}\wedge v} \Big( r\widehat{\mathfrak U}_{I-i}(U_s^{b,c}(\Psi_g)) -rU_s^g(\Psi_g) - [h_s^{1,g} + (1-\theta^g_s)h_s^{2,g}]\lambda_s^{k^{\star,g}(\Psi_g)} \Big)\mathrm{d}s \\
& -  \ds\sum_{i=N_t}^{I-1}\int_{\tau^i\wedge v}^{\tau^{i+1}\wedge v} \lambda_s^{k^{\star,g}(\Psi_g)}\Big(  \theta_s\widehat{\mathfrak U}_{I-i-1}(U_{s^-}^{b,c}(\Psi_g)-h_s^{1,b,c})-\widehat{\mathfrak U}_{I-i}(U_s^{b,c}(\Psi_g)) \Big) \mathrm{d}s \\
& +  \sum_{i=N_t}^{I-1}\int_{\tau^i\wedge v}^{\tau^{i+1}\wedge v} \Big( h_s^{1,g} + \widehat{\mathfrak U}_{I-i-1}(U_{s^-}^{b,c}(\Psi_g)-h_s^{1,b,c})-\widehat{\mathfrak U}_{I-i}(U_{s^-}^{b,c}(\Psi_g)) \Big) \mathrm{d}N_s  +   \Big( h_s^{2,g} - \widehat{\mathfrak U}_{I-i-1}(U_{s^-}^{b,c}(\Psi_g)-h_s^{1,b,c}) \Big)  \mathrm{d}H_s \\
= &    \ds\sum_{i=N_t}^{I-1}\int_{\tau^i\wedge v}^{\tau^{i+1}\wedge v}  \big(r+\lambda_s^{k^{\star,g}}\big)\Big( \widehat{\mathfrak U}_{I-i}(U_s^{b,c}(\Psi_g)) - U_s^g(\Psi_g) \Big) + \Big(h_s^{2,g} - \widehat{\mathfrak U}_{I-i-1}(U_{s}^{b,c}(\Psi_g)-h_s^{1,b,c}) \Big)\theta^g_s\lambda_s^{k^{\star,g}} \mathrm{d}s \\
& +  \sum_{i=N_t}^{I-1}\int_{\tau^i\wedge v}^{\tau^{i+1}\wedge v} \Big(h_s^{1,g} + \widehat{\mathfrak U}_{I-i-1}(U_{s^-}^{b,c}(\Psi_g)-h_s^{1,b,c})-\widehat{\mathfrak U}_{I-i}(U_{s^-}^{b,c}(\Psi_g)) \Big) \mathrm{d}N_s + \Big( h_s^{2,g} - \widehat{\mathfrak U}_{I-i-1}(U_{s^-}^{b,c}(\Psi_g)-h_s^{1,b,c}) \Big) \mathrm{d}H_s.
\end{align*}
Recall from Remark \ref{remark incentive compatibility} that on the upper boundary, we have 
\[
h_s^{1,g}= \widehat{\mathcal U}_{I-N_{s^-}}(U_{s^-}^{b,c}(\Psi_g))- \widehat{\mathcal U}_{I-N_{s^-}-1}(U_{s^-}^{b,c}(\Psi_g)-h_s^{1,b,c}(\Psi_g)),~ h_s^{2,g}= \widehat{\mathcal U}_{I-N_{s^-}-1}(U_{s^-}^{b,c}(\Psi_g)-h_s^{1,b,c}(\Psi_g)),
\]
so that for $i=N_t$ the drift of the right--hand side is $0$ in $[\tau^i,\tau^{i+1})$ and the jump at time $\tau^{i+1}$ is also $0$. It is easy to see that the same happens for every $i\in\{N_t,\dots,I\}$ and therefore $\ell_v\leq 0$ for every $v\geq 0$ which means $\ell_v=0$ for every $v\geq t$. 
\end{proof}

\vspace{0.5em}
We go on with the

\begin{proof}[Proof of Proposition \ref{ub contract}]

\vspace{0.5em}
$(i)$ We have from the proof of Proposition \ref{ub absorbing} that the processes $(\theta^g,h^{1,b,c},h^{2,b,c})$ are necessarily maximisers of the system of HJB equations \eqref{DPequation}. We can go back to the proof of Proposition \ref{prop:uj functions}, which is based on Corollary \ref{corollary:theta}, to observe that for $u^{b,c}<\widehat b_j$ the optimal $\theta\in C^j$ is uniquely given by $\theta=0$.

\vspace{0.5em}
$(ii)$ Observe that for every $(t,u^{b,c},u^g)\in[0,\tau]\times\widehat{\mathcal V}_{I-N_t}\times\widehat{\mathcal V}_{I-N_t}$ and $\Psi_g\in\widehat{\mathcal{A}}^g(t,u^g,u^{b,c})$ we have
\begin{align*}
U_t^{b,c}(\Psi_g) &\geq  \E^{\P^{k^{\star,g}(\Psi_g)}}\bigg[\int_t^\tau \mathrm{e}^{-r(s-t)}(\rho_b \mathrm{d}D_s^g+Bk_s^{\star,g}(\Psi_g)\mathrm{d}s)\bigg| \mathcal{G}_t\bigg] \\
&=  \frac{\rho_b}{\rho_g} U_t^g(\Psi_g) + \E^{\P^{k^{\star,g}(\Psi_g)}}\bigg[\int_t^\tau \mathrm{e}^{-r(s-t)}Bk_s^{\star,g}(\Psi_g)\mathrm{d}s\bigg| \mathcal{G}_t\bigg]\bigg( 1 - \frac{\rho_b}{\rho_g} \bigg) \geq  \frac{\rho_b}{\rho_g} U_t^g(\Psi_g).
\end{align*}
Then $U_{s_0}^{b,c}(\Psi_g)=\frac{\rho_g}{\rho_b} U_{s_0}^g(\Psi_g)$ implies that $
k_s^{\star,g}(\Psi_g)=k_s^{\star,b,c}(\Psi_g)=0, \text{ for every } s\in[s_0,\tau),
$
and in consequently 
\[
U_{s}^{b,c}(\Psi_g)=\frac{\rho_g}{\rho_b} U_{s}^g(\Psi_g)\geq b_s, \text{ for every } s\in[s_0,\tau).
\]
\end{proof}

\vspace{0.5em}
We continue with the

\begin{proof}[Proof of Proposition \ref{prop:value function ub}] We divide the proof in $2$ steps.

\vspace{0.5em}
\hspace{3em}{$\bullet$ \bf Step $1$:} We start with the region $u^{b,c}>\widehat b_{I-N_t}$. Let $\Psi_g=(D^g,\theta^g,h^{1,b,c},h^{2,b,c})\in\overline{\mathcal A}^{g}(t,u^{b,c})$ be such that $U_t^{b,c}(\Psi_g)=u^{b,c}\geq \widehat b_{I-N_t}$, $U_t^g(\Psi_g)= \widehat{\mathfrak U}_{I-N_t}(u^{b,c})$. From Proposition \ref{ub contract} we know that
\[
U_s^{b,c}(\Psi_g) \geq \widehat b_{I-N_s},~ k^{\star,b,c}(\Psi_g)=0,~ s\in[t,\tau).
\]
Therefore, Problem \eqref{eq:ub investor problem} is equivalent to
\[V_t^{\mathfrak{U},g}(u^{b,c}) = 
\underset{\Psi_g\in\overline{\mathcal A}^{g}(t,u^{b,c})}{\rm sup}\displaystyle\E^{\P}\bigg[ \int_t^\tau \mu(I-N_s)\mathrm{d}s - \int_t^\tau \mathrm{d}D_s^g \bigg],\; \textrm{s.t}\; \begin{cases} \displaystyle U_s^{b,c}(\Psi_g) \geq \widehat b_{I-N_s}, s\in[t,\tau), \\[0.5em]
 \displaystyle\E^{\P}\bigg[ \int_t^\tau \mathrm{e}^{-r(s-t)}\mathrm{d}D_s^g \bigg] = \frac{u^{b,c}}{\rho_b}.
\end{cases}
\]
This is exactly the problem considered in Pag\`es and Possama\"{i} \cite{pages2014mathematical} (see Theorem 3.15). We conclude that $ V_t^{\mathfrak{U},g}(u^{b,c}) =  v^b_{I-N_t}(u^{b,c}).$

\vspace{0.5em}
\hspace{3em}{$\bullet$ \bf Step $2$:} For the rest of the upper boundary, observe that the system of HJB equations associated to \eqref{eq:ub investor problem} is is given by $\widehat{\mathcal V}_0\equiv0$, and for any $1\leq j\leq I$
\begin{equation} \label{DPequation-valuefunction-upperboundary}
\min \Bigg\{ 
- \sup_{(\theta,h^1,h^2)\in C^{\mathfrak{U},j}} 
\left\{ 
\begin{array}{c} 
\widehat{\mathcal V}_j^\prime(u^{b,c}) \big( ru^{b,c} - Bk^{b,c} + (h^1 + (1-\theta)h^2) \widehat\lambda^{k^{b,c}}_j \big) \\[0.4em] 
+ \mu j + \widehat\lambda^{k^g}_j\theta \widehat{\mathcal V}_{j-1}(u^{b,c} - h^1) -\widehat\lambda^{k^g}_j \widehat{\mathcal V}_j(u^{b,c}) 
\end{array} 
\right\},  \
 \widehat{\mathcal V}_j^\prime(u^{b,c})+\frac{1}{\rho_b} 
\Bigg\}
=0,
\end{equation}
for every $u^{b,c}\geq \frac{Bj}{r+\widehat\lambda_j^{\rm SH}}$, with the boundary condition $\widehat{\mathcal V}_j(Bj/(r+\widehat\lambda_j^{\rm SH}))= \mu j/\widehat\lambda_j^{\rm SH},$
and where $$k^{b,c}:=j{\bf 1}_{\{h^1+(1-\theta)h^2<\widehat b_j\}},~ k^g:=j{\bf 1}_{\{\widehat{\mathcal U}_j^\star(u^{b,c})-\theta \widehat{\mathcal U}_{j-1}^\star(u^{b,c}-h^1)<\widehat b_j\}},$$ 
and the set of constraints $C^{\mathfrak{U},j}$ determined by Proposition \ref{ub contract} is defined by
\[
C^{\mathfrak{U},j}:=\bigg\{(\theta,h^1,h^2)\in[0,1]\times\R_+^2: h^1+h^2=u^{b,c},h^2\geq\frac{B(j-1)}{r+\widehat\lambda_{j-1}^{\rm SH}}, \theta{\bf 1}_{\{u^{b,c}< \widehat b_j\}}= (k^{b,c}+k^g){\bf 1}_{\{u^{b,c} \geq \widehat b_j\}}= 0 \bigg\} .
\]
Then, for any $u^{b,c}<\widehat b_j$, the diffusion equation in \eqref{DPequation-valuefunction-upperboundary} reduces to the ODE
\begin{equation} \label{investorODE}
0   =   \widehat{\mathcal V}_j^\prime(u^{b,c}) \big( \big(r+\widehat\lambda_j^{\rm SH}\big)u^{b,c}-B j \big)  - \widehat{\mathcal V}_j(u^{b,c})\widehat\lambda^{k^g}_j + \mu j ,
\end{equation}
with the boundary condition $\widehat{\mathcal V}_j \big(\frac{Bj}{r+\widehat \lambda_j^{\rm SH}} \big)=\frac{\mu j}{\widehat \lambda_j^{\rm SH}}$. If $u^{b,c}<x_j^\star$, we get that 
\[
\widehat{\mathcal V}_j(u^{b,c}) = \frac{\mu j}{\widehat\lambda_j^{\rm SH}} + C_1 \left( \left(\frac{r+\widehat\lambda_j^{\rm SH}}{\widehat\lambda_j^{\rm SH}}\right) u^{b,c} - \frac{B j}{ \widehat\lambda_j^{\rm SH}} \right)^\frac{\widehat\lambda_j^{\rm SH}}{ r + \widehat\lambda_j^{\rm SH} },
\]
for some $C_1\in\mathbb{R}$. If $u^{b,c} \in \left[ x_j^\star,\widehat b_j \right)$, equation \eqref{investorODE} is solved by
\[ 
\widehat{\mathcal V}_j(u^{b,c}) = \frac{\mu j}{\widehat\lambda_j^0} + C_2 \left( \left( \frac{r+\widehat\lambda_j^{SH}}{\widehat\lambda_j^0}\right) u^{b,c} - \frac{B j}{ \widehat\lambda_j^0} \right)^\frac{\widehat\lambda_j^0}{ r + \widehat\lambda_j^{SH} },
\]
for some $C_2\in\mathbb{R}$. The values of $C_1$ and $C_2$ for which the solution of equation \eqref{investorODE} is continuous are
\[
C_1 = \frac {  \frac{\mu j}{\widehat\lambda_j^0} - \frac{\mu j}{\widehat\lambda_j^{\rm SH}} +   \big( \frac{\rho_b}{\rho_g}  \big)^\frac{\widehat\lambda_j^0}{ r + \widehat\lambda_j^0 } \big( v_j^b(\widehat b_j) - \frac{\mu j}{\widehat\lambda_j^0}  \big) }{ \big( \frac{\rho_b}{\rho_g} \big)^{ \frac{\widehat\lambda_j^{\rm SH}}{ r + \widehat\lambda_j^0 } } \big(  \frac{ \widehat b_j (r+\widehat\lambda_j^0)}{\widehat\lambda_j^{\rm SH} } \big)^{ \frac{\widehat\lambda_j^{\rm SH}}{ r + \widehat\lambda_j^{\rm SH} } } },~ 
C_2 = \bigg( v_j^b(\widehat b_j) - \frac{\mu j}{\widehat\lambda_j^0} \bigg) \bigg( \widehat b_j\frac{r+\widehat\lambda_j^0}{\widehat\lambda_j^0} \bigg)^{-\frac{\widehat \lambda_j^0}{r+\widehat \lambda_j^{\rm SH}}}.
\]
It follows from the properties of the map $v_j^b$, that the resulting function $\widehat{\mathcal V}_j$ is a concave map with slope greater than $-1/\rho_b$ and therefore the family $\{\widehat{\mathcal V}_j\}_{1\leq j\leq I}$ is a solution of the system of HJB equations \eqref{DPequation-valuefunction-upperboundary}. It can be proved similarly as in the proof of Theorem \ref{verification theorem} (see also Theorem 3.15 in \cite{pages2014mathematical}), that the verification result holds for this family of functions. We therefore omit the proof of this result.
\end{proof}

\vspace{0.5em}
We continue with the

\begin{proof}[Proof of Proposition \ref{prop:value function lower boundary-b}]
By definition we have the set equality $\widehat{\mathcal A}^{g}(t,\widehat{\mathfrak L}_{I-N_t}(u^b),u^b)=\widehat{\mathcal A}^{b}(t,\widehat{\mathfrak L}_{I-N_t}(u^b),u^b)$. From Lemmas \ref{lemma:contracts lb1} and \ref{lemma:contracts lb2} we know that for every $\Psi_b\in\widehat{\mathcal A}^{b}(t,\widehat{\mathfrak L}_{I-N_t}(u^b),u^b)$, both agents always shirk under $\Psi_b$, therefore the objective functions in the definitions of $V^{\mathfrak L,g}_t(u^b)$ and $V^{\mathfrak L,b}_t(u^b)$ are also the same and equality holds.
\end{proof}

\vspace{1em}
We go on with the 

\vspace{0.5em}

\begin{proof}[Proof of Proposition \ref{prop:value function ub-b}]
The proof is identical to the proof of Proposition \ref{prop:value function ub}, with the only difference that since the principal is hiring the bad agent, for $u^b<\widehat b_j$ the ODE associated to the value function is
\begin{equation*} 
0   =   \widehat{\mathcal V}_j^\prime(u^b) \big( \big(r+\widehat\lambda_j^{\rm SH}\big)u^b-B j \big)  - \widehat{\mathcal V}_j(u^b)\widehat\lambda_j^{\rm SH} + \mu j ,
\end{equation*}
with the boundary condition $\widehat{\mathcal V}_j \big(\frac{Bj}{r+\widehat \lambda_j^{\rm SH}} \big)=\frac{\mu j}{\widehat \lambda_j^{\rm SH}}$.
\end{proof}

\vspace{0,5em}
We end this section with the

\begin{proof}[Proof of Proposition \ref{prop:optimal contracts lower boundary}]
The payments and the value of $\theta^\star$ in the case $u^b\geq C(I-N_t)$ are a direct consequence of the proof of Proposition \ref{prop:value function lb1}. From the proof of Proposition \ref{prop:value function lb2} we have that if $u^b < C(I-N_t)$ then 
\[
\theta_s^\star = {\bf 1}_{\Big\{ s - t > \frac{1}{r} \ln\Big( \frac{\nu(u^b)B\widehat \lambda_{I-N_t}^{\rm SH}}{\mu(r+\widehat \lambda_{I-N_t}^{\rm SH})} \Big) \Big\}}, 
\]
where $\nu(u^b)$ the solution of the associated dual problem. Since the quantity inside of the logarithm decreases with time, we have that $\theta^\star$ is a process which starts at zero, jumps to one at some instant and keeps constant afterwards. This means that if $\theta^\star$ jumps to one at some time $s$ and the project is still running, necessarily the continuation utility of the bad agent is equal to $C(I-N_s)$ because the project will continue until the last default. 
\end{proof}

\section{Extensions of the model}
\subsection{Endogenous reservation utility}

\begin{proof}[Proof of Proposition \ref{prop:reservationutilities}]
Define the dynamic version of $R_0^i$ by
\[
R^i_t := \sup_{k\in\mathfrak K} \mathbb E^{\P^k}\bigg[\int_t^{\tau^I} \mathrm{e}^{-r(s-t)}(\rho_i\mu(I-N_s)+Bk_s) \mathrm{d}s \bigg| \Fc_t \bigg].
\]
Observe that, thanks to our earlier results, we know that the previous expression depends on $t$ only through the value of $I-N_t$. Call then $\widehat{R}_{I-N_t}^i=R^i_t$, the value when there are $I-N_t$ loans left. The explicit value of $\widehat{R}^i_1$ and the optimal action of the agent of type $\rho_i$ in $(\tau^{I-1},\tau)$ were obtained in the study of short--term contracts with constant payments, in section \ref{sec:appcons}. Suppose now that $j>1$ and that the value of $\widehat{R}^i_{j-1}$ as well as the optimal action of the agent after default $\tau^{I-j}$ are known.

\vspace{0.5em} 
If the agent decides to monitor all the loans in $(\tau^{I-j},\tau^{I-j+1})$, his expected utility will be given by
\[
u^i(0):= \frac{\rho_i \mu j}{r+\widehat\lambda_j^0} + \frac{\widehat\lambda_j^0}{r+\lambda_j^0} \widehat{R}^i_{j-1}.
\]
The process $h^{1,i}(0)$ associated to this action is given by
\[
h^{1,i}(0):= u^i(0) - \widehat{R}^i_{j-1} = \frac{\rho_i \mu j}{r+\widehat\lambda_j^0} - \frac{r}{r+\widehat\lambda_j^0} \widehat{R}^i_{j-1}.
\]
Therefore, it is incentive compatible to monitor all the loans in $(\tau^{I-1},\tau)$ if and only if 
\[
h^{1,i}(0) \geq b_j ~ \iff ~ \rho_i \mu j - r \widehat{R}^i_{j-1} \geq b_j(r+\widehat\lambda_j^0).
\]
Similarly, if the agent chooses to shirk in $(\tau^{I-j},\tau^{I-j+1})$, his expected utility will be equal to
\[
u^i(j):= \frac{\rho_i \mu j + B j}{r+\widehat\lambda_j^{\rm SH}} + \frac{\widehat\lambda_j^{SH}}{r+\widehat\lambda_j^{\rm SH}} \widehat{R}^i_{j-1}.
\]
The process $h^{1,i}(0)$ associated to this action is given by
\[
h^{1,i}(j):= u^i(j) - \widehat{R}^i_{j-1} = \frac{\rho_i \mu j + B j}{r+\widehat\lambda_j^{\rm SH}} - \frac{r}{r+\widehat\lambda_j^{\rm SH}} \widehat{R}^i_{j-1},
\]
and it is incentive compatible to not monitor any loan in $(\tau^{I-1},\tau)$ if and only if 
\begin{align*}
h^{1,i}(j) < b_j  \iff &  \rho_i \mu j + B j - r \widehat{R}^i_{j-1} < b_j(r+\widehat\lambda_j^{\rm SH})
\iff  \rho_i \mu j - r \widehat{R}^i_{j-1} < b_j(r+\widehat\lambda_j^0).
\end{align*}

\end{proof}

\subsection{Unbounded relationship between utilities of the banks}\label{sec:extension}

A possible extension of our model could rely on a further differentiation between the work of the two banks, {\it i.e.} when both banks work, the good one would be more efficient in the sense that the associated default intensity is strictly smaller than that of the bad bank. We can do this by introducing an extra type variable with values $m_g$ and $m_b$, with $m_g<m_b$ and modelling the hazard rate of a non-defaulted loan $j$ at time $t$, when it is monitored by a bank of type $i$ as $\alpha_t^{j,i}=\alpha_{I-N_t}(1+e_t^{j,i}m_i+(1-e_t^{j,i})\varepsilon).$ Then, if the banks fails to monitor $k$ loans, the default intensity will be $$\lambda_t^{k,i}=\alpha_{I-N_t}((I-N_t)(1+m_i)+(\varepsilon-m_i)k_t).$$ We did not consider such a situation because it creates a degeneracy, in the sense that the credible set no longer has an upper boundary. Indeed, consider for simplicity the case $j=1$ and take any $u_0^b\geq b_1^j$, $t^{\star}\geq0$ and choose the corresponding payment 
$$
c(t^{\star}):=u_0^b\frac{e^{(r+\widehat\lambda_1^{0,b})t^{\star}}(r+\widehat\lambda_1^{0,b})}{\rho_b}\geq\frac{b_1^b(r+\widehat\lambda_1^{0,b})}{\rho_b}\geq\frac{b_1^g(r+\widehat\lambda_1^{0,g})}{\rho_g}.
$$  

Then, under the contract with delay and constant payments given by $dD_s=c(t^{\star}) 1_{\{s>t^{\star}\}}ds$ the bad bank will always work and her value function will be equal to $u_0^b$ (see section \ref{contractswithdelay}). Notice that the optimal strategy for the good bank will be also to work at every time. Then, her value function is equal to 
$$
u_0^g:=u_0^b\frac{\rho_g(r+\widehat\lambda_1^{0,b})}{\rho_b(r+\widehat\lambda_1^{0,g})}e^{(\widehat\lambda_1^{0,b}-\widehat\lambda_1^{0,g})t^{\star}}.
$$ 
We see that by increasing $t^{\star}$, it is possible to make $u_0^g$ as big as we want and keep fixed the value of the bad bank. This means that the credible set will have no upper boundary in the interval $[b_1^b,\infty)$. Moving to any $j>1$ and considering short-term contracts with delay, with $\theta=0$ and the analogous payments, we observe the same degeneracy and the credible set will have no upper boundary in the interval $[b_j^b,\infty)$.

\vspace{0.5em}
One way out of this problem would be to consider different discount rates for the banks, $r_b$ and $r_g$, and assume that the default intensities are such that $\lambda_t^{0,b}+r_b\leq \lambda_t^{0,g}+r_g$. However, this complicates things a lot because simple statements that we expect to be true are very difficult to prove or need assumptions on the parameters of the problem. For example the inequality $U_t^g(D,\theta) \geq U_t^b(D,\theta)$ is no longer clear at all. We therefore refrained from going into that direction, and leave it for potential future research.


\end{document}